\newif\ifabstract
\abstracttrue
\newif\iffull
\ifabstract \fullfalse \else \fulltrue \fi

\documentclass[11pt]{article}
\usepackage{amsfonts}
\usepackage{amssymb}
\usepackage{amstext}
\usepackage{amsmath}
\usepackage{xspace}
\usepackage{theorem}
\usepackage{graphicx}
\usepackage{graphics}
\usepackage{colordvi}
\usepackage{colordvi}

\usepackage[small,compact]{titlesec}

\usepackage{wrapfig}

\advance \oddsidemargin by -0.8in

        {\hspace*{\fill}$\Box$\par\vspace{4mm}}

\newcounter{definition}
\newenvironment{definition}{\refstepcounter{definition}\par\medskip\noindent%
   \textbf{Definition~\thedefinition.} \rmfamily}{\medskip}

\newcommand{\qed}{\hfill\vbox{\hrule height.2pt\hbox{\vrule width.2pt height5pt \kern5pt
\vrule width.2pt} \hrule height.2pt}}

\newcommand{\connect}{\leadsto}

\newcommand{\sconnect}{\overset{\mbox{\tiny{1:1}}}{\leadsto}}

\newcommand{\SC}{\mathsf{SC}}

\newcommand{\alphawl}{\ensuremath{\alpha_{\mbox{\tiny{\sc WL}}}}}
\newcommand{\alphaWL}{\alphawl}
\newcommand{\algsc}{\ensuremath{{\mathcal{A}}_{\mbox{\textup{\scriptsize{ARV}}}}}\xspace}

\newcommand{\alphasc}{\ensuremath{\alpha_{\mbox{\tiny{\sc ARV}}}}}
\newcommand{\gkrv}{\ensuremath{\gamma_{\mbox{\tiny{\sc CMG}}}}}
\newcommand{\gKRV}{\gkrv}
\newcommand{\gcmg}{\ensuremath{\gamma_{\mbox{\tiny{\sc CMG}}}}}
\newcommand{\alphaCMG}{\ensuremath{\alpha_{\mbox{\tiny{\sc CMG}}}}}

\newcommand{\floor}[1]{\ensuremath{\left\lfloor#1\right\rfloor}}

\newcommand{\event}{{\cal{E}}}

\newcommand{\NP}{\mbox{\sf NP}\xspace}

\newcommand{\polylog}[1]{\mathrm{polylog(#1)}}

\newcommand{\ZPTIME}{\mbox{\sf ZPTIME}}

\newcommand{\opt}{\mathsf{OPT}}

\newcommand{\set}[1]{\left\{ #1 \right\}}
\newcommand{\sse}{\subseteq}

\newcommand{\partition}{\mathsf{PARTITION}}
\newcommand{\separate}{\mathsf{SEPARATE}}

\newcommand{\tset}{{\mathcal T}}

\newcommand{\ttset}{\tilde{\mathcal T}}
\newcommand{\tT}{\tilde{T}}
\newcommand{\tmset}{\tilde{\mathcal M}}

\newcommand{\iset}{{\mathcal{I}}}
\newcommand{\pset}{{\mathcal{P}}}
\newcommand{\tpset}{\tilde{\mathcal{P}}}

\newcommand{\qset}{{\mathcal{Q}}}

\newcommand{\cset}{{\mathcal{C}}}
\newcommand{\fset}{{\mathcal{F}}}
\newcommand{\mset}{{\mathcal M}}

\newcommand{\wset}{{\mathcal{W}}}

\newcommand{\rset}{{\mathcal{R}}}

\newcommand{\hset}{{\mathcal{H}}}

\newcommand{\sset}{{\mathcal{{S}}}}
\newcommand{\gset}{{\mathcal{G}}}

\newcommand{\nots}{\overline S}

\newcommand{\be}{\begin{enumerate}}
\newcommand{\ee}{\end{enumerate}}
\newcommand{\bd}{\begin{description}}
\newcommand{\ed}{\end{description}}
\newcommand{\bi}{\begin{itemize}}
\newcommand{\ei}{\end{itemize}}

\newtheorem{lemma}{Lemma}
\newtheorem{theorem}{Theorem}

\newtheorem{corollary}{Corollary}
\newtheorem{claim}{Claim}

\newenvironment{proof}{\smallskip\noindent{\bf Proof:}}{\hfill\stopproof}
\def\stopproof{\square}
\def\square{\vbox{\hrule height.2pt\hbox{\vrule width.2pt height5pt \kern5pt
\vrule width.2pt} \hrule height.2pt}}

\newcommand{\tk}{\tilde k}

\renewcommand{\phi}{\varphi}
\newcommand{\eps}{\epsilon}

\newcommand{\half}{\ensuremath{\frac{1}{2}}}

\newcommand{\poly}{\operatorname{poly}}

\newcommand{\prob}[2][]{\text{\bf Pr}_{#1}\left [#2\right]}

\newenvironment{properties}[2][0]
{
\begin{enumerate} \setcounter{enumi}{#1}}{\end{enumerate}}

\newcommand{\EDP}{\mbox{\sf EDP}\xspace}
\newcommand{\EDPwC}{\mbox{\sf EDPwC}\xspace}

\newcommand{\out}{\operatorname{out}}

\begin{document}
\title{A Polylogarithimic Approximation Algorithm for Edge-Disjoint Paths with Congestion $2$}

\author{
Julia Chuzhoy \thanks{Toyota Technological Institute, Chicago, IL
60637. Email: {\tt cjulia@ttic.edu}. Supported in part by NSF CAREER grant CCF-0844872 and Sloan Research Fellowship.}
\and
Shi Li\thanks{Center for Computational Intractability,
Department of Computer Science, Princeton University. Email: {\tt shili@cs.princeton.edu}.
Supported by NSF awards MSPA-MCS 0528414, CCF 0832797,  AF 0916218 and CCF-0844872.}
}

\begin{titlepage}

\maketitle
\thispagestyle{empty}
\begin{abstract}
In the Edge-Disjoint Paths with Congestion problem (\EDPwC), we are given an undirected $n$-vertex graph $G$, a collection $\mset=\set{(s_1,t_1),\ldots,(s_k,t_k)}$ of demand pairs and an integer $c$. The goal is to connect the maximum possible number of the demand pairs by paths, so that the maximum edge congestion - the number of paths sharing any edge - is bounded by $c$. When the maximum allowed congestion is $c=1$, this is the classical Edge-Disjoint Paths problem (\EDP).

The best current approximation algorithm for \EDP achieves an $O(\sqrt n)$-approximation, by rounding 
 the standard multi-commodity flow relaxation of the problem. This matches the $\Omega(\sqrt n)$ lower bound  on the integrality gap of this relaxation. We show an $O(\poly\log k)$-approximation algorithm for \EDPwC with congestion $c=2$, by rounding the same multi-commodity flow relaxation. This gives the best possible congestion for a sub-polynomial approximation of \EDPwC via this relaxation. Our results are also close to optimal in terms of the number of pairs routed, since \EDPwC is known to be hard to approximate to within a factor of $\tilde{\Omega}\left((\log n)^{1/(c+1)}\right )$ for any constant congestion $c$. Prior to our work, the best approximation factor for \EDPwC with congestion $2$ was $\tilde O(n^{3/7})$, and the best algorithm achieving a polylogarithmic approximation required congestion $14$.
\end{abstract}
\end{titlepage}

\section{Introduction}
One of the central and most extensively studied graph routing problems is the Edge-Disjoint Paths problem (\EDP). In this problem, we are given an undirected $n$-vertex graph $G=(V,E)$, and a collection $\mset=\set{(s_1,t_1),\ldots, (s_k,t_k)}$ of $k$ source-sink pairs, that we also call demand pairs. The goal is to find a collection $\pset$ of edge-disjoint paths, connecting the maximum possible number of the demand pairs.

Robertson and Seymour~\cite{RobertsonS} have shown that \EDP can be solved efficiently, when the number $k$ of the demand pairs is bounded by a constant. However, for general values of $k$, it is NP-hard to even decide whether all pairs can be simultaneously routed via edge-disjoint paths~\cite{Karp}. A standard approach to designing approximation algorithms for \EDP and other routing problems, is to first compute a multi-commodity flow relaxation, where instead of connecting the demand pairs with paths, we are only required to send the maximum amount of multi-commodity flow between the demand pairs, with at most one flow unit sent between every pair. Such a fractional solution can be computed efficiently by using the standard multi-commodity flow LP-relaxation, and it can then be rounded to obtain an integral solution. Indeed, the best current approximation algorithm for the \EDP problem, due to Chekuri, Khanna and Shepherd~\cite{EDP-alg}, achieves an $O(\sqrt n)$-approximation using this approach. Unfortunately, a simple example by Garg, Vazirani and Yannakakis~\cite{trees2}  (see also Section~\ref{sec: gap} in the Appendix), 
shows that the integrality gap of the multi-commodity flow relaxation can be as large as $\Omega(\sqrt n)$, thus implying that the algorithm of~\cite{EDP-alg} is essentially the best possible for \EDP, when using this approach.
This integrality gap appears to be a major barrier to obtaining better approximation algorithms for \EDP. Indeed, we do not know how to design better approximation algorithms even for some seemingly simple special cases of planar graphs, called the brick-wall graphs (see section~\ref{sec: gap} of the Appendix). With the current best hardness of approximation factor standing on $\Omega(\log^{1/2-\epsilon}n)$ for any constant $\eps$ (unless $\NP$ is contained in $\ZPTIME(n^{\poly \log n})$~\cite{AZ-undir-EDP,ACGKTZ}), the approximability of the \EDP problem remains one of the central open problems in the area of routing.

A natural question is whether we can obtain better approximation algorithms by slightly relaxing the disjointness requirement, and allowing the paths to share edges. We say that a set $\pset$ of paths is an $\alpha$-approximate solution with congestion $c$, iff the paths in $\pset$ connect at least $\opt/\alpha$ of the demand pairs, while every edge of $G$ appears on at most $c$ paths in $\pset$. Here, $\opt$ is the value of the optimal solution to \EDP, where no congestion is allowed. This relaxation of the \EDP problem is called \EDP with congestion  (\EDPwC). The \EDPwC problem is a natural framework to study the tradeoff between the number of pairs routed and the congestion, and it is useful in scenarios where we can afford a small congestion on edges.

The classical randomized rounding technique of Raghavan and Thompson~\cite{RaghavanT} gives a constant factor approximation for \EDPwC, when the  congestion $c$ is $\Omega(\log n/\log\log n)$. More generally, for any congestion value $c$, factor $O(n^{1/c})$-approximation algorithms are known for \EDPwC~\cite{AzarR, BavejaS, KolliopoulosS}. Recently, Andrews~\cite{Andrews} has shown a randomized $O(\poly\log n)$-approximation algorithm  with congestion $c=O(\poly\log\log n)$, and Chuzhoy~\cite{EDP-old} has shown a randomized $O(\poly\log k)$-approximation algorithm with congestion $14$. For the congestion value $c=2$, Kawarabayashi and Kobayashi~\cite{KawarabayashiK} have recently shown an $\tilde O(n^{3/7})$-approximation algorithm, thus improving the best previously known $O(\sqrt n)$-approximation for $c=2$~\cite{AzarR, BavejaS, KolliopoulosS}. We note that all the above mentioned algorithms rely on the standard multi-commodity flow LP relaxation of the problem. It is easy to see that the values of the optimal solution of this LP relaxation for the \EDP problem, where no congestion is allowed, and for the \EDPwC problem, where congestion $c$ is allowed, are within a factor $c$ from each other. Therefore, the statements of these results remain valid even when the approximation factor is computed with respect to the optimal solution to the \EDPwC problem.

In this paper, we show a randomized $O(\poly\log k)$-approximation algorithm for \EDPwC with congestion $2$. Given  an instance $(G,\mset)$ of the \EDP problem, our algorithm w.h.p. routes at least $\Omega(\opt/\poly\log k)$ pairs with congestion $2$, where $\opt$ is the maximum number of pairs that can be routed with no congestion. Our algorithm also achieves an $O(\poly\log k)$-approximation when compared with the optimal solution to \EDPwC with congestion $2$. As all the algorithms for \EDP and \EDPwC mentioned above, our algorithm also performs a rounding of the standard multi-commodity flow relaxation for \EDP. Therefore, our result shows that when congestion $2$ is allowed, the integrality gap of this relaxation improves from $\Omega(\sqrt n)$ to polylogarithmic. 
Our result is essentially optimal with respect to this relaxation, both for the congestion and the number of pairs routed, in the following sense. As observed above, if we are interested in obtaining a sub-polynomial approximation for \EDP via the multi-commodity flow relaxation, then the best congestion we can hope for is $2$. On the other hand, Andrews et al.~\cite{ACGKTZ} have shown that the integrality gap of the multi-commodity flow relaxation for \EDPwC  is $\Omega\left (\left(\frac{\log n}{(\log\log n)^2}\right )^{1/(c+1)}\right )$ for any constant congestion $c$. In particular, the integrality gap for congestion $2$ is polylogarithmic, though the degree of the logarithm is much lower than the degree we obtain in our approximation algorithm. Andrews et al.~\cite{ACGKTZ} have  also shown that for any constant $\eps$, for any $1\leq c\leq O\left(\frac{\log \log n}{\log\log\log n}\right )$, there is no $O\left ((\log n)^{\frac{1-\eps}{c+1}}\right)$-approximation algorithm for \EDPwC with congestion $c$, unless $\NP \subseteq \ZPTIME(n^{\poly \log n})$. In particular, this gives an $\Omega\left (\log^{(1-\eps)/3}n\right )$-hardness of approximation for \EDPwC with congestion $2$.

While the approximability status of the \EDP problem remains open, our results show a fundamental difference between routing with congestion $1$ and routing with congestion 2 or higher. Suppose we are given a solution $\pset$ to the \EDP problem that connects $D$ of the demand pairs with congestion $c$, and we are interested in obtaining another solution with a lower congestion. Our results provide an efficient randomized algorithm to find a solution connecting $\Omega\left(D/(c\poly\log k)\right )$ of the demand pairs with congestion $2$. That is, we can lower the congestion to $2$ with only a factor $(c\poly\log k)$ loss in the number of the demand pairs routed\footnote{We can view the set $\pset$ of paths as a fractional solution to the \EDP problem instance, where $1/c$ flow units are sent along each path. This gives a fractional solution of value $D/c$. We can then use our algorithm to route $\Omega\left(D/(c\poly\log k)\right )$ of the demand pairs with congestion $2$, by rounding this fractional solution.}. However, if we are interested in routing with no congestion, then we may have to lose an $\Omega(\sqrt n)$-factor in the number of pairs routed. For example, in the integrality gap construction of Garg, Vazirani and Yannakakis~\cite{trees2}, there is a solution that routes $k=\Theta(\sqrt n)$ pairs with congestion $2$, but if we require a routing with congestion $1$, then at most one pair can be routed (see section~\ref{sec: gap} of the Appendix).

We note that better approximation algorithms are known for some special cases of the \EDP problem. Rao and Zhou~\cite{RaoZhou} have shown that if the value of the global minimum cut in the input graph $G$ is $\Omega(\log^5n)$, then there is an efficient randomized $O(\poly\log n)$-approximation algorithm for \EDP. Interestingly, this algorithm is also based on the multi-commodity flow relaxation. The \EDP problem is also known to have polylogarithmic approximation algorithms on bounded-degree expander graphs~\cite{LR,BFU,BFSU,KleinbergR,Frieze}, and constant-factor approximation algorithms on trees~\cite{trees2,trees1}, grids and grid-like graphs~\cite{grids1,grids2,grids3,grids4}.
Routing problems have also been extensively studied on planar graphs. Chekuri, Khanna and Shepherd~\cite{CKS-planar1,CKS-planar2} have shown a poly-logarithmic approximation algorithm for \EDPwC with congestion $2$ and a constant approximation algorithm with congestion $4$ on planar graphs. Both results have recently been improved by Seguin-Charbonneau and Shepherd~\cite{SS-planar}, who showed a constant factor approximation algorithm with congestion $2$. When no congestion is allowed, Kleinberg~\cite{Kleinberg-planar} has shown an $O(\log^2n)$-approximation for Eulerian planar graphs. Kawarabayashi and Kobayashi~\cite{Kplanar} have recently improved this result to an $O(\log n)$-approximation, for both Eulerian and $4$-connected planar graphs. However, improving the $O(\sqrt n)$-approximation algorithm for \EDP on general planar graphs still remains elusive.

A problem closely related to \EDP is {\sf{Congestion Minimization}}, where the goal is to route all the demand pairs, while minimizing the maximum congestion. The classical randomized rounding technique of Raghavan and Thompson~\cite{RaghavanT}, when applied to the multi-commodity flow relaxation of the problem,  achieves an $O(\log n/\log\log n)$-approximation. This is the best currently known approximation algorithm for the problem. The best current hardness of approximation ratio, due to Andrews and Zhang~\cite{AZ-undir-cong}, is $\Omega\left(\frac{\log\log n}{\log\log \log n}\right )$, under the assumption that $\NP \not\subseteq \ZPTIME(n^{\poly \log n})$.

{\bf Our results.}
Our main result is summarized in the following theorem.


\begin{theorem}\label{thm: main}
There is an efficient randomized algorithm, that, given a graph $G$, and a collection $\mset$ of $k$ source-sink pairs, w.h.p. finds a routing of $\Omega(\opt/(\poly\log k))$ of the pairs in $\mset$ with congestion at most $2$, where $\opt$ is the maximum number of pairs that can be routed with congestion 2.
\end{theorem}

We now provide an overview of our  techniques and compare them to~\cite{EDP-old}. We denote by $\tset$ the set of vertices participating in the demand pairs, and we call them \emph{terminals}. 
As in previous work~\cite{CKS,RaoZhou,Andrews,EDP-old}, we use the notion of well-linkedness. Given a graph $G=(V,E)$ and a subset $\tset$ of vertices called terminals, we say that $G$ is $\alpha$-well linked for the terminals, iff for any partition $(A,B)$ of $V$, 
$|E(A,B)|\geq \alpha\cdot \min\set {|A\cap \tset|,|B\cap \tset|}$. Chekuri, Khanna and Shepherd~\cite{ANF,CKS} have shown an efficient algorithm, that, given any \EDP instance $(G,\mset)$, partitions it into a number of sub-instances $(G_1,\mset_1),\ldots,(G_{\ell},\mset_{\ell})$, such that, on the one hand, each instance $G_i$ is $1$-well-linked for the set of terminals participating in $\mset_i$, and on the other hand, the sum of the values of the optimal fractional solutions in all these instances is $\Omega(\opt/\log^2k)$. Therefore, it is enough to find a polylogarithmic approximation with congestion $2$ in each such sub-instance separately. From now on we assume that we are given an instance $(G,\mset)$, where $G$ is $1$-well-linked for the set $\tset$ of terminals.

Chekuri, Khanna and Shepherd~\cite{ANF,CKS,CKS-planar1} have suggested the following high-level approach to solving \EDP instances $(G,\mset)$, where $G$ is well-linked for the terminals. They start by defining a graph called a \emph{crossbar}: a graph $H$ with a subset $Y\sse V(H)$ of vertices is called a crossbar with congestion $c$, iff any matching over the vertices of $Y$ can be routed with congestion at most $c$ in $H$. They then note that if we could show an algorithm that finds a crossbar $(H,Y)$ in graph $G$, with $|Y|=k/\poly\log k$, and constant congestion, then we can obtain a polylogarithmic approximation to \EDPwC with constant congestion. 
An algorithm for constructing such a crossbar with a constant congestion follows from the recent work of~\cite{EDP-old}.

We follow this approach, and define a structure that we call a \emph{good crossbar}, which gives slightly stronger properties than the general crossbar defined above. For any subset $S\sse V$ of vertices, let $\out(S)$ denote the set of edges with one endpoint in $S$ and one endpoint in $V\setminus S$. Informally, we say that a vertex set $S$ is $\alpha$-well-linked for a subset $\Gamma\sse \out(S)$ of edges, iff the graph $G[S]$ is $\alpha$-well-linked for the set $\Gamma$ of terminals. Formally, we require that for any partition $(A,B)$ of $S$, $|E(A,B)|\geq \alpha\cdot \min\set{|\out(A)\cap \Gamma|,|\out(B)\cap \Gamma|}$.

A good crossbar consists of three parts. The first part is a \emph{good family of vertex subsets} $\fset=\set{S_1,\ldots, S_{\gamma}}$, where $\gamma=O(\log^2n)$. The sets $S_1,\ldots,S_{\gamma}$ are all vertex-disjoint, and only contain non-terminal vertices. Each set $S_j$ is associated with a subset $\Gamma^*_j\sse \out(S_j)$ of $k^*=k/\poly\log k$ edges, such that $S_j$ is $1$-well-linked for $\Gamma^*_j$. 
The second part of the good crossbar is a collection $\mset^*\sse \mset$ of $k^*/2$ demand pairs, and the third part is a collection $\tau={T_1,\ldots,T_{k^*}}$ of $k^*$ trees. Let $\tset^*$ denote the set of all terminals participating in the demand pairs in $\mset^*$. Then for each $1\leq i\leq k^*$, the tree $T_i$ contains a distinct terminal $t_i\in \tset^*$, and a distinct edge $e_{i,j}\in \Gamma^*_j$, for all $1\leq j\leq \gamma$. In other words, for each set $S_j\in \fset$, we can view $\Gamma^*_j=\set{e_{1,j},\ldots,e_{k^*,j}}$, where $e_{i,j}\in E(T_i)$ for all $1\leq i\leq k^*$. (See Figure~\ref{fig:crossbar}.)

\begin{figure}[h]
\scalebox{0.6}{\includegraphics{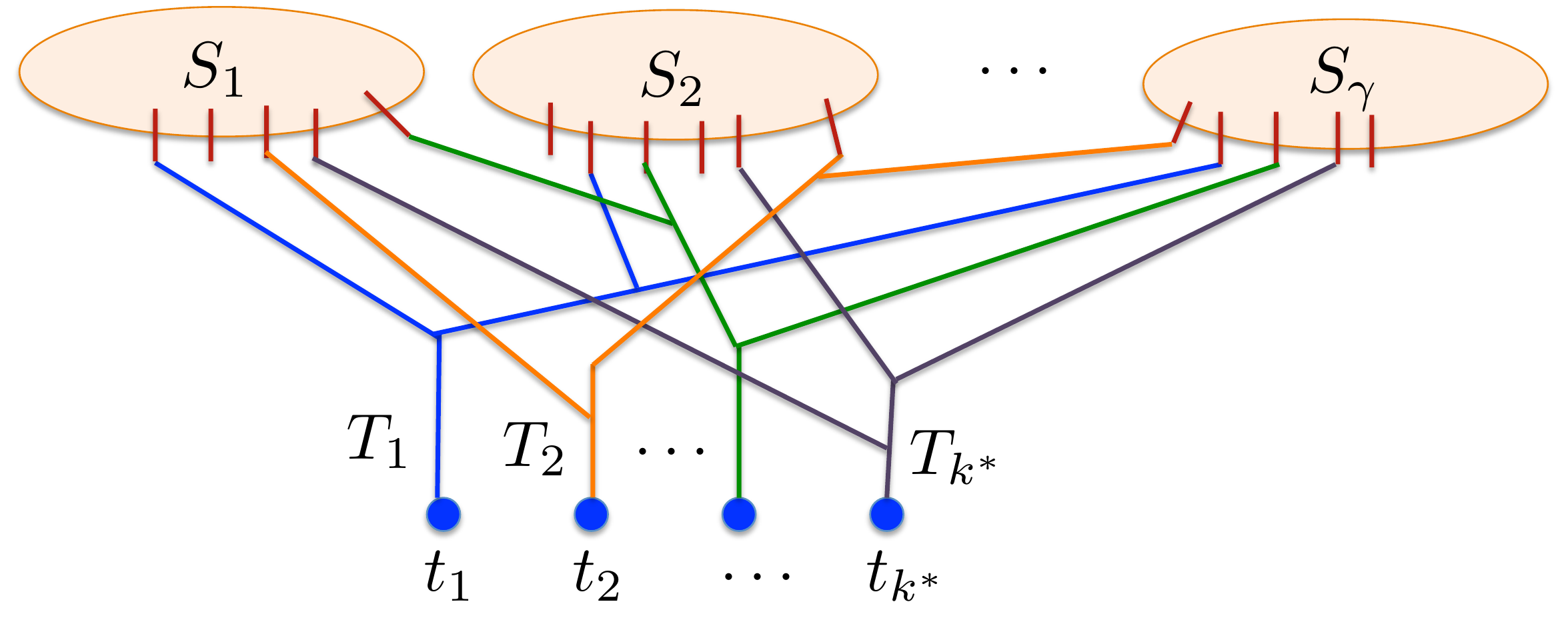}} \caption{A good crossbar}\label{fig:crossbar}
\end{figure}

Given a good crossbar $(\fset,\mset^*,\tau)$, we say that the congestion of this crossbar is $c$ iff every edge of $G$ appears in at most $c$ trees in $\tau$, and every edge of $G$ with both endpoints in the same set $S_j\in \fset$ appears in at most $c-1$ such trees. Chuzhoy~\cite{EDP-old} has implicitly defined a good crossbar, and has shown that, given a good crossbar that causes congestion $c$, there is an efficient randomized algorithm to route $\Omega(k^*/\poly\log k)$ demand pairs with congestion at most $c$ in graph $G$. This algorithm uses the cut-matching game of Khandekar, Rao and Vazirani~\cite{KRV} to embed an expander into $G$, and then finds a routing in this expander using the algorithm of Rao and Zhou~\cite{RaoZhou}.
She has also shown an efficient algorithm for constructing a good crossbar with congestion $14$, thus obtaining an $O(\poly\log k)$-approximation with congestion $14$ for \EDPwC.

We follow a similar approach here, except that we construct a good crossbar with congestion $2$. Combining this with the result of \cite{EDP-old}, we obtain an $O(\poly\log k)$-approximation to \EDPwC with congestion $2$. We now provide more details about the construction of the good crossbar of~\cite{EDP-old}, and the barriers that need to be overcome to reduce the congestion to $2$. The algorithm of~\cite{EDP-old} consists of three steps. In the first step, a good family $\fset=\set{S_1,\ldots,S_{\gamma}}$ of vertex subsets is constructed, and for each $1\leq j\leq \gamma$, a subset $\Gamma^*_j\sse \out(S_j)$ is selected. Additionally, for each set $S_j\in \fset$, there is a set $\pset_j$ of $k^*$ edge-disjoint paths in graph $G$, connecting every edge in $\Gamma^*_j$ to a distinct terminal in $\tset$.
In the second step, we construct a family $\tau'=\set{T_1',\ldots,T_{k^*}'}$ of trees, where every tree $T_i\in \tau'$ contains a distinct terminals $t_i\in \tset$, and a distinct edge $e_{i,j}\in \Gamma^*_j$ for each $1\leq j\leq \gamma$. However, the terminals in set $\tset'=\set{t_i\mid T_i\in \tau'}$ do not necessarily form source-sink pairs in $\mset$. In order to construct the trees in $\tau'$,  consider the graph $G'$ obtained from $G$ by contracting every set $S_j\in \fset$ into a super-node $v_j$. Then the problem of finding the set $\tau'$ of trees is similar to the problem of packing Steiner trees in $G'$. Moreover, the collections $\set{\pset_j}_{j=1}^{\gamma}$ of paths can be interpreted as a fractional solution of value $\Omega(k/\poly\log k)$ to this problem. We can then round this solution to obtain the family $\tau'$ of trees, using known algorithms for packing Steiner trees. The fact that we need to deal with the clusters $S_j$ instead of the super-nodes $v_j$ makes the problem more challenging technically, but due to the fact that the sets $S_j$ are well-linked for the sets $\Gamma^*_j$ of edges, the problem is still tractable. The well-linkedness of the sets $S_j\in \fset$ is exploited in this algorithm, by routing several matchings between the edges of $\Gamma^*_j$ across each such set $S_j$, in order to simulate the super-node $v_j$.
Finally, in the third step, we select a subset $\mset^*\sse \mset$ of $k^*/2$ demand pairs, and connect all terminals participating in the pairs in $\mset^*$ to the terminals in set $\tset'$. The union of these new paths with the trees in $\tau'$ gives the final collection $\tau$ of trees.

There are several factors contributing to the accumulation of congestion in this construction. We mention the main two barriers to reducing the congestion to $2$ here. The first problem is that, given the good family $\fset$ of vertex subsets, it is possible that for some set $S_j\in \fset$ the paths in $\pset_j$ may contain edges lying in other clusters $S_{j'}\in \fset$. Moreover, these paths do not necessarily enter and leave the set $S_{j'}$ through the edges in $\tset^*_{j'}$. 
The construction of the trees in $\tau'$ uses each set $S_{j'}\in \fset$ to route several matchings between the edges of $\Gamma^*_{j'}$. If the edges of $G[S_{j'}]$ additionally participate in the paths in $\set{\pset_j}_{j=1}^{\gamma}$, then this leads to accumulation of congestion. Indeed, an edge of graph $G$ may appear in up to $8$ trees of $\tau'$ in the construction of~\cite{EDP-old}. Ideally, we would like to construct a good family $\fset$ of vertex subsets, where for each set $S_j\in \fset$, the paths in $\pset_j$ do not contain any vertices lying in other subsets $S_{j'}\in \fset$. However, it is not clear whether it is possible to achieve such strong properties.

The second difficulty is that step 2 and step 3 are executed separately, each contributing to the total congestion. It appears that one has to incur a congestion of at least $2$ when constructing the set $\tau'$ of trees, using current techniques. If the terminals in set $\tset'$ do not form demand pairs in $\mset$, then we need to additionally select a subset $\mset^*\sse \mset$ of the demand pairs, and to route the terminals participating in $\mset^*$ to the terminals of $\tset'$, thus increasing the congestion beyond $2$.

In order to find a good family $\fset$ of vertex subsets, the algorithm of~\cite{EDP-old} performs a number of iterations. In each iteration we start with what is called a legal contracted graph $G'$. This graph is associated with a collection $\cset$ of disjoint subsets of non-terminal vertices of $G$, such that each set $C\in \cset$ is well-linked for $\out(C)$, and $G'$ is obtained from $G$ by contracting every cluster $C\in \cset$ into a super-node. Additionally, we require that for each cluster $C\in \cset$, $|\out(C)|$ is small. We call such a clustering $\cset$ a \emph{good clustering}. At the beginning of the algorithm, $\cset=\emptyset$ and $G'=G$. In every iteration, given a legal contracted graph $G'$, the algorithm either computes a good family $\fset$ of vertex subsets, or produces a new legal contracted graph $G''$, containing strictly fewer vertices than $G'$. This guarantees that after $n$ iterations, the algorithm produces a good family $\fset$ of vertex subsets.

In order to overcome the two problems mentioned above and avoid accumulating the congestion, we combine all three steps of the algorithm together. 
We define a potential function $\phi$ over collections $\cset$ of disjoint non-terminal vertex subsets, where $\phi(\cset)$ roughly measures the number of edges in the graph obtained from $G$ by contracting every cluster in $\cset$ into a super-node. The potential function $\phi$ has additional useful properties, that allow us to perform a number of standard operations (such as the well-linked decomposition) on the clusters of $\cset$, without increasing the potential value.

Our algorithm also consists of a number of iterations (that we call phases). In each such phase, we start with some legal contracted graph $G'$ and a corresponding good clustering $\cset'$. We then either construct a good crossbar, or produce a new good clustering $\cset''$ with $\phi(\cset'')<\phi(\cset')$, together with the corresponding new legal contracted graph $G''$. Each phase is executed as follows. We start with some clustering $\cset^*$ of the vertices of $G$, where $\phi(\cset^*)<\phi(\cset')$, but $\cset^*$ is not necessarily a good clustering. We then perform a number of iterations. In each iteration, we select a family $\fset=\set{S_1,\ldots,S_{\gamma}}$ of disjoint subsets of non-terminal vertices, that we treat as a potential good family vertex subsets, and we try to find a subset $\mset^*\sse \mset$ of $k^*$ demand pairs and a family $\tau$ of trees, to complete the construction of a good crossbar. If we do not succeed in constructing a good crossbar in the current iteration, then we use the family $\fset$ of vertex subsets to refine the current clustering $\cset^*$, such that the potential of the new clustering goes down by a significant amount. This ensures that after polynomially-many iterations, we will succeed in either constructing a good crossbar, or a good clustering $\cset^*$ with $\phi(\cset^*)<\phi(\cset')$.

This combination of all three steps of the algorithm of~\cite{EDP-old} appears necessary to overcome the two barriers described above. For example, it is possible that the family $\fset$ of vertex subsets is a good family, but we are still unable to extend it to a good crossbar (for example because of the problem of the paths in $\pset_j$ using the edges of $G[S_{j'}]$, as described above). Still, we will be able to make progress in such cases by refining the current clustering $\cset^*$. Similarly, we construct the trees and connect the terminals participating in pairs in $\mset^*$ to them simultaneously, to avoid accumulating congestion. Again, whenever we are unable to do so, we will  be able to refine the current clustering $\cset^*$.

{\bf Organization.} We start with preliminaries in Section~\ref{sec: Prelims} and provide an overview of our algorithm in Section~\ref{sec: alg overview}. We develop machinery to analyze vertex clusterings in Section~\ref{sec: clusterings}, and complete the algorithm description in Sections~\ref{sec: Alg} and~\ref{sec: proof of iteration theorem}. For convenience, a list of parameters is provided in Section~\ref{sec: appendix-params}, and the known lower bounds on the integrality gap of the multi-commodity flow relaxation for \EDP are provided in Section~\ref{sec: gap} of the Appendix.

\label{--------------------------------------------------sec: Preliminaries------------------------------------------------}
\section{Preliminaries}\label{sec: Prelims}


We assume that we are given an undirected  $n$-vertex graph $G=(V,E)$, and a set $\mset=\set{(s_1,t_1),\ldots,(s_k,t_k)}$ of $k$ source-sink pairs, that we also call demand pairs. We denote by $\tset$ the set of vertices that participate in pairs in $\mset$, and we call them \emph{terminals}. Let $\opt$ denote the maximum number of demand pairs that can be simultaneously routed via edge-disjoint paths. Given any subset $\mset'\sse \mset$ of the demand pairs, we denote by $\tset(\mset')\sse \tset$ the subset of terminals participating in the pairs in $\mset'$.

Using a standard transformation, we can assume w.l.o.g. that every terminal $t\in \tset$ participates in exactly one source-sink pair: otherwise, for each demand pair in which $t$ participates, we can add a new terminal to the graph, that will replace $t$ in the demand pair, and connect this new terminal to $t$. Similarly, using standard transformations, we can assume that the degree of every terminal is exactly $1$, and the degree of every non-terminal vertex is at most $4$. In order to achieve the latter property, we replace every vertex $v$ whose degree $d_v>4$ with a $d_v\times d_v$ grid, and connect the edges incident on $v$ to the vertices of the first row of the grid. It is easy to verify that these transformations do not affect the solution value. From now on we assume that the degree of every terminal is $1$, the degree of every non-terminal vertex is at most $4$, and every terminal participates in at most one demand pair.

For any subset $S\sse V$ of vertices, we denote by $\out_G(S)=E_G(S,V\setminus S)$, and by $E_G(S)$ the subset of edges with both endpoints in $S$, omitting the subscript $G$ when clear from context. Throughout the paper, we say that a random event succeeds w.h.p., if the probability of its success is $(1-1/\poly(n))$. All logarithms are to the base of $2$.

Let $\pset$ be any collection of paths in graph $G$. We say that paths in $\pset$ cause congestion $\eta$ in $G$, iff for every edge $e\in E$, at most $\eta$ paths in $\pset$ contain $e$.
Given any pair $(E_1,E_2)$ of subsets of edges, we denote by $F:E_1\connect_{\eta}E_2$ the flow where every edge in $E_1$ sends one flow unit to the edges in $E_2$, and the congestion due to the flow $F$ is at most $\eta$. If additionally every edge in $E_2$ receives at most one flow unit, then we denote this flow by $F:E_1\sconnect_{\eta} E_2$. Similarly, given a set $\pset$ of paths connecting the edges of $E_1$ to the edges of $E_2$, we denote $\pset:E_1\connect_{\eta} E_2$ iff $\pset=\set{P_e\mid e\in E_1}$, where $e$ is the first edge on $P_e$, and the total congestion caused by the paths in $\pset$ is at most $\eta$. If every edge of $E_2$ has at most one path terminating at it, then we denote $\pset:E_1\sconnect_{\eta}E_2$. We use a similar notation for flows and path sets connecting subsets of vertices to each other, or a subset vertices with a subset of edges. 

Given a subset $S$ of vertices and two subsets $E_1,E_2\sse \out(S)$ of edges, we say that the flow $F:E_1\connect_{\eta}E_2$ is \emph{contained in $S$} iff every flow-path is completely contained in $G[S]$, except for its first and last edges, that belong to $\out(S)$. Similarly, we say that a set $\pset:E_1\connect_{\eta}E_2$ of paths is contained in $S$ iff all inner edges on every path in $\pset$ belong to $G[S]$.



{\bf Sparsest Cut and the Flow-Cut Gap.} 
Suppose we are given a graph $G=(V,E)$, and a subset $\tset\sse V$ of $k$ terminals. 
The sparsity of a cut $(S,\nots)$ in $G$ is $\Phi(S)=\frac{|E(S,\nots)|}{\min\set{|S\cap \tset|, |\nots\cap \tset|}}$, and the value of the sparsest cut in $G$ is defined to be:
$\Phi(G)=\min_{S\subset V}\set{\Phi(S)}$.
The goal of the sparsest cut problem is, given an input graph $G$ and a set $\tset$ of terminals, to find a cut of minimum sparsity. Arora, Rao and Vazirani~\cite{ARV} have shown an $O(\sqrt {\log k})$-approximation algorithm for the sparsest cut problem. We denote this algorithm by \algsc, and its approximation factor by $\alphasc(k)=O(\sqrt{\log k})$.

A problem dual to sparsest cut is the maximum concurrent flow problem. For the above definition of the sparsest cut problem, the corresponding variation of the concurrent flow problem asks to find the maximum value $\lambda$, such that every pair of terminals can send $\lambda/k$ flow units to each other simultaneously with no congestion. The flow-cut gap is the maximum ratio, in any graph, between the value of the minimum sparsest cut and the maximum value $\lambda$ of concurrent flow. The value of the flow-cut gap in undirected graphs, that we denote by $\beta(k)$ throughout the paper, is $\Theta(\log k)$~\cite{LR, GVY,LLR,Aumann-Rabani}. Therefore, if $\Phi(G)=\alpha$, then every pair of terminals can send $\frac{\alpha}{k\beta(k)}$ flow units to each other with no congestion.
Equivalently, every pair of terminals can send $1/k$ flow units to each other with congestion at most $\beta(k)/\alpha$. Moreover, any matching on the set $\tset$ of terminals can be fractionally routed with congestion at most $2\beta(k)/\alpha$.

{\bf Well-Linkedness.}
The notion of well-linkedness has been widely used in graph decomposition and routing, see e.g.~~\cite{CKS,RaoZhou,Andrews}. While the main idea is similar, the definition details differ from paper to paper. Our definition of well-linkedness is similar to that of~\cite{EDP-old}.


\begin{definition}
Let $S$ be any subset of vertices of a graph $G$. For any integer $k_1$, for any $0<\alpha\leq 1$, we say that set $S$ is $(k_1,\alpha)$-well-linked iff for any pair $T_1,T_2\sse \out(S)$ of disjoint subsets of edges, with $|T_1|+|T_2|\leq k_1$, for any partition $(X,Y)$ of $S$ with $T_1\sse \out(X)$ and $T_2\sse \out(Y)$, $|E_G(X,Y)|\geq \alpha\cdot\min\set{|T_1|,|T_2|}$. \end{definition}

Suppose a set $S$ is not $(k_1,\alpha)$-well-linked. We say that a partition $(X,Y)$ of $S$ is a \emph{$(k_1,\alpha)$-violating partition}, iff there are two subsets $T_1\sse \out(X)\cap \out(S)$, $T_2\sse \out(Y)\cap \out(S)$ of edges with $|T_1|+|T_2|\leq k_1$, and $|E_G(X,Y)|<\alpha\cdot\min\set{|T_1|,|T_2|}$.

\begin{definition} Given a graph $G$, a subset $S$ of its vertices, a parameter $\alpha>0$, and a subset $\Gamma\sse \out(S)$ of edges, we say that $S$ is $\alpha$-well-linked for $\Gamma$, iff for any partition $(A,B)$ of $S$, $|E(A,B)|\geq \alpha\cdot \min\set{|\Gamma\cap \out(A)|,|\Gamma\cap \out(B)|}$. We say that the set $S$ is $\alpha$-well-linked iff it is $\alpha$-well-linked for the set $\out(S)$ of edges.
\end{definition}


Notice that if $|\out(S)|\leq k_1$, then $S$ is $\alpha$-well-linked iff it is $(k_1,\alpha)$-well-linked.

Similarly, if we are given a graph $G$ and a subset $\tset$ of its vertices called terminals, we say that $G$ is $\alpha$-well linked for $\tset$, iff for any partition $(A,B)$ of $V(G)$, $|E(A,B)|\geq \alpha\cdot \min\set{|A\cap \tset|,|B\cap \tset|}$.
Notice that if $G$ is $\alpha$-well-linked for $\tset$, then for any pair $(\tset_1,\tset_2)$ of subsets of $\tset$ with $|\tset_1|=|\tset_2|$, we can efficiently find a collection $\pset$ of paths, $\pset:\tset_1\sconnect_{\lceil 1/\alpha\rceil}\tset_2$. This follows from the min-cut max-flow theorem and the integrality of flow.

Given a graph $G$, a subset $S\sse V(G)$ of vertices, and a subset $\Gamma\sse \out(S)$ of edges, we define an instance $\SC(G,S,\Gamma)$ of the sparsest cut problem as follows. First, we sub-divide every edge $e\in \Gamma$ by a new vertex $t_e$, and we let $\tset(\Gamma)=\set{t_e\mid e\in \Gamma}$. Let $G_S(\Gamma)$ be the sub-graph of the resulting graph, induced by $S\cup \tset(\Gamma)$. The instance $\SC(G,S,\Gamma)$ of the sparsest cut problem is defined over the graph $G_S(\Gamma)$, where the vertices of $\tset(\Gamma)$ serve as terminals. Observe that for all $\alpha\leq 1$, the value of the sparsest cut in $\SC(G,S,\Gamma)$ is at least $\alpha$ iff set $S$ is $\alpha$-well-linked with respect to $\Gamma$ in graph $G$. If $\Gamma=\out_G(S)$, then we will denote the corresponding instance of the sparsest cut problem by $\SC(G,S)$.

{\bf The Grouping Technique.}
The grouping technique was first introduced by Chekuri, Khanna and Shepherd~\cite{ANF}, and has since been widely used in algorithms for network routing~\cite{CKS, RaoZhou, Andrews,EDP-old}, to boost  the network connectivity and well-linkedness parameters. 
We start with a simple standard grouping procedure, summarized in the following theorem. A proof of the following theorem can be found e.g. in~\cite{ANF}.

\begin{theorem}\label{thm: grouping}
Suppose we are given a connected graph $G=(V,E)$, with weights $w(v)$ on vertices $v\in V$, and a parameter $p$, such that for each $v\in V$, $0\leq w(v)\leq p$, and $\sum_{v\in V}w(v)\geq p$. Then we can efficiently find a partition $\gset$ of $V$, and for each group $U\in \gset$, find a tree $T_U\sse G$ containing all vertices of $U$, such that the trees $\set{T_U}_{U\in \gset}$ are edge-disjoint,  and for each $U\in \gset$, $p\leq \sum_{v\in U}w(v)\leq 3p$.
\end{theorem}

We also use a more sophisticated grouping technique, due to Chekuri, Khanna and Shepherd~\cite{ANF}. The next two theorems summarize the specific settings in which the technique is used. Since these settings are slightly different from the setting of~\cite{ANF}, we provide their proofs for completeness in Section~\ref{sec: proofs for groupings} in the Appendix. The proofs closely follow the arguments of~\cite{ANF}.

\begin{theorem}\label{thm: grouping-many-sets}
Suppose we are given a connected graph $G=(V,E)$ with maximum vertex degree at most $4$, and $r$ subsets $\tset_1,\ldots,\tset_r$ of vertices called terminals (we do not require that they are disjoint). Assume further that $G$ is $\alpha$-well-linked for the set $\bigcup_{j=1}^{r}\tset_j$ of terminals, for some $\alpha<1$, $|\tset_1|=k_1$, and that for each $1< j\leq r$, $|\tset_j|=k_2$,
 where $k_1,k_2\geq \Omega(r^2\log r/\alpha)$. Then there is an efficient randomized algorithm that  w.h.p. computes, for each $1\leq j\leq r$, a subset $\tset'_j\sse \tset_j$ of terminals, such that all sets $\set{\tset_j'}_{j=1}^r$ are mutually disjoint, $|\tset_j'|=\Omega(\alpha/r^2)\cdot |\tset_j|$ for all $1\leq j\leq r$, and $G$ is $1$-well-linked for $\bigcup_{j=1}^r\tset_j'$.
\end{theorem}


\begin{theorem}\label{thm: grouping-advanced}
Suppose we are given a connected graph $G=(V,E)$ with maximum vertex degree at most $4$, two subsets $\tset_1,\tset_2$ of vertices called terminals, and a perfect matching $\mset$ over the terminals of $\tset_1$. Assume further that $|\tset_1|=k_1$, $|\tset_2|=k_2$, where $k_1,k_2\geq 100/\alpha$,
and $G$ is $\alpha$-well-linked for the set $\tset_1\cup \tset_2$ of terminals, for some $\alpha<1$. Then there is an efficient algorithm, that 
either computes a routing of a subset $\mset'\sse \mset$ of $\Omega(\alpha k_1)$ pairs on edge-disjoint paths, or
returns two {\bf disjoint} subsets $\tset_1'\sse \tset_1$, $\tset_2'\sse \tset_2$, so that $|\tset_1'|=\Omega( \alpha k_1)$, $|\tset_2'|=\Omega(\alpha k_2)$,  $G$ is $1$-well-linked for $\tset_1'\cup \tset_2'$, and for every pair $(s,t)\in \mset$, either both $s$ and $t$ belong to $\tset_1'$, or neither of them does.
\end{theorem}

{\bf Expanders and the Cut-Matching Game.}
We say that a (multi)-graph $G=(V,E)$ is an $\alpha$-expander, iff
$\min_{\stackrel{S\sse V:}{|S|\leq |V|/2}}\set{\frac{|E(S,\nots)|}{|S|}}\geq \alpha$.

We use the cut-matching game of Khandekar, Rao and Vazirani~\cite{KRV}. In this game, we are given a set $V$ of $N$ vertices, where $N$ is even, and two players: a cut player, whose goal is to construct an expander $X$ on the set $V$ of vertices, and a matching player, whose goal is to delay its construction. The game is played in iterations. We start with the graph $X$ containing the set $V$ of vertices, and no edges.
In each iteration $j$, the cut player computes a bi-partition $(A_j,B_j)$ of $V$ into two equal-sized sets, and the matching player returns some perfect matching $M_j$ between the two sets. The edges of $M_j$ are then added to $X$. Khandekar, Rao and Vazirani have shown that there is a strategy for the cut player, guaranteeing that after $O(\log^2N)$ iterations we obtain a $\half$-expander w.h.p. Subsequently, Orecchia et al.~\cite{better-CMG} have shown the following improved bound:


\begin{theorem}[\cite{better-CMG}]\label{thm: CMG}
There is a probabilistic algorithm for the cut player, such that, no matter how the matching player plays, after $\gkrv(N)=O(\log^2N)$ iterations, graph $X$ is an $\alphaCMG(N)=\Omega(\log N)$-expander, with constant probability.
\end{theorem}

\label{--------------------------------------------------sec: Starting point------------------------------------------------}
\section{Algorithm Overview}\label{sec: alg overview}
Throughout the paper, we denote by $\gkrv=\gkrv(k)=O(\log^2k)$ the parameter from Theorem~\ref{thm: CMG}.
The central combinatorial object in our algorithm is what we call a \emph{good crossbar}, that we define below.

\begin{definition}
Given a graph $G=(V,E)$, a set $\mset$ of $k$ source-sink pairs, and a parameter $k^*=k/\poly\log k$, a good crossbar consists of the following three components:

\begin{enumerate}
\item A family $\sset^*=\set{S^*_1,\ldots,S^*_{\gkrv}}$ of disjoint subsets of {\bf non-terminal} vertices.
 Each set $S_j^*\in \sset^*$ is associated with a subset $\Gamma^*_j\sse \out(S^*_j)$ of $2k^*$ edges, and $S^*_j$ is $1$-well-linked for $\Gamma^*_j$.

\item A subset $\mset^*\sse \mset$ of $k^*$ demand pairs. Let $\tset^*=\tset(\mset^*)$ be the corresponding set of terminals.

\item A collection $\tau^*=\set{T_1,\ldots,T_{2k^*}}$ of trees in graph $G$. Each tree $T_i\in \tau^*$, contains a distinct terminal $t_i\in \tset^*$, and for each $1\leq j\leq \gkrv$, tree $T_i$ contains a distinct edge $e_{i,j}\in \Gamma^*_j$. In other words, $\tset^*=\set{t_1,\ldots,t_{2k^*}}$, where $t_i\in T_i$ for each $1\leq i\leq 2k^*$; and for each $1\leq j\leq \gkrv$, $\Gamma^*_j=\set{e_{1,j},\ldots,e_{2k^*,j}}$, where $e_{i,j}\in T_i$ for each $1\leq i\leq 2k^*$.
\end{enumerate}

We say that the congestion of the good crossbar is $c$ iff every edge of $G$ participates in at most $c$ trees in $\tau^*$, while every edge in $\bigcup_{j=1}^{\gkrv}E(S^*_j)$ belongs to at most $c-1$ such trees.
\end{definition}


Our main result is that if $G$ is $1$-well-linked for the set $\tset$ of terminals, then there is an efficient randomized algorithm, that either finds a good congestion-2 crossbar, or routes a subset of $k/\poly\log k$ demand pairs with congestion at most $2$ in $G$. We summarize this result in the following theorem.

\begin{theorem}\label{thm: main: find good crossbar or find a routing}
Assume that we are given an undirected graph $G=(V,E)$ with vertex degrees at most $4$, and a set $\mset$ of $k$ demand pairs, defined over a set $\tset$ of terminals. Assume further that the degree of every terminal is $1$, every terminal participates in exactly one pair in $\mset$, and $G$ is $1$-well-linked for $\tset$. Then there is an efficient randomized algorithm, that with high probability outputs one of the following:
(i) Either a subset $\mset'\sse \mset$ of $k/\poly\log k$ demand pairs and the routing of the pairs in $\mset'$ with congestion at most $2$ in $G$;
or (ii) a good congestion-2 crossbar  $(\sset^*,\mset^*,\tau^*)$.
\end{theorem}

The proof of Theorem~\ref{thm: main} follows from the proof of Theorem~\ref{thm: main: find good crossbar or find a routing} using techniques from previous work, and it appears in Section~\ref{sec: complete proof of main thm} of the Appendix. 
We note that the algorithm of~\cite{EDP-old} also proceeded by (implicitly) constructing a good congestion-$14$ crossbar. Our main challenge is to reduce the congestion of the good crossbar to $2$.
From now on we focus on proving Theorem~\ref{thm: main: find good crossbar or find a routing}.


As in~\cite{EDP-old}, throughout the algorithm, we maintain a partition $\cset$ of the vertices of $G$ into clusters, and a legal contracted graph of $G$, obtained by contracting each cluster $C\in \cset$ into a super-node $v_C$. The algorithm performs a number of phases, where in every phase we either find the desired routing of a subset of demand pairs in graph $G$, or compute a good congestion-$2$ crossbar, or find a new partition $\cset'$ of the vertices of $G$ into clusters, whose corresponding contracted graph is strictly smaller than the current contracted graph. In the next section we define several types of clusterings the algorithm uses, the notion of the legal contracted graphs, and several operations on a given clustering, that are used throughout the algorithm.

\iffalse@@@

\section{Good Family of Vertex Sets}
In this section we define a good family of vertex sets. The rest of the proof will consist of two parts. In the first part we show that we can find a good family in $G$. In the second part we show that given a good family, we can find a routing of a $(1/\poly\log k)$-fraction of pairs in $\mset$ with congestion at most $2$.

\begin{definition}
We say that $\set{S_1,\ldots,S_r}$, for $r\geq 8\gkrv$ is a good family of vertex subsets, iff:

\begin{itemize}
\item All sets $S_i$ are vertex-disjoint and they do not contain terminals in $\tset$.
\item For all $1\leq i\leq r$, we are given a set $\Gamma_i\sse \out(S_i)$ of edges, and $S_i$ is $\alphaWL$-well-linked for $\Gamma_i$. That is, for any partition $(X,Y)$ of $S_i$, $|E(X,Y)|\geq \alphaWL\cdot\min\set{|\out(X)\cap \Gamma_i|,|\out(Y)\cap \Gamma_i|}$.
\item For all $1\leq i<r$, we are given a set  $\pset_i$ of $k/\poly\log k$ edge-disjoint paths, connecting $\Gamma_i$ to $\Gamma_{i+1}$, and these paths do not contain any vertices of $\tset\cup\left (\bigcup_{i=1}^rS_i\right )$ as inner vertices.

\item We are given a subset $\mset'\sse \mset$ of $k/\poly\log k$ source-sink pairs, such that one of the following holds:

\begin{itemize}
\item Either there is a set $\pset_0$ of edge-disjoint paths, connecting every terminal that participates in the pairs in $\mset'$ to $\Gamma_1$;

\item Or we can partition the terminals participating in the pairs in $\mset'$ into two subsets, $\tset'$ and $\tset''$, such that for each pair $(s,t)\in \mset'$, $s\in \tset'$, $t\in \tset''$ or vice-versa, and there are two sets $\pset_0: \tset'\connect_1\Gamma_1$, $\pset_{r+1}:\tset''\connect_1\Gamma_r$ of paths.
\end{itemize}
In either case, the paths of $\pset_0\cup \pset_{r+1}$  do not contain any vertices of $\tset\cup\left (\bigcup_{i=1}^rS_i\right )$ as inner vertices. 
\end{itemize}
\end{definition}

The rest of the proof consists of two parts. First, we show that we can efficiently find a good family of vertex subsets in $G$. The second part finds the routing given the good family of vertex subsets.

@@@
\fi

\label{------------------------------------sec: legal contracted graphs---------------------------------------}
\section{Vertex Clusterings and Legal Contracted Graphs}
\label{sec: clusterings}

Our algorithm uses the following parameters.
Let $\gkrv=\gkrv(k)=O(\log^2k)$ be the parameter for the number of iterations in the cut-matching game from Theorem~\ref{thm: CMG}. We let $\gamma=2^{24}\gkrv^4$.
We will also use the following two parameters for well-linkedness: $\alpha=\frac{1}{2^{11}\gamma\log k}=\Omega\left (\frac{1}{\log^{9}k}\right )$, used to perform the well-linked decomposition, and
$\alphaWL=\frac{\alpha}{\alphasc(k)}=\Omega\left (\frac{1}{\log^{9.5}k}\right )$ - the well-linkedness factor we achieve.
Finally, we use a parameter $k_1=\frac{k}{192\gamma^3\log \gamma}=\frac{k}{\poly\log k}$, and we assume that the parameter $k$ is large enough, so $k_1>4/\alpha$ (otherwise we can simply route one of the demand pairs to obtain an $O(\poly\log k)$-approximation).
We say that a cluster $C\sse V(G)$ is \emph{large} iff $|\out(C)|\geq k_1$, and we say that it is \emph{small} otherwise.


\begin{definition}
Given a partition $\cset$ of the vertices of $V(G)$ into clusters, we say that $\cset$ is an \emph{acceptable clustering} of $G$ iff:

\begin{itemize}
\item  Every terminal $t\in \tset$ is in a separate cluster, that is, $\set{t}\in \cset$; 
\item Each small cluster $C\in \cset$ is $\alphaWL$-well-linked; and 
\item Each large cluster $C\in \cset$ is a connected component.
An acceptable clustering that contains no large clusters is called a \emph{good clustering}.
\end{itemize}
\end{definition}

\begin{definition}
Given a good clustering $\cset$ of $G$, a graph $H_{\cset}$ is a legal contracted graph of $G$ associated with $\cset$, iff we can obtain $H_{\cset}$ from $G$ by contracting every $C\in\cset$ into a super-node $v_C$. We remove all self-loops, but we do not remove parallel edges. (Note that the terminals are not contracted since each terminal has its own cluster).
\end{definition}

The following claim was proved in \cite{EDP-old}. We include its proof here for completeness, since we have changed some parameters.


\begin{claim}\label{claim: legal graph has many edges}
If $G'$ is a legal contracted graph for $G$, then $G'\setminus \tset$ contains at least $k/3$ edges.
\end{claim}

\begin{proof}
For each terminal $t\in \tset$, let $e_t$ be the unique edge adjacent to $t$ in $G'$, and let $u_t$ be the other endpoint of $e_t$. We partition the terminals in $\tset$ into groups, where two terminals $t,t'$ belong to the same group iff $u_t=u_{t'}$. Let $\gset$ be the resulting partition of the terminals. Since the degree of every vertex in $G'$ is at most $k_1$, each group $U\in \gset$ contains at most $k_1$ terminals. Next, we partition the terminals in $\tset$ into two subsets $X,Y$, where $|X|,|Y|\geq k/3$, and for each group $U\in \gset$, either $U\sse X$, or $U\sse Y$ holds. We can find such a partition by greedily processing each group $U\in \gset$, and adding all terminals of $U$ to one of the subsets $X$ or $Y$, that currently contains fewer terminals. Finally, we remove terminals from set $X$ until $|X|=k/3$, and we do the same for $Y$. Since graph $G'$ is $1$-well-linked for the terminals, it is possible to route $k/3$ flow units from the terminals in $X$ to the terminals in $Y$, with congestion at most $1$. Since no group $U$ is split between the two sets $X$ and $Y$, each flow-path must contain at least one edge of $G'\setminus \tset$. Therefore, the number of edges in $G'\setminus \tset$ is at least $k/3$.
\end{proof}

Given {\bf any} partition $\cset$ of the vertices of $G$, we define a potential $\phi(\cset)$ for this clustering. The idea is that $\phi(\cset)$ will serve as a tight bound on the number of edges connecting different clusters in $\cset$. At the same time, the potential function is designed in such a way, that we can perform a number of operations on the current clustering (such as, for example, well-linked decompositions of the clusters), without increasing the potential. 

At a high level, our algorithm maintains a good clustering $\cset$ of $V(G)$, where at the beginning, every vertex belongs to a separate cluster. The algorithm  consists of a number of phases, where in each phase we start with some good clustering $\cset$, and either route $k/\poly\log k$ of the demand pairs with congestion at most $2$ in $G$, or find a good congestion-$2$ crossbar in $G$, or produce another good clustering $\cset'$ with $\phi(\cset')\leq \phi(\cset)-1$. After a polynomial number of phases, we will therefore either obtain a routing of $k/\polylog k$ demand pairs with congestion at most $2$, or find a good congestion-$2$ crossbar.

Each phase of the algorithm is executed as follows. We start with the current good clustering $\cset$, and construct another acceptable clustering $\cset'$ with $\phi(\cset')<\phi(\cset)-1$. We then perform a number of iterations, where in each iteration we either route $k/\poly\log k$ demand pairs in $G$, or find a good congestion-$2$ crossbar, or revisit the current acceptable clustering $\cset'$, and perform an operation that produces another acceptable clustering, whose potential is strictly smaller than that of $\cset'$. In the end, we will either successfully route $k/\poly\log k$ demand pairs in $G$, or construct a good congestion-$2$ crossbar in $G$, or we will find a good clustering $\cset'$ with $\phi(\cset')<\phi(\cset)-1$. We now proceed to define the potential function, and the operations we perform on clusterings.

Suppose we are given {\bf any} partition $\cset$ of the vertices of $G$. 
For any integer $h$, we define a potential $\phi(h)$, as follows. For $h< k_1$,  $\phi(h)=4\alpha\log h$.
In order to define $\phi(h)$ for  $h\geq k_1$, we consider the sequence $\set{n_0,n_1,\ldots}$ of numbers, where $n_i=\left(\frac 3 2\right )^i k_1$. The potentials for these numbers are $\phi(n_0)=\phi(k_1)=4\alpha\log k_1+4\alpha$, and for $i>0$, $\phi(n_i)=4\frac{\alpha k_1}{n_i}+\phi(n_{i-1})$. Notice that for all $i$, $\phi(n_i)\leq 12\alpha+4\alpha\log k_1\leq 8\alpha\log k_1\leq \frac{1}{2^8\gamma}$.

We now partition all integers $h>k_1$ into sets $S_1,S_2,\ldots$, where set $S_i$ contains all integers $h$ with $n_{i-1}\leq h< n_i$. For $h\in S_i$, we define $\phi(h)=\phi(n_{i-1})$.
This finishes the definition of potentials of integers. Clearly, for all $h$, $\phi(h)\leq \frac{1}{2^8\gamma}$.

Assume now that we are given some edge $e\in E$. If both endpoints of $e$ belong to the same cluster of $\cset$, then we set its potential $\phi(e)=0$. Otherwise, if $e=(u,v)$, and $u\in C$ with $|\out(C)|=h$, while $v\in C'$ with $|\out(C')|=h'$, then we set $\phi(e)=1+\phi(h)+\phi(h')$. We think of $\phi(h)$ as the contribution of $u$, and $\phi(h')$ the contribution of $v$ to $\phi(e)$. Notice that $\phi(e)\leq 1.1$. Finally, we set $\phi(\cset)=\sum_{e\in E}\phi(e)$.

Suppose we are given any partition $\cset$ of $V(G)$. Our first step is to show that we can perform a well-linked decomposition of small clusters in $\cset$, without increasing the potential. 


\begin{theorem}\label{thm: well-linked decomposition of small clusters}
Let $\cset$ be any partition of $V(G)$, and let $C\in \cset$ be any small cluster, such that $G[C]$ is connected. Then there is an efficient algorithm that finds a partition $\wset$ of $C$ into small clusters, such that each cluster $R\in \wset$ is $\alphaWL$-well-linked, and additionally, if $\cset'$ is a partition obtained from $\cset$ by removing $C$ and adding the clusters of $\wset$ to it, then $\phi(\cset')\leq \phi(\cset)$.
\end{theorem}

\begin{proof}
We perform a standard well-linked decomposition of $C$. We maintain a partition $\wset$ of $C$, where at the beginning, $\wset=\set{C}$. We then perform a number of iterations.

In each iteration, we select a cluster $S\in \wset$, and set up the following instance of the sparsest cut problem.
First, we sub-divide every edge $e\in \out(S)$ with a vertex $t_e$, and let $\tset_S=\set{t_e\mid e\in \out(S)}$. We then consider the sub-graph of the resulting graph induced by $S\cup \tset_S$, where the vertices of $\tset_S$ serve as terminals. We
run the algorithm \algsc on the resulting instance of the sparsest cut problem. If the sparsity of the cut produced by the algorithm is less than $\alpha$, then we obtain a partition $(X,Y)$ of $S$, with $|E(X,Y)|<\alpha\cdot\min\set{|\out(S)\cap \out(X)|,|\out(S)\cap \out(Y)|}$. In this case, we remove $S$ from $\wset$, and add $X$ and $Y$ instead. The algorithm ends when for every cluster $S\in \wset$, $\algsc$ returns a partition of sparsity at least $\alpha$. We are then guaranteed that every cluster in $\wset$ is $\alpha/\alphasc(k)=\alphaWL$-well-linked, and it is easy to verify that all these clusters are small.

It now only remains to show that the potential does not increase. Each iteration of the algorithm is associated with a partition of the vertices of $G$, obtained from $\cset$ by removing $C$ and adding all clusters of the current partition $\wset$ of $C$ to it. It is enough to show that if $\cset'$ is the current partition of $V(G)$, and $\cset''$ is the partition obtained after one iteration, where a set $S\in \cset$ was replaced by two sets $X$ and $Y$, then $\phi(\cset'')\leq \phi(\cset')$.

Assume w.l.o.g. that $|\out(X)|\leq |\out(Y)|$, so $|\out(X)|\leq 2|\out(S)|/3$. Let $h=|\out(S)|,h_1=|\out(X)|, h_2=|\out(Y)|$, and recall that $h,h_1,h_2<k_1$.
The changes to the potential are the following:

\begin{itemize}
\item The potential of the edges in $\out(Y)\cap \out(S)$ only goes down.

\item The potential of every edge in $\out(X)\cap \out(S)$ goes down by $\phi(h)-\phi(h_1)=4\alpha\log h-4\alpha\log h_1=4\alpha\log \frac{h}{h_1}\geq 4\alpha\log 1.5\geq 2.3\alpha$, since $h_1\leq 2h/3$. So the total decrease in the potential of the edges in $\out(X)\cap \out(S)$ is at least $2.3\alpha\cdot |\out(X)\cap \out(S)|$.

\item The edges in $E(X,Y)$ did not contribute to the potential initially, and now contribute at most $1+\phi(h_1)+\phi(h_2)\leq 2$ each. Notice that $|E(X,Y)|\leq \alpha \cdot |\out(X)\cap \out(S)|$, and so they contribute at most $2\alpha \cdot |\out(X)\cap \out(S)|$ in total.
\end{itemize}
Clearly, from the above discussion, the overall potential goes down.
\end{proof}

Assume that we are given an acceptable clustering $\cset$ of $G$. We now define two operations on $G$,  each of which produces a new acceptable clustering of $G$, whose potential is strictly smaller than $\phi(\cset)$.

{\bf Action 1: Partitioning a large cluster.}
Suppose we are given an acceptable clustering $\cset$ of $G$, a large cluster $C\in \cset$, and a $(k_1,\alpha)$-violating partition $(X,Y)$ of $C$. 
In order to perform this operation, we first replace $C$ with $X$ and $Y$ in $\cset$. If, additionally, any of the clusters, $X$ or $Y$, become small, then we perform a well-linked decomposition of that cluster using Theorem~\ref{thm: well-linked decomposition of small clusters}, and update $\cset$ with the resulting partition. 
Clearly, the final partitioning $\cset'$ is an acceptable clustering. 
We denote this operation by $\partition(C,X,Y)$.

\begin{claim}\label{claim: bound on potential for partition}
Let $\cset'$ be the outcome of operation $\partition(C,X,Y)$. Then $\phi(\cset')<\phi(\cset)-2$. 
\end{claim}

\begin{proof}
Let $\cset''$ be the clustering obtained from $\cset$, by replacing $C$ with $X$ and $Y$. From Theorem~\ref{thm: well-linked decomposition of small clusters}, it is enough to prove that $\phi(\cset'')<\phi(\cset)-2$.

Assume w.l.o.g. that $|\out(X)|\leq |\out(Y)|$. 
Let $h=|\out(C)|$, $h_1=|\out(X)|$, $h_2=|\out(Y)|$, so $h_1< 2h/3$. Assume that $h\in S_i$. Then either $h_1\in S_{i'}$ for $i'\leq i-1$, or $h_1<k_1$. The changes in the potential can be bounded as follows:

\begin{itemize}
\item The potential of every edge in $\out(Y)\cap \out(C)$ does not increase.
\item The potential of every edge in $\out(X)\cap \out(C)$ goes down by $\phi(h)-\phi(h_1)$. 
Assuming that $h\in S_i$, $\phi(h)\geq \frac{4\alpha k_1}{n_{i-1}}+\phi(h_{1})\geq \frac{4\alpha k_1}{h_1}+\phi(h_{1})$. The total decrease in the potential of the edges in $\out(X)\cap \out(C)$ is at least 
$\frac{4\alpha k_1}{h_1}\cdot |\out(X)\cap \out(C)|\geq 2\alpha k_1$, since $|\out(X)\cap \out(C)|>h_1/2$.
\item Additionally, every edge in $E(X,Y)$ now pays $1+\phi(h_1)+\phi(h_2)<2$. Since $|E(X,Y)|\leq \alpha k_1/2$, the total increase in this part is at most $\alpha k_1$.
\end{itemize}

The overall decrease in the potential is at least $\alpha k_1\geq 2$.
\end{proof}

{\bf Action 2: Separating a large cluster.}
Let $\cset$ be the current acceptable partition, and let $C\in \cset$ be a large cluster in $\cset$. Assume further that we are given a cut $(A,B)$ in graph $G$, with $C\sse A$, $\tset\sse B$, and $|E_G(A,B)|< k_1/2$. We perform the following operation, that we denote by $\separate(C,A)$.

Consider some cluster $S\in \cset$. If $S$ is a small cluster, but $S\setminus A$ is a large cluster, then we modify $A$ by removing all vertices of $S$ from it. Notice that in this case, the number of edges in $E(S)$ that originally contributed to the cut $(A,B)$,  $|E(S\cap A,S\cap  B)|>|\out(S)\cap E(A)|$ must hold, so $|\out(A)|$ only goes down as a result of this modification. We assume from now on that if $S\in \cset$ is a small cluster, then $S\setminus A$ is also a small cluster.
We build a new partition $\cset'$ of $V(G)$ as follows. First, we add every connected component of $G[A]$ to $\cset$. Notice that all these clusters are small, as $|\out(A)|<k_1/2$. Next, for every cluster $S\in\cset$, such that $S\setminus A\neq \emptyset$, we add every connected component of $G[S\setminus A]$ to $\cset'$. Notice that every terminal $t\in \tset$ is added as a separate cluster to $\cset'$. So far we have defined a new partition $\cset'$ of $V(G)$. This partition may not be acceptable, since we are not guaranteed that every small cluster of $\cset'$ is well-linked. In our final step, we perform a well-linked decomposition of every small cluster of $\cset'$, using Theorem~\ref{thm: well-linked decomposition of small clusters}, and obtain the final acceptable partition $\cset''$ of vertices of $G$. 
Notice that if $S\in \cset''$ is a large cluster, then there must be some {\bf large} cluster $S'$ in the original partition $\cset$ with $S\sse S'$.

\begin{claim}\label{claim: bound on potential for separation}
 Let $\cset''$ be the outcome of operation  $\separate(C,A)$. Then $\phi(\cset'')\leq \phi(\cset)-1$. 
 \end{claim}

 \begin{proof}
 In order to prove the claim, it is enough to prove that $\phi(\cset')\leq \phi(\cset)-1$, since, from Theorem~\ref{thm: well-linked decomposition of small clusters}, well-linked decompositions of small clusters do not increase the potential.

We now show that $\phi(\cset')\leq \phi(\cset)-2$.
We can bound the changes in the potential as follows:

\begin{itemize}
\item Every edge in $\out(A)$ contributes at most $1.1$ to the potential of $\cset''$, and there are at most $ \frac{k_1-1} 2$ such edges. These are the only edges whose potential in $\cset''$ may be higher than their potential in $\cset$.

\item Every edge in $\out(C)$ contributed at least $1$ to the potential of $\cset'$, and there are at least $k_1$ such edges, since $C$ is a large cluster.
\end{itemize}

Therefore, the decrease in the potential is at least $k_1-\frac{1.1(k_1-1)}2\geq 1$.
\end{proof}

To summarize, given any acceptable clustering $\cset$ of the vertices of $G$, let $E'$ be the set of edges whose endpoints belong to distinct clusters of $\cset$. Then $|E'|\leq \phi(\cset)\leq 1.1|E'|$. So the potential is a good estimate on the number of edges connecting the different clusters. We have also defined two actions on large clusters of $\cset$: $\partition(C,X,Y)$, that replaces a large cluster $C$ with a pair $X,Y$ of clusters, where $(X,Y)$ is a $(k,\alpha)$-violating partition of $C$, and $\separate(C,A)$, where $A$ is a cut of size less than $k/2$, separating a large cluster $C$ from the terminals. Each such action returns a new acceptable clustering, whose potential goes down by at least $1$.

\label{------------------------------------------------------sec: alg---------------------------------------------}
\section{The Algorithm}\label{sec: Alg}

In this section we prove Theorem~\ref{thm: main: find good crossbar or find a routing}, by providing an efficient randomized algorithm, that w.h.p. either computes a subset $\mset'\sse \mset$ of $k/\poly\log k$ demand pairs and their routing with congestion at most $2$ in $G$, or finds a good congestion-$2$ crossbar in $G$.

We maintain, throughout the algorithm, a good clustering $\cset$ of $G$. Initially, $\cset$ is a partition of $V(G)$, where every vertex of $G$ belongs to a distinct cluster, that is, $\cset=\set{\set{v}\mid v\in V(G)}$. Clearly, this is a good clustering. We then perform a number of phases. In every phase, we start with some good clustering $\cset$ and the corresponding legal contracted graph $H_{\cset}$. The phase output is one of the following: either (1) a subset $\mset'\sse \mset$ of $k/\poly\log k$ pairs, with  the routing of the pairs in $\mset'$ with congestion at most $2$ in $G$, or (2) a good congestion-$2$ crossbar $(\sset^*,\mset^*,\tau^*)$ in $G$; or (3) another good clustering $\cset'$, such that $\phi(\cset')\leq \phi(\cset)-1$. 
In the first two cases, we terminate the algorithm, and output either the routing of the pairs in $\mset'$, or the good crossbar. In the latter case, we continue to the next phase. After $O(|E|)$ phases, our algorithm will successfully terminate with the required output. It is therefore enough to prove the following theorem.

\begin{theorem}\label{thm: find good crossbar or a smaller contracted graph}
Let $\cset$ be any good clustering of the vertices of $G$, and let $H_{\cset}$ be the corresponding legal contracted graph. Then there is an efficient randomized algorithm that w.h.p. computes one of the following:

\begin{itemize}
\item Either a subset $\mset'\sse \mset$ of $k/\poly\log k$ pairs and a routing of the pairs in $\mset'$ in graph $G$ with congestion at most $2$;

\item or a good congestion-$2$ crossbar $(\sset^*,\mset^*,\tau^*)$ in $G$;

\item or a new good clustering $\cset'$, such that $\phi(\cset')\leq \phi(\cset)-1$.
\end{itemize}
\end{theorem}

The rest of this section is dedicated to proving Theorem~\ref{thm: find good crossbar or a smaller contracted graph}.

We assume that we are given a good clustering $\cset$ of the vertices of $G$, and the corresponding legal contracted graph $G'=H_{\cset}$.

 Let $m$ be the number of edges in $G'\setminus \tset$. From Claim~\ref{claim: legal graph has many edges}, $m\geq k/3$.
As a first step, we randomly partition the vertices in $G'\setminus \tset$ into $\gamma$ subsets $X_1,\ldots,X_{\gamma}$,
where each vertex $v\in V(G')\setminus \tset$ selects an index $1\leq j\leq \gamma$ independently uniformly at random, and is then added to $X_j$. We need the following claim, that appeared in~\cite{EDP-old}. The proof appears in Section~\ref{sec: proof of partitioning claim} of the Appendix for completeness.

\begin{claim}\label{claim: random partition into gamma sets} With probability at least $\half$, for each $1\leq j\leq \gamma$, $|\out_{G'}(X_j)|< \frac{10m}{\gamma}$, while $|E_{G'}(X_j)|\geq \frac{m}{2\gamma^2}$.
\end{claim}
 

 Given a partition $X_1,\ldots,X_{\gamma}$, we can efficiently check whether the conditions of Claim~\ref{claim: random partition into gamma sets} hold. If they do not hold, we repeat the randomized partitioning procedure.  From Claim~\ref{claim: random partition into gamma sets}, we are guaranteed that w.h.p., after $\poly(n)$ iterations, we will obtain a partition with the desired properties. Assume now that we are given the partition $X_1,\ldots,X_{\gamma}$ of $V(G')\setminus \tset$, for which the conditions of Claim~\ref{claim: random partition into gamma sets} hold. Then for each $1\leq j\leq \gamma$, $|E_{G'}(X_j)|>\frac{|\out_{G'}(X_j)|}{20\gamma}$. Let $X'_j\sse V(G)\setminus \tset$ be the set obtained from $X_j$, after we un-contract each cluster, that is, for each super-node $v_C\in X_j$, we replace $v_C$ with the vertices of $C$. Notice that $\set{X'_j}_{j=1}^{\gamma}$ is a partition of $V(G)\setminus\tset$.

The plan for the rest of the proof is as follows. For each $1\leq j\leq \gamma$, we will maintain an acceptable clustering $\cset_j$ of the vertices of $G$. That is, for each $1\leq j\leq \gamma$, $\cset_j$ is a partition of $V(G)$. In addition to being an acceptable clustering, it will have the following property:

\begin{properties}{P}
\item If $C\in \cset_j$ is a large cluster, then $C\subseteq X'_j$. \label{property of Cj}
\end{properties}

The initial partition $\cset_j$, for $1\leq j\leq \gamma$ is obtained as follows. Recall that $\cset$ is the current good clustering of the vertices of $G$, and every cluster $C\in \cset$ is either contained in $X'_j$, or it is disjoint from it. First, we add to $\cset_j$ all clusters $C\in \cset$ with $C\cap X'_j=\emptyset$. Next, we add to $\cset_j$ all connected components of $G[X'_j]$. If any of these components is  a small cluster, then we perform a well-linked decomposition of this cluster, using Theorem~\ref{thm: well-linked decomposition of small clusters}, and update $\cset_j$ accordingly. Let $\cset_j$ be the resulting final partition. Clearly, it is an acceptable clustering, with property (\ref{property of Cj}). Moreover, we show that $\phi(\cset_j)\leq \phi(\cset)-1$:

\begin{claim}
For each $1\leq j\leq \gamma$, $\phi(\cset_j)\leq \phi(\cset)-1$.
\end{claim}
\begin{proof}
Let $\cset'_j$ be the partition of $V(G)$, obtained as follows: we add to $\cset'_j$ all clusters $C\in \cset$ with $C\cap X_j=\emptyset$, and we add all connected components of $G[X_j]$ to $\cset'_j$ (that is, $\cset'_j$ is obtained like $\cset_j$, except that we do not perform well-linked decompositions of the small clusters). From Theorem~\ref{thm: well-linked decomposition of small clusters}, it is enough to prove that $\phi(\cset'_j)\leq \phi(\cset)-1$. The changes of the potential from $\cset$ to $\cset'_j$ can be bounded as follows:

\begin{itemize}
\item The edges in $E_{G'}(X_j)$ contribute at least $1$ to $\phi(\cset)$ and contribute $0$ to $\phi(\cset'_j)$.
\item The potential of edges in $\out_G(X'_j)$ may increase. The increase is at most $\phi(n)\leq \frac{1}{2^8\gamma}$ per edge. So the total increase is at most $\frac{|\out_{G'}(X_j)|}{2^8\gamma}\leq \frac{|E_{G'}(X_j)|}{4}$. These are the only edges whose potential may increase.
\end{itemize}

Overall, the decrease in the potential is at least $\frac{|E_{G'}(X_j)|}{2}\geq\frac{m}{4\gamma^2}\geq \frac{k}{12\gamma^2}\geq 1$.
\end{proof}

 If any of the partitions $\cset_1,\ldots,\cset_{\gamma}$ is a good partition, then we have found a good partition $\cset'$ with $\phi(\cset')\leq \phi(\cset)-1$. We terminate the algorithm and return $\cset'$. Otherwise, we select an arbitrary large cluster $S_j\in \cset_j$. We then consider the resulting collection $S_1,\ldots,S_{\gamma}$ of large clusters, and try to exploit them to construct a good crossbar. Since for each $1\leq j\leq \gamma$, $S_j\sse X'_j$, the sets $S_1,\ldots,S_{\gamma}$ are mutually disjoint and they do not contain terminals.
 Our algorithm performs a number of iterations, using the following theorem.

 \begin{theorem}\label{thm: iteration}
  Suppose we are given, for each $1\leq j\leq \gamma$, an acceptable partition $\cset_j$ of $V(G)$ that has the property (\ref{property of Cj}), and contains at least one large cluster $S_j$. Then there is an efficient randomized algorithm, that w.h.p. computes one of the following:
 
 \begin{itemize}
 \item Either a subset $\mset'\sse \mset$ of $k/\poly\log k$ demand pairs, and a routing of pairs in $\mset'$ with congestion at most $2$ in $G$;
 
 \item Or a good congestion-$2$ crossbar $(\sset^*,\mset^*,\tau^*)$;

 \item Or a $(k_1,\alpha)$-violating partition $(X,Y)$ of $S_j$,  for some $1\leq j\leq \gamma$;
 
 \item Or a cut $(A,B)$ in $G$ with $S_j\sse A$, $\tset\sse B$ and $|E_G(A,B)|<k_1/2$, for some $1\leq j\leq \gamma$.
 \end{itemize}
 \end{theorem}

 We provide the proof of Theorem~\ref{thm: iteration} in the following section, and complete the proof of Theorem~\ref{thm: find good crossbar or a smaller contracted graph} here. Suppose we are given a good partition $\cset$ of the vertices of $G$. For each $1\leq j\leq \gamma$, we compute an acceptable partition $\cset_j$ of $V(G)$ 
 as described above. If any of the partitions $\cset_j$ is a good partition, then we terminate the algorithm and return $\cset_j$. From the above discussion, $\phi(\cset_j)\leq \phi(\cset)-1$. Otherwise, for each $1\leq j\leq \gamma$, we select any large cluster $S_j\in \cset_j$, and apply Theorem~\ref{thm: iteration} to the current family $\set{\cset_j}_{j=1}^{\gamma}$ of acceptable clusterings. If the outcome of Theorem~\ref{thm: iteration} is a subset $\mset'$ of demand pairs with a routing of these pairs in $G$, then we terminate the algorithm and return this routing. If the outcome is a good congestion-$2$ crossbar, then we terminate the algorithm and return this good crossbar. We say that the iteration is successful if one of these two cases happen. Otherwise, we apply the appropriate action: $\partition(S_j,X,Y)$, or $\separate(S_j,A)$ to the clustering $\cset_j$. As a result, we obtain an acceptable clustering $\cset'_j$, with $\phi(\cset'_j)\leq \phi(\cset_j)-1$. Moreover, it is easy to see that this clustering also has Property~(\ref{property of Cj}): if the $\partition$ operation is performed, then we only partition existing clusters; if the $\separate$ operation is performed, then the only large clusters in the new partition $\cset'_j$ are subsets of large clusters in $\cset_j$.

 If all clusters in $\cset'_j$ are small, then we can again terminate the algorithm with a good partition $\cset'_j$, with $\phi(\cset'_j)\leq \phi(\cset)-1$. Otherwise, we select any large cluster $S'_j\in \cset'_j$, and continue to the next iteration. Overall, as long as we do not complete a successful iteration, and we do not find a good clustering $\cset'$ of $V(G)$ with $\phi(\cset')\leq \phi(\cset)-1$, we make progress in each iteration by decreasing the potential of one of the partitions $\cset_j$ by at least $1$, by performing either a $\separate$ or a $\partition$ operation on one of the large clusters of $\cset_j$. After polynomially many iterations we are then guaranteed to complete a successful iteration, or find a good clustering $\cset'$ with $\phi(\cset')\leq \phi(\cset)-1$, and finish the algorithm. Therefore, in order to complete the proof of Theorem~\ref{thm: find good crossbar or a smaller contracted graph} it is now enough to prove Theorem~\ref{thm: iteration}.

 \label{---------------------------------------------sec: iteration execution--------------------------------}
 \section{Proof of Theorem~\ref{thm: iteration}\label{sec: proof of iteration theorem}}
 
  Let $\rset=\set{S_1,\ldots,S_{\gamma}}$.
 
 Throughout the algorithm, we will sometimes be interested in routing flow across the sets $S_j\in \rset$. Specifically, given two subsets $\Gamma,\Gamma'\sse \out(S_j)$ of edges, with $|\Gamma|=|\Gamma'|\leq k_1/2$, we will be interested in routing the edges of $\Gamma$ to the edges of $\Gamma'$ inside $S_j$, with congestion at most $ 1/\alpha$. In other words, we will be looking for a flow $F: \Gamma\sconnect_{1/\alpha}\Gamma'$. Notice that if such a flow does not exist, then we can find a $(k_1,\alpha)$-violating partition $(X,Y)$ of $S_j$, by using the min-cut max-flow theorem. We can then return this partition and terminate the algorithm. Therefore, in order to simplify the exposition of the algorithm, we will assume that whenever the algorithm attempts to find such a flow, it always succeeds.
From the integrality of flow, we can also find a set $\pset: \Gamma\sconnect_{\lceil 1/\alpha\rceil }\Gamma'$ of paths contained in $S_j$.

We start by verifying that for each $1\leq j\leq\gamma$, the vertices of $S_j$ can send $k_1/2$ flow units with no congestion to the terminals. If this is not the case for some set $S_j$, then there is a cut $(A,B)$ with $S_j\sse A$, $\tset\sse B$ and $|E_G(A,B)|<k_1/2$. We then return the partition $(A,B)$ of $G$ and finish the algorithm.
 From now on we assume that each set $S_j$ can send $k_1/2$ flow units with no congestion to the terminals.
 
 Let $G'$ be the graph obtained from $G$ by replacing every edge of $G$ by two bi-directed edges.
 The rest of the proof consists of three steps. In the first step, we construct a degree-$3$ tree $\tT$, whose vertex set $V(\tT)=\set{v_S\mid S\in \rset'}$, for a large enough family $\rset'\sse \rset$ of vertex subsets, and each edge $e=(v_S,v_{S'})$ in tree $\tT$ corresponds to a collection $\pset_e$ of paths in graph $G'$, connecting the vertices of $S$ to the vertices of $S'$ or vice versa. Moreover, we will ensure that the paths in $\bigcup_{e\in E(\tT)}\pset_e$  only cause congestion $1$ in $G'$. In the second step, we find a subset $\mset^*\sse \mset$ of the demand pairs, and route the terminals in $\tset(\mset^*)$ to the vertices of $S\cup S'$, where $(S,S')$ is some pair of vertex subsets in $\rset'$. In the final third step, we construct a good congestion-$2$ crossbar.

\subsection{Step 1: Constructing the Tree}\label{subsec: step 1: building trees}
This step is summarized in the following theorem.

\begin{theorem}\label{thm: building trees}
There is an efficient algorithm, that either computes a $(k_1,\alpha)$-violating partition of some set $S\in \rset$, or  finds a subset $\rset'\sse \rset$ of size $r=8\gcmg$, a tree $\tT$ of maximum degree $3$ and vertex set $V(\tT)=\set{v_{S}\mid S\in \rset'}$, and a collection $\pset_e$ of $k_2=\Omega\left (\frac{k_1\alpha\cdot\alphaWL}{\gamma^{3.5}}\right )$ paths in graph $G'$ for every edge $e\in E(\tT)$, such that:

\begin{itemize}
\item For each edge $e=(v_S,v_{S'})\in E(\tT)$, every path $P\in \pset_e$ connects a vertex of $S$ to a vertex of $S'$, or a vertex of $S'$ to a vertex of $S$, and it does not contain the vertices of $\bigcup_{S''\in \rset'}S''$ as inner vertices.

\item The set $\bigcup_{e\in E(\tT)}\pset_e$ of paths causes congestion at most $1$ in $G'$.
\end{itemize}
\end{theorem}

\begin{proof}
 Since the set $\tset$ of terminals is $1$-well-linked, every pair $(S_j,S_{j'})$ of vertex subsets can send $k_1/2$ flow units to each other with congestion at most $3$ (concatenate the flows from $S_j$ to a subset $\tset_1$ of the terminals, from $S_{j'}$ to a subset $\tset_2$ of the terminals, and the flow between the two subsets of the terminals).
 Equivalently, for each pair $(S_j,S_{j'})$, there are at least $\lfloor k_1/6\rfloor $ edge-disjoint paths connecting $S_j$ to $S_j'$ in $G$. We can assume that these paths do not contain any terminals, as the degree of every terminal in $G$ is $1$. We say that a path $P$ is \emph{direct} iff it does not contain the vertices of $\tset\cup\left (\bigcup_{j=1}^{\gamma}S_j\right )$ as its inner vertices.
 
 We build the following graph $Z$. Let $k'=\lfloor \frac{k_1}{6\gamma^2}\rfloor =\frac{k}{\poly\log k}$.
 The vertices of $Z$ are $\set{v_1,\ldots,v_{\gamma}}$, where vertex $v_j$ represents the set $S_j$. There is an edge $(v_j,v_{j'})$ for $j\neq j'$ in $Z$ iff there are at least $k'$ direct edge-disjoint paths connecting $S_j$ to $S_{j'}$ in $G$. (Notice that this can be checked efficiently). 
 
It is easy to verify that graph $Z$ is connected: indeed, assume otherwise. Let $A$ be any connected component of $Z$, and let $B$ contain the rest of the vertices. Let $v_j\in A$, $v_{j'}\in B$. Since there are at least $\lfloor k_1/6\rfloor $ edge-disjoint paths connecting $S_j$ and $S_{j'}$ in $G$, and these paths do not contain the terminals, there must be at least $\lfloor k_1/6\rfloor $ direct edge-disjoint paths connecting the vertices of $\bigcup_{v_i\in A}S_i$ to the vertices of $\bigcup_{v_i\in B}S_i$. Since $|A|+|B|=\gamma$, at least one pair $(S_i,S_{i'})$ with $v_i\in A$, $v_{i'}\in B$, has at least $k'$ direct edge-disjoint paths connecting them.

Let $T$ be any spanning tree of $Z$, and let $\gamma'=\gamma^{1/4}$. 
For every pair $v_i,v_j$ of vertices in graph $Z$, let $n(v_i,v_j)$ be the maximum number of direct edge-disjoint paths connecting $S_i$ to $S_j$ in $G$. For every edge $e=(v_i,v_j)$ of the tree $T$, let $w(e)=n(v_i,v_j)$, and let $w(T)=\sum_{e\in E(T)}w(e)$. Root the tree $T$ at any vertex.
We now perform a number of iterations, that, informally, aim to maximize the number of leaves and the value $w(T)$ of the tree $T$. In every iteration, one of the following two improvement steps is performed, whenever possible.

\paragraph{Improvement Step 1:} 
Let $v_j$ be any vertex of the tree $T$, whose parent $v_{j'}$ is a non-root degree-$2$ vertex. If there is any non-leaf vertex $v_{j''}$, such that $v_{j''}$ is not the descendant of $v_j$ in the tree $T$, and the edge $(v_j,v_{j''})$ belongs to the graph $Z$, then we delete the edge $(v_j,v_{j'})$ from the tree $T$, and add the edge $(v_j,v_{j''})$ to it. Notice that this operation increases the number of leaves in $T$ by $1$. 

\paragraph{Improvement Step 2:}
Let $v_j$ be any non-root degree-$2$ vertex of the tree, whose parent $v_{j'}$ is also a degree-$2$ non-root vertex. Let $e=(v_j,v_{j'})$ be the edge of tree $T$ connecting $v_j$ to its father. Consider the two connected components of $T\setminus e$, and let $e'\in E(Z)$ be the edge connecting these two components, with the maximum value $w(e')$. If $w(e')>w(e)$, then we delete $e$ from $T$ and add $e'$ to it. Notice that this operation does not decrease the number of leaves in $T$, and it increases $w(T)$ by at least $1$.

We perform one of the above two improvement steps, while possible. Since each such improvement step either increases the number of leaves in $T$, or leaves the number of leaves unchanged but increases $w(T)$ by at least $1$, after polynomially many improvement steps, we obtain a final tree $T$, on which none of these improvement steps can be executed. Let $L$ denote the set of leaves of $T$.

We say that Case 1 happens, if $T$ contains at least $\gamma'$ leaves. Otherwise, we say that Case 2 happens.
We consider each of the two cases separately below. Since Case 2 is easier to analyze, we consider it first.

\paragraph{Case 2}
A path $P$ in the tree $T$ is called a $2$-path iff it does not contain the root of $T$, and all its vertices have degree $2$ in $T$. A $2$-path $P$ is maximal iff it is not a strict sub-path of any other $2$-path. Since the number of leaves in $T$ is less than $\gamma'=\gamma^{1/4}$, the number of maximal $2$-paths in $T$ is at most $2\gamma'$. Therefore, $T$ contains at least one $2$-path of length at least $\gamma'$. We let $P'$ be such a $2$-path, that minimizes the size of the sub-tree rooted at the last (lowest) vertex of $P'$. Recall that $r=8\gkrv\leq \lceil \gamma'/8\rceil$. Let $P$ be the sub-path of $P'$, consisting of the last $r$ vertices of $P'$. Assume w.l.o.g. that $P=(v_1,\ldots,v_r)$, where  $v_r$ is the vertex that lies deepest in the tree $T$. 
We set $\rset'=\set{S_1,\ldots,S_r}$, and we will build a tree $\tT$ over the set $v_1,\ldots,v_r$ of vertices as required. In fact in this case, the tree $\tT$ will be a path, $(v_{S_1},\ldots,v_{S_r})$. Therefore, we only need to define, for each edge $e=(v_{S_i},v_{S_{i+1}})$ on this path, the corresponding set $\pset_e$ of paths in graph $G'$.

Since the number of leaves in the original tree $T$ is at most $\gamma'$, and every $2$-path in the sub-tree of $v_r$ has length at most $\gamma'$  (by the choice of $P'$), the size of the sub-tree rooted at $v_r$ is bounded by $3\gamma'^2$.

We need the following claim:

\begin{claim}\label{claim: properties of path P}
For each $1\leq i<r$, $n(v_i,v_{i+1})\geq \frac{k_1}{24\gamma^{3/2}}$, and for each $1\leq i,i'\leq r$ with $|i-i'|>1$, $n(v_i,v_{i'})< \frac{k_1}{6\gamma^2}$.
\end{claim}

\begin{proof}
Fix some $1\leq i< r$, and consider the edge $e=(v_i,v_{i+1})$. Notice that the size of the sub-tree rooted at $v_{i+1}$ is bounded by $r+3\gamma'^2\leq 4\gamma'^2$, since the size of the sub-tree rooted at $v_r$ is bounded by $3\gamma'^2$. Let $T_1,T_2$ be the two connected components of $T\setminus\set{e}$. Recall that there are at least $\lfloor k_1/6\rfloor $ direct paths connecting the sets corresponding to vertices of $T_1$ to the sets corresponding to vertices of $T_2$. So there is at least one pair $v\in T_1$, $v'\in T_2$ of vertices, with $n(v,v')\geq \frac{k_1}{24\gamma\cdot \gamma'^2}=\frac{k_1}{24\gamma^{3/2}}$. Since we could not perform the second improvement step, the weight of the edge $(v_i,v_{i+1})$ must be at least $ \frac{k_1}{24\gamma^{3/2}}$.

In order to prove the second assertion, consider some pair $v_i,v_{i'}$ of vertices on path $P$, where $|i-i'|>1$, and assume w.l.o.g. that $v_i$ is a descendant of $v_{i'}$. Since we could not execute the first improvement step for vertex $v_i$, there is no edge $(v_i,v_{i'})$ in graph $Z$. Therefore, $n(v_i,v_{i'})< \frac{k_1}{6\gamma^2}$.
\end{proof}

For each $1\leq i<r$, let $\tpset_i$ be the set of $\lceil \frac{k_1}{24\gamma^{3/2}}\rceil$ direct edge-disjoint paths, connecting $S_i$ to $S_{i+1}$. 
While the paths in each set $\tpset_i$ are edge-disjoint, the total congestion caused by the set $\bigcup_i\tpset_i$ of paths may be as high as $(r-1)$. In order to overcome this difficulty, we perform edge-splitting on an auxiliary graph $H$, constructed as follows. Start with the set of all the vertices and edges lying on the paths in $\bigcup_{i=1}^{r-1}\tpset_i$. Next, for each $1\leq i\leq r$, we contract the vertices of $S_i$, that belong to the current graph, into a single vertex $v_i$. Finally, we replace every edge by a pair of bi-directed edges. Let $H$ be the final directed graph. Observe that $H$ is a directed Eulerian graph, that has the following properties:

\begin{itemize}

\item Graph $H$ does not contain any terminals, or vertices from the sets $S\in\rset\setminus \rset'$.

\item For each $1\leq i<r$, there are at least $\lceil \frac{k_1}{24\gamma^{3/2}}\rceil$ edge-disjoint paths connecting 
$v_i$ to $v_{i+1}$, that do not contain any vertices in $\set{v_1,v_2,\ldots,v_r}$ as inner vertices.


\item For all $1\leq i\neq i'\leq r$, with $|i-i'|>1$, there are at most $\frac{k_1}{6\gamma^2}$ paths connecting $v_i$ to $v_{i'}$, that 
do not contain the vertices of $\set{v_1,v_2,\ldots,v_r}$ as intermediate vertices.
\end{itemize}
 
We use the following theorem to perform edge-splitting in graph $H$.
Let $D=(V,A)$ be any directed multigraph with no self-loops. For any pair $(v,v')\in V$ of vertices, their connectivity $\lambda(v,v';D)$ is the maximum number of edge-disjoint paths connecting $v$ to $v'$ in $D$. Given a pair $a=(u,v)$, $b=(v,w)$ of edges, a splitting-off procedure replaces the two edges $a,b$ by a single edge $(u,w)$. We denote by $D^{a,b}$ the resulting graph. We use the extension of Mader's theorem~\cite{Mader} to directed graphs, due to Frank~\cite{Frank} and Jackson~\cite{Jackson}. Following is a simplified version of Theorem 3 from~\cite{Jackson}:

\begin{theorem}\label{thm: splitting-off}
Let $D=(V,A)$ be an Eulerian digraph, $v\in V$ and $a=(v,u)\in A$.
Then there is an edge $b=(w,v)\in A$, such that for all $y,y'\in V\setminus\set{v}$:
$\lambda(y,y';D)=\lambda(y,y';D^{ab})$.
\end{theorem}

 We iteratively perform edge-splitting in graph $H$ using Theorem~\ref{thm: splitting-off}, by repeatedly choosing vertices $v\not\in\set{v_1\ldots,v_r}$, until all such vertices become isolated. Let $H'$ denote the resulting graph, after we discard all isolated vertices and make all edges undirected. Then $V(H')=\set{v_1,\ldots,v_r}$, and for each $1\leq i\neq i'\leq r$, each edge $e=(v_i,v_{i'})$ corresponds to a {\bf direct} path $P_e$ in graph $G'$, connecting some vertex of $S_i$ to some vertex of $S_{i'}$ (that is, $P_e$ does not contain any vertices of $\left (\bigcup_{S\in \rset}S\right )\cup \tset$ as inner vertices). Moreover, the set $\set{P_e\mid e\in E(H')}$ of paths causes congestion at most $1$ in $G'$, and the edges $e'\in E(G)$ whose both endpoints belong to the same set $S_i$ do not lie on any such path $P_e$.

 We need the following claim.
 
 \begin{claim}\label{claim: path structure preserved after splitting}
 For each $1\leq i, i'\leq r$ with $|i-i'|>1$, there are at most $\lceil \frac{k_1}{3\gamma^2}\rceil$ parallel edges $(v_i, v_{i'})$ in $H'$, and for each $1\leq i< r$, there are at least $\frac{k_1}{48\gamma^{3/2}}$ parallel edges $(v_i,v_{i+1})$ in $H'$.
 \end{claim}
 
\begin{proof}
Assume for contradiction that the first assertion is false, and let $v_i,v_{i'}$ be a pair of vertices in $H'$ with $1\leq i,i'\leq r$ and $|i-i'|>1$, such that there are more than $\lceil \frac{k_1}{3\gamma^2}\rceil $ parallel edges $(v_i,v_{i'})$ in $H'$. Let $\pset'$ be the set of paths in graph $G$, containing, for each such parallel edge $e$, the corresponding path $P_e$. Then $\pset'$ contains more than $\lceil \frac{k_1}{3\gamma^2}\rceil$ direct paths, connecting the vertices of $S_{i}$ to the vertices of $S_{i'}$ in $G$. Moreover, the paths in $\pset'$ cause congestion at most $2$ in $G$. Therefore, there is a flow $F$ between $S_i$ and $S_{i'}$ of value at least $\half\lceil \frac{k_1}{3\gamma^2}\rceil+\half$, with no congestion in graph $G$, where all flow-paths are direct paths. From the integrality of flow, there is a set of at least $\frac{k_1}{6\gamma^2}$ direct edge-disjoint paths connecting $S_{i}$ to $S_{i'}$ in $G$, contradicting Claim~\ref{claim: properties of path P}.

We now turn to prove the second assertion. Assume otherwise, and let $1\leq i<r$ be some index, such that there are fewer than $\frac{k_1}{48\gamma^{3/2}}$ parallel edges $(v_i,v_{i+1})$ in graph $H'$. Consider the following partition $(A,B)$ of $V(H')$: $A=\set{v_1,\ldots,v_{i}}$, $B=\set{v_{i+1},\ldots,v_r}$.
 Recall that graph $H$ contained at least $\frac{k_1}{24\gamma^{3/2}}$ edge-disjoint paths connecting $v_i$ to $v_{i+1}$. Since the edge splitting operations preserve the connectivity between $v_i$ and $v_{i'}$, graph $H'$ must contain at least $\frac{k_1}{24\gamma^{3/2}}$ edge-disjoint paths connecting $v_i$ to $v_{i+1}$. In particular, $|E_{H'}(A,B)|\geq \frac{k_1}{24\gamma^{3/2}}$. Let $E'\sse E_{H'}(A,B)$ contain all edges  between vertex sets $A$ and $B$, except for the set of parallel edges $(v_i,v_{i+1})$. Since there are fewer than $\frac{k_1}{48\gamma^{3/2}}$ such parallel edges, $|E'|> \frac{k_1}{48\gamma^{3/2}}$. 
 
Notice however that there are at most $\frac{r^2} 4-1$ pairs of vertices $(v_j,v_j')\in A\times B$, where $(v_j,v_{j'})\neq (v_i,v_{i+1})$. Each such pair contributes at most $\lceil \frac{k_1}{3\gamma^2}\rceil $ edges to $E'$, from the first assertion.
Therefore,

\[|E'|\leq \left (\frac{r^2} 4-1\right )\cdot \lceil \frac{k_1}{3\gamma^2}\rceil\leq \frac{r^2\cdot k_1}{12\gamma^2}< \frac{k_1}{48\gamma^{3/2}},\]
 
 as $r\leq\lceil\frac{\gamma'} 8\rceil\leq\frac{\gamma'}{4}=\frac{\gamma^{1/4}}{4}$, a contradiction.
\end{proof}

Our final tree $\tT$ is simply a path connecting the vertices $(v_{S_1},v_{S_2},\ldots,v_{S_r})$ in this order. For each edge $e=(v_{S_i},v_{S_{i+1}})$ on this path, we let $\pset_e$ be the set of paths in graph $G'$ corresponding to the set of parallel edges connecting $v_i$ to $v_{i+1}$ in graph $H'$. We discard paths from each set $\pset_e$ until $|\pset_e|=k_2$. It is immediate to verify that $\bigcup_{e\in E(\tT)}\pset_e$ cause congestion at most $1$ in $G'$, and moreover, each such path does not contain the vertices of $\bigcup_{S\in \rset'}S$ as inner vertices.

\paragraph{Case 1}

Recall that in this case, $T$ contains at least $\gamma'$ leaves. Let $\rset'\sse \rset$ be any set of $r<\gamma'$ vertex subsets, corresponding to the leaves in tree $T$. For simplicity of notation, we will assume that $\rset'=\set{S_1,\ldots,S_r}$. Our first step is to select, for each $S\in \rset'$, a subset $\Gamma(S)\sse \out(S)$ of $k_3=\Omega(k'\alpha\cdot \alphaWL/r^3)$ edges, such that for each $S,S'\in \rset'$, there is a set $\pset(S,S'):\Gamma(S)\sconnect_1\Gamma(S')$ of paths in graph $G$. We do so in the following claim.

\begin{claim}
There is an efficient algorithm, that either computes a $(k_1,\alpha)$-violating partition for some set $S\in \rset$, or computes, for each set $S\in \rset'$, a subset $\Gamma(S)\sse \out(S)$ of $k_3=\Omega(k'\alpha\cdot \alphaWL/r^3)$ edges, such that for each $S,S'\in \rset'$, there is a set $\pset(S,S'):\Gamma(S)\sconnect_1\Gamma(S')$ of paths in graph $G$, that do not contain the vertices of $\bigcup_{S''\in \rset'}S''$ as inner vertices.
\end{claim}
\begin{proof}
Let  $v_{j^*}$ be the root of the tree $T$, and let $S_{j^*}\in \rset$ be the corresponding vertex subset. We claim that for each set $S_j\in \rset'$, there is a flow $F_j$ of value $k'$, connecting the vertices of $S_{j^*}$ to the vertices of $S_j$ in $G$, with congestion at most $\frac{1}{\alpha} + \gamma\leq \frac{2}{\alpha}$, such that the paths in $\pset_j$ do not contain any vertices of $\bigcup_{S\in \rset'}S$ as inner vertices. Indeed, consider the path $(v_j^*=v_{i_1},v_{i_2},\ldots,v_{i_h}=v_j)$ in the tree $T$, connecting $v_{j^*}$ to $v_j$. 
For each edge $e_z=(v_{i_z},v_{i_z+1})$ on this path, there is a set $\qset_z$ of $k'$ direct edge-disjoint paths connecting the vertices of $S_{i_z}$ to the vertices of $S_{i_{z+1}}$, from the definition of the graph $Z$. Denote 

\[\Gamma^2_z=\set{e'\in \out(S_{i_z})\mid e'\mbox{ is the first edge on some path in }\qset_z},\]

 and similarly 
 
 \[\Gamma^1_{z+1}=\set{e'\in \out(S_{i_{z+1}})\mid e'\mbox{ is the last edge on some path in }\qset_z}.\]
 
  For each $1<z<h$, we now have $\Gamma^1_z,\Gamma^2_z\sse \out(S_{i_z})$, two subsets of edges of size $k'<k_1/2$. From our assumption, there is a flow $F'_z:\Gamma^1_z\sconnect_{1/\alpha}\Gamma^2_z$ inside set $S_z$. Concatenating the flows $(\qset_1,F'_1,\qset_2,F'_2,\ldots,F'_{h-1},\qset_{h-1})$, we obtain the desired flow $F_j$ of value $k'$. The total congestion caused by paths in $\bigcup_{z=1}^{h-1}\qset_z$ is at most $\gamma$ (since $h\leq \gamma$, and the paths in each set $\qset_z$ are edge-disjoint), while each flow $F'_z$ causes congestion at most $1/\alpha$ inside the graph $G[S_{i_z}]$. Therefore, the total congestion due to flow $F_j$ is bounded by $\frac{1}{\alpha}+\gamma\leq \frac{2}{\alpha}$.

Scaling all flows $F_j$, for $S_j\in \rset'$ down by factor $2r/\alpha$, we obtain a new flow $F$, where every set $S_j\in \rset'$ sends $\frac{k'\alpha}{2r}$ flow units to $S_{j^*}$, and the total congestion due to $F$ is at most $1$. From the integrality of flow, there is a collection $\set{\pset_j}_{j=1}^r$ of path sets, where for each  $1\leq j\leq r$, set $\pset_j$ contains $\lfloor \frac{k'\alpha}{2r}\rfloor$ paths connecting $S_j$ to $S_{j^*}$, the paths in $\bigcup_{j=1}^r\pset_j$ are edge-disjoint, and they do not contain the vertices of $\bigcup_{j=1}^rS_j$ as inner vertices. We will also assume w.l.o.g. that the paths in $\bigcup_{j=1}^r\pset_j$ do not contain the vertices of $S_{j^*}$ as inner vertices.

For each $1\leq j\leq r$, let $E_j\sse \out(S_{j^*})$ be the set of edges $e$, such that $e$ is the last edge on some path in $\pset_j$, and let $E'=\bigcup_{j=1}^rE_j$. Notice that $|E'|=r\cdot \lfloor \frac{k'\alpha}{2r}\rfloor<k_1$. We set up an instance of the sparsest cut problem as follows. First, partition every edge $e\in E'$ by a vertex $t_e$, and let $\ttset$ be the set of the resulting new vertices. Consider the sub-graph $G^*$ of the resulting graph induced by $S_{j^*}\cup \ttset$, with the set $\ttset$ of terminals. We apply the algorithm \algsc to the resulting instance of the sparsest cut problem. If the algorithm returns a cut whose sparsity is less than $\alpha$, then we have found a $(k_1,\alpha)$-violating partition of $S_{j^*}$. We return this partition, and terminate the algorithm. Otherwise, we are guaranteed that $S_{j^*}$ is $\alphawl$-well-linked for $E'$. 

We then apply Theorem~\ref{thm: grouping-many-sets} to the same graph $G^*$, with the set $\ttset$ of terminals, where the§ terminal subsets are defined to be $\tset_j=\set{t_e\mid e\in E_j}$ for $1\leq j\leq r$.
Notice that for all $j$, $|\tset_j|=k/\poly\log k$, while $r=\poly\log k$ and $\alphaWL=1/\poly\log k$, so the conditions of Theorem~\ref{thm: grouping-many-sets} hold.
From Theorem~\ref{thm: grouping-many-sets}, we obtain, for each $1\leq j\leq r$, a subset $\tset_j'\sse \tset_j$ of $\Omega\left (\frac{k'\cdot \alpha\cdot\alphaWL}{r^3}\right )=k_3$ terminals, such that $G^*$ is $1$-well-linked for $\bigcup_j\tset_j'$. Let $E'_j\sse E_j$ be the set of edges corresponding to the vertices in $\tset_j'$, let $\pset'_j\sse \pset_j$ be the set of paths terminating at the edges of $E'_j$, and let $\Gamma(S_j)\sse \out(S_j)$ be the set of edges where the paths in $\pset'_j$ start, for each $1\leq j\leq r$. Then for each $1\leq j\leq r$, $|\Gamma(S_j)|=|\tset_j'|=k_3$. Moreover, we claim that for all $1\leq j,j'\leq r$, there is a set $\pset(S_j,S_{j'}):\Gamma(S_j)\sconnect_1\Gamma(S_{j'})$ of paths in graph $G$, that do not contain the vertices of $\bigcup_{S\in \rset'}S$ as inner vertices.  
Indeed, since $S_{j^*}$ is $1$-well-linked for $\bigcup_iE'_i$, there is a set $\pset_{j,j'}:E_j\sconnect_1E_{j'}$ of paths contained in $S_{j^*}$. The set $\pset(S_j,S_{j'})$ of paths is obtained by concatenating $\pset_j',\pset_{j,j'}$ and $\pset_{j'}'$.
\end{proof}

Our next step is to perform edge splitting, similarly to Case 2. We build an auxiliary graph $H$ from graph $G$, as follows. We delete from $G$ all terminals, and for each $1\leq j\leq r$, we delete all edges in $\out(S_j)\setminus \Gamma(S_j)$. Next, we contract each set $S_j$ into a super-node $v_j$, for all $1\leq j\leq r$. Finally, we replace every edge in the resulting graph by a pair of bi-directed edges. Let $H$ be this final graph. Notice that in-degree and the out-degree of every super-node $v_j$, for $1\leq j\leq r$, in graph $H$ is exactly $k_3$, and for every pair $v_j,v_j'$ of such super-nodes, there are $k_3$ edge-disjoint paths connecting $v_j$ to $v_{j'}$, and $k_3$ edge-disjoint paths connecting $v_{j'}$ to $v_j$ in $H$. 
 We iteratively perform edge-splitting in graph $H$ using Theorem~\ref{thm: splitting-off}, by repeatedly choosing vertices $v\not\in\set{v_1\ldots,v_r}$, until all such vertices become isolated. Let $H'$ denote the resulting graph, after we discard all isolated vertices. Then $V(H')=\set{v_1,\ldots,v_r}$, and for each $1\leq i\neq i'\leq r$, each edge $e=(v_i,v_{i'})$ corresponds to a path $P_e$ in graph $G'$, connecting some vertex of $S_i$ to some vertex of $S_{i'}$, and $P_e$ does not contain any vertices of $\bigcup_{S\in \rset'}S$ as inner vertices  (recall that $G'$ is the directed graph obtained from $G$ by replacing each edge with a pair of bi-directed edges). Moreover, the set $\set{P_e\mid e\in E(H')}$ of paths causes congestion at most $1$ in $G'$.
For each vertex $v_j$ in graph $H'$, the in-degree and the out-degree of $v_j$ are equal to $k_3$, and for each pair $v_j,v_{j'}$ of vertices, there are exactly $k_3$ edge-disjoint paths connecting $v_j$ to $v_{j'}$ in $H'$, and $k_3$ edge-disjoint paths connecting $v_{j'}$ to $v_j$.

We now construct a graph $Z'$, whose vertex set is $\set{v_1,\ldots,v_r}$, and there is an edge $(v_i,v_j)$ in graph $Z'$ iff the total number of parallel edges $(v_i,v_j)$ and $(v_j,v_i)$ in graph $H'$ is at least $2k_3/r^3$. 
The proof of the next claim is identical to the proof of Claim 3 in~\cite{EDP-old} and is omitted here.

\begin{claim}
There is an efficient algorithm to find a spanning tree $\tilde{T}$ of maximum vertex degree at most $3$ in graph $Z'$.
\end{claim}

We output the tree $\tilde{T}$ as our final tree. Notice that each edge $e=(v_i,v_j)\in E(\tilde{T})$ corresponds to a collection $\pset_e$ of paths, where each path connects either a vertex of $v_i$ to a vertex of $v_j$, or vice versa in graph $G'$ (these are the paths that correspond to the edges of $H'$ in the original graph $G'$, obtained by the edge splitting procedure). Moreover, the paths in $\bigcup_{e\in E(\tilde{T})}\pset_e$ cause congestion at most $1$ in $G'$, and these paths do not contain the vertices of $\bigcup_{j=1}^rS_j$ as inner vertices. The number of paths in each set is:

\[ |\pset_e|\geq \frac{2k_3}{r^3}=\Omega\left(\frac{k'\alpha\cdot\alphaWL}{r^6}\right )=\Omega\left(\frac{k_1\alpha\cdot\alphaWL}{\gamma^2 r^6}\right )=\Omega\left (\frac{k_1\alpha\cdot\alphaWL}{\gamma^{3.5}}\right )=k_2.\]
\end{proof}

\subsection{Step 2: connecting the terminals}\label{subsec: step 2: connecting the terminals}
We assume w.l.o.g. that $\rset'=\set{S_1,\ldots,S_r}$.
In this step we connect a subset of terminals to two subsets $S,S'\in \set{S_1,\ldots,S_r}$, using the following theorem.

\begin{theorem}\label{thm: step 2 - connecting the terminals}
There is an efficient algorithm, that either finds a subset $\mset'\sse \mset$ of $k/\poly\log k$ demand pairs, and a routing of the pairs in $\mset'$ via edge-disjoint paths in $G$, or finds a subset $\mset_1\sse \mset$ of $k_4=\Omega(k_2/r^2)$ demand pairs, two subsets $\pset_1',\pset_1''$ of paths in $G$, and for each edge $e\in E(\tT)$, a subset $\pset'_e\sse \pset_e$ of $\lfloor k_2/2\rfloor $ paths, such that:

\begin{itemize}
\item We are given a partition of $\tset(\mset_1)$ into two subsets $\tset_1',\tset_1''$, where for each pair $(s,t)\in \mset_1$, $s\in \tset_1',t\in \tset_1''$, or vice versa.

\item There are two sets $S,S'\in \rset'$ (with possibly $S=S'$), such that $\pset_1': \tset_1'\connect \out(S)$, $\pset_1'':\tset_1''\connect \out(S')$.

\item The paths in $\pset_1'\cup \pset_1''\cup\left (\bigcup_{e\in E(\tilde{T})}\pset_e'\right )$ do not contain any vertices of $\bigcup_{j=1}^rS_j$ as inner vertices, and they cause congestion at most $2$ in $G$.
\end{itemize}
\end{theorem}
\begin{proof}
 
We start with the following lemma, whose proof uses standard techniques.

\begin{lemma}\label{lem: connecting the terminals}
Let $S\in \rset$ be any set. Then there is an efficient algorithm, that either finds a subset $\mset'\sse \mset$ of $k/\poly\log k$ demand pairs and a routing of $\mset'$ via edge-disjoint paths in $G$, or finds a subset $\mset_0\sse \mset$ of $\Omega(k_1)$ pairs, and a set $\pset:\tset(\mset_0)\connect_1\out(S)$ of paths in graph $G$, connecting every terminal in $\tset(\mset_0)$ to some edge of $\out(S)$ via edge-disjoint paths.
\end{lemma}

\begin{proof}
Recall that we have assumed that there is a flow $F$ of value $k_1/2$ and no congestion in graph $G$ between $\out(S)$ and $\tset$. From the integrality of flow, there is a set $\pset^*$ of $\lfloor k_1/2\rfloor$ edge-disjoint paths connecting the terminals in $\tset$ to the edges in $\out(S)$. Recall that the degree of every terminal in $G$ is $1$, so every terminal is an endpoint of at most one path in $\pset^*$. Let $\tset_0\sse \tset$ be the subset of $\lfloor k_1/2\rfloor$ terminals that serve as endpoints of paths in $\pset^*$. We partition the remaining terminals into $h=\lceil \frac{k}{\lfloor k_1/2\rfloor}\rceil=O(\gamma^3\log\gamma)$ subsets of at most $\lfloor k/2\rfloor $ terminals each. Since the set $\tset$ of terminals is $1$-well-linked, for each $1\leq j\leq h$, there is a flow $\fset_j: \tset_j\connect_1\tset_0$. Concatenating flow $F_j$ with the paths in $\pset^*$, and taking the union of all resulting flows over all $1\leq j\leq h$, we obtain a flow $F'$, where every terminal in $\tset$ sends one flow unit to some edge in $\out(S)$, and the congestion caused by $F'$ is at most $2h=O(\gamma^3\log\gamma)$.

Our next step is to perform a grouping of the terminals, using Theorem~\ref{thm: grouping}, with the parameter $q=20h$. Let $\gset$ be the resulting partition of the terminals. Recall that each set $U\in \gset$ contains at least $q$ and at most $3q$ terminals, and it is associated with a tree $T_U$, containing the terminals of $U$, such that the trees $\set{T_U}_{U\in \gset}$ are edge-disjoint.

We partition $\mset$ into two subsets: $\mset_1'$ containing all pairs $(s,t)$ where both $s$ and $t$ belong to the same group $U\in \gset$, and $\mset_2'$ containing all remaining pairs.

Assume first that $|\mset_1'|\geq |\mset|/2$. We then select a subset $\mset'\sse \mset_1'$ of demand pairs as follows: for each group $U\in \gset$, we select one arbitrary pair $(s,t)\in \mset_1'$ with $s,t\in U$, if such a pair exists, and add it to $\mset'$. Notice that since every group contains at most $60h$ terminals, the number of pairs 
in $\mset'$ is at least $\frac{k}{120h}=\Omega\left (\frac{k}{\gamma^3\log\gamma}\right )=\Omega \left (\frac{k}{\poly\log k}\right )$. Every pair $(s,t)\in \mset'_1$ can be connected by some path contained in the tree $T_U$, where $s,t\in U$. Since the trees $\set{T_U}_{U\in \gset}$ are edge-disjoint, the resulting routing of the pairs in $\mset'$ is via edge-disjoint paths.
We assume from now on that $|\mset_2'|\geq |\mset|/2$. We need the following simple claim.

\begin{claim} 
We can efficiently find a subset $\mset''\sse \mset'_2$ of at least $k_1/960$ demand pairs, such that there is a flow $F^*:\tset(\mset'')\connect_{1.1}\out(S)$ in $G$.
\end{claim}
\begin{proof}
We start with $\mset''=\emptyset$. While $\mset_2'\neq \emptyset$, select any pair $(s,t)\in \mset'_2$ and add it to $\mset''$. Assume that $s\in U$, $t\in U'$. Remove from $\mset_2'$ all pairs $(s',t')$, where either $s'$ or $t'$ belong to $U\cup U'$. Notice that for each pair added to $\mset''$, at most $6q=120h$ pairs are deleted from $\mset_2'$. Therefore, at the end of this procedure, $|\mset''|\geq \frac{k}{240h}\geq\frac{k}{240}\cdot\frac{k_1}{4k}=\frac{k_1}{960}$.

We now show that there is a flow $F^*:\tset(\mset'')\connect_{1.1}\out(S_1)$ in $G$. Let $F''$ be the flow $F'$, scaled down by factor $20h$. Then every terminal in $\tset$ sends $1/(20h)$ flow units to the edges of $\out(S)$, and the total congestion caused by $F''$ is at most $1/10$.
In order to define the flow $F^*$, consider some terminal $t\in \tset(\mset'')$, and let $U_t\in \gset$ be the group to which it belongs. Let $U'_t\sse U_t$ be any subset of $20h$ terminals in $U_t$. Terminal $t$ then sends $1/(20h)$ flow units to each terminal in $U_t'$, inside the tree $T_{U_t}$. This flow is then concatenated with the flow that leaves each terminal in $U'_t$ in $F''$. Taking the union of this flow for each $t\in \tset(\mset'')$, we obtain a flow $F^*$, where every terminal in $\tset(\mset'')$ sends one flow unit to the edges of $\out(S)$. In order to bound the edge congestion, we observe that flow $F''$ contributes at most $1/10$ congestion to every edge. Since the connected components $\set{T_{U_t}}_{t\in \tset(\mset'')}$ are edge-disjoint, the flow inside these components contributes at most $1$ to the total edge congestion. Overall, the congestion caused by the flow $F^*$ is at most $1.1$.
\end{proof}

From the integrality of flow, there is a collection $\pset$ of $\lfloor \frac{2|\mset''|}{1.1}\rfloor\geq 1.8 |\mset''|$ of edge-disjoint paths connecting the terminals in $\tset(\mset'')$ to the edges of $\out(S)$. Let $\tset'\sse \tset(\mset'')$ be the set of terminals where these paths originate, and let $\mset_0\sse \mset''$ be the subset of the demand pairs contained in $\tset'$. Since $|\tset'|\geq 1.8|\mset''|$, we get that $|\mset_0|\geq |\mset''|/2$. Thus, we have obtained a set $\mset_0$ of $\Omega(k_1)$ demand pairs, such that all terminals in $\tset(\mset_0)$ can be connected to the edges of $\out(S)$ via edge-disjoint paths.
\end{proof}

In order to complete the proof of Theorem~\ref{thm: step 2 - connecting the terminals}, we start by applying Lemma~\ref{lem: connecting the terminals} to the set $S=S_1$. If the outcome is a subset $\mset'\sse \mset$ of $k/\poly\log k$ demand pairs and a routing of pairs in $\mset'$ via edge-disjoint paths in $G$, then we return $\mset'$ together with this routing, and terminate the algorithm.

We therefore assume from now on, that Lemma~\ref{lem: connecting the terminals} returns a set $\mset_0\sse \mset$ of $\Omega(k_1)$ pairs, together with the set $\pset:\tset(\mset_0)\connect_1\out(S_1)$ of edge-disjoint paths in graph $G$. We discard terminal pairs from $\mset_0$, and the corresponding paths from $\pset$ until $|\mset_0|=\lfloor k_2/4\rfloor$, so $|\pset|=|\tset(\mset_0)|\leq k_2/2$ holds.

For each path $P\in \pset$, we direct the path from the terminal towards $S_1$, and we truncate it at the first vertex that belongs to $\bigcup_{i=1}^rS_i$. The new collection $\pset$ of paths then connects every terminal in $\tset(\mset_0)$ to some vertex of $\bigcup_{i=1}^rS_i$. Moreover, the paths in $\pset$ do not contain the vertices of $\bigcup_{i=1}^rS_i$ as inner vertices.

 
 We now combine the two sets of paths, $\pset$ and $\bigcup_{e\in E(\tilde{T})}\pset_e$, to ensure that they cause congestion at most $1$ in $G'$ altogether, using the following simple claim.

\begin{claim}
We can efficiently find, for each edge $e\in E(\tT)$, a subset $\pset_e'\sse \pset_e$ of $\lfloor k_2/2\rfloor $ paths, and a collection $\pset':\tset(\mset_0)\connect_1\bigcup_{i=1}^rS_i$ of paths in $G'$, such that the paths in set $\left (\bigcup_{e\in E(\tilde{T})}\pset'_e\right )\cup \pset'$ cause congestion $1$ in $G'$.
\end{claim}
\begin{proof}
The proof is very similar to the arguments used by Conforti et al.~\cite{CHR}. We re-route the paths in $\pset$ via some paths in $\bigcup_{e\in E(\tT)}\pset_e$, by setting up an instance of the stable matching problem in a multi-graph. In this problem, we are given a complete bipartite multigraph $\tilde G=(A,B,E)$, where $|A|=|B|$. Each vertex $v\in A\cup B$ specifies an ordering $R_v$ of the edges adjacent to $v$ in $\tilde G$. A complete matching $M$ between the vertices of $A$ and $B$ is called stable iff, for every edge $e=(a,b)\in E\setminus M$, the following holds. Let $e_a,e_b$ be the edges adjacent to $a$ and $b$ respectively that belong to $M$. Then either $a$ prefers $e_a$ over $e$, or $b$ prefers $e_b$ over $e$. Conforti et al.~\cite{CHR}, generalizing the famous theorem of Gale and Shapley~\cite{Gale-Shapley}, show an efficient algorithm to find a complete stable matching $M$ in any such multigraph.

Let $\qset=\bigcup_{e\in E(\tilde{T})}\pset_e$.
We set up an instance of the stable matching problem as follows. Set $A$ contains a vertex $v(P)$ for each path $P\in \pset$, and set $B$ contains a vertex $v(Q)$ for each path $Q\in \qset$.
In order to ensure that $|A|=|B|$, we add dummy vertices to $A$ as needed.

For each pair $P\in \pset$, $Q\in \qset$ of paths, for each edge $e$ that these paths share, we add an edge $(v(P),v(Q))$, and we think of this edge as representing the edge $e$. We add additional dummy edges as needed to turn the graph into a complete bipartite graph.

Finally, we define preference lists for vertices in $A$ and $B$. For each vertex $v(P)\in A$, the edges incident to $v(P)$ are ordered according to the order in which they appear on the path $P$. The dummy edges incident to $v(P)$ are ordered arbitrarily at the end of the list. For each vertex $v(Q)\in B$, its adjacent edges are appear in the reverse order of their appearance on the path $Q$, with the dummy edges appearing in the end of the preference list. The preference lists of the dummy vertices are arbitrary.

Let $M$ be a stable matching in the resulting graph $\tilde G$. For each $e\in E(\tT)$, we let $\pset'_e\sse \pset_e$ be the set of paths $Q$, whose vertex $v(Q)$ is matched to a dummy vertex in $M$. Since $|\pset|\leq k_2/2\leq |\pset_e|/2$, this ensures that $|\pset_e'|\geq k_2/2$. 

In order to construct the set $\pset'$ of paths, consider some path $P\in \pset$. If $P$ participates in the matching $M$ via a dummy edge, then we add $P$ to $\pset'$. Otherwise, assume that $(v(P),v(Q))\in M$, and the edge $(v(P),v(Q))$ corresponds to some edge $e$ shared by both $P$ and $Q$. We then replace $P$ with a path $P'$ that consists of two subpaths: a segment of $P$ from the beginning of $P$ until the edge $e$, and a segment of $Q$ from the edge $e$ to the end of $Q$. Notice that $P'$ starts at the same terminal as $P$, and terminates at some vertex of $\bigcup_{i=1}^r S_i$. Path $P'$ is then added to set $\pset'$.

This finishes the definition of sets $\pset'$, $\set{\pset_e'}_{e\in E(\tT)}$ of paths. Observe that the paths in $\pset'$ connect every terminal in $\tset(\mset_0)$ to some vertex of $\bigcup_{j=1}^rS_j$, and they do not contain the vertices of $\bigcup_{j=1}^rS_j$ as inner vertices. It now only remains to show that the set $\left (\bigcup_{e\in E(\tilde{T})}\pset_e'\right )\cup \pset'$ of paths causes congestion $1$ in $G'$.

Clearly, the paths in $\bigcup_{e\in E(\tilde{T})}\pset_e'$ are edge-disjoint from each other in graph $G'$, since the paths in $\bigcup_{e\in E(\tilde{T})}\pset_e$ were edge-disjoint from each other.

Assume now for contradiction that some pair of paths $P',P''\in \pset'$ share an edge. Assume w.l.o.g. that $P'$ was obtained from concatenating a segment of some path $P_1\in \pset$ and some path $Q_1\in \qset$ via some edge $e_1$ (where possibly $Q_1$ is empty and $P'=P_1$), and similarly $P''$ was obtained from concatenating a segment of some path $P_2\in \pset$ and some path $Q_2\in \qset$ via some edge $e_2$. The paths $P',P''$ may share an edge only if the first segment of $P'$ shares an edge with the second segment of $P''$, or vice versa. Assume w.l.o.g. that it is the former, that is, there is an edge $e'$, that appears on $P_1$ before $e_1$, and $e'$ appears on $Q_2$ after $e_2$ (notice that it is impossible that $e'=e_1$ or $e'=e_2$, since $Q_1,Q_2$ cannot share an edge). But then both $P_1$ and $Q_2$ prefer the edge $e'$ to their current matching in $M$, contradicting the fact that $M$ is a stable matching.

Finally, if some pair of paths $P'\in \pset'$ and $P''\in \bigcup_{e\in E(\tilde{T})}\pset_e'$ share an edge, we reach a contradiction exactly as before.
\end{proof}

Let $\tset_0=\tset(\mset_0)$. Let $S\in \rset'$ be the set such that at least $|\tset_0|/r$ paths in $\pset'$ terminate at $S$. 
Let $\pset^1\sse \pset'$ be the subset of paths that terminate at $S$, and let $\tset^1\sse\tset_0$ be the subset of terminals where these paths originate. Let $\mset^1\sse \mset_0$ be the subset of pairs $(s,t)$ where both $s,t\in \tset^1$, and let $\mset^2\sse \mset_0$ be the subset of pairs $(s,t)$ where $s\in \tset^1$, $t\not\in \tset^1$.

Assume first that $|\mset^1|\geq |\tset^1|/4$. Then $\mset^1$ contains $\Omega\left(\frac{k_2}{r}\right )>k_4$ demand pairs, and we have a set $\pset^1:\tset(\mset^1)\connect S$ of paths connecting the terminals participating in pairs in $\mset^1$ to the vertices of $S$. 
We discard arbitrary pairs in $\mset^1$, until $|\mset^1|=k_4$ holds, and we discard the corresponding paths from $\pset^1$.
We then partition the resulting set $\tset(\mset^1)$ of terminals into two arbitrary subsets, $\tset_1',\tset_1''$, such that for each demand pair $(s,t)\in \mset^1$, exactly one terminal belongs to each of the subsets. This defines a corresponding partition $(\pset_1',\pset_1'')$ of the set $\pset^1$ of paths. We set $S=S'$ in this case.

From now on we assume that $|\mset^1|<|\tset^1|/4$, and so $|\mset^2|=\Omega\left(\frac{k_2}{r}\right )$.
Let $\tset^2=\set{t\mid (s,t)\in \mset^2, s\in \tset^1}$, and let $\pset^2\sse \pset'$ be the subset of paths originating at the terminals of $\tset^2$. Let $S'\in \rset'$ be the set where at least $|\tset^2|/r$ of the paths in $\pset^2$ terminate. We let $\pset_1''\sse \pset^2$ be the subset of paths terminating at the vertices of $S'$, and we let $\tset_1''\sse \tset^2$ be the subset of terminals where these paths originate. We let $\tset_1'=\set{s\mid (s,t)\in \mset^2,t\in \tset_1''}$, $\pset_1'\sse \pset^1$ the set of paths originating at the terminals of $\tset_1'$, and $\mset_1\sse \mset^2$ the set of pairs $(s,t)$ with $s\in \tset_1'$, $t\in \tset_1''$. Notice that $|\mset_1|=\Omega\left(\frac{k_2}{r^2}\right )=k_4$ as required.
\end{proof}

Let $v$ and $v'$ be the vertices corresponding to the sets $S$ and $S'$ respectively in the tree $\tT$. Assume first that $v\neq v'$. Since the maximum vertex degree in $\tT$ is at most $3$, if we remove $v$ and $v'$ from $\tT$, we obtain at most $5$ sub-trees. Let $T_0$ denote the sub-tree incident on both $v$ and $v'$, $T_1,T_2$ the sub-trees incident on $v$, and $T_3,T_4$ the sub-trees incident on $v'$ (possibly some of the sub-trees are empty). Let $T\in \set{T_1,\ldots,T_4}$ be the sub-tree containing the most vertices, and assume w.l.o.g. that $T=T_1$, so it is incident on $v$. Let $T^*$ be the tree obtained from $\tT$, by deleting the sub-trees $T_2,T_3,T_4$ from it. It is easy to see that $T^*$ must contain at least $r/4\geq 2\gKRV$ vertices, and it contains both vertices $v$ and $v'$. Moreover, the degree of $v$ is at most $2$, and the degree of $v'$ is $1$ in $T^*$. We will view $v$ as the root of $T^*$.

If $v=v'$, then we root the tree $\tT$ at vertex $v$, and consider the sub-trees rooted at the children of $v$. We delete all but one of these sub-trees, leaving the sub-tree containing the most vertices. We denote by $T^*$ the resulting tree. In this case, the degree of $v$ is $1$, and $T^*$ again contains at least $2\gKRV$ vertices.

Let $\rset''\sse \rset'$ be the collection of sets $S_j$ whose corresponding vertex $v_j\in V(T^*)$, so $|\rset''|\geq 2\gKRV$.

For each terminal $t\in \tset(\mset_1)$, let $P_t\in \pset_1'\cup \pset_1''$ be the unique path originating at $t$.
Recall that every edge $e\in \out(S)\cup \out(S')$ may participate in up to two paths in $\pset_1'\cup \pset_1''$ (and in this case $S=S'$ must hold, and $e$ must be the last edge on these paths). It will be convenient to ensure that each such edge participates in at most one such path. In order to achieve this, we construct two new subsets $\mset_1',\mset_1''\sse \mset_1$ of demand pairs, as follows. Start with $\mset_1'=\mset_1''=\emptyset$. While $\mset_1\neq \emptyset$, select any pair $(s,t)\in \mset_1$, and delete it from $\mset_1$. Let $e$ be the last edge on $P_s$, and let $e'$ be the last edge on $P_t$. If $e=e'$, then we add pair $(s,t)$ to $\mset_1'$, and continue to the next iteration.
Otherwise, we add $(s,t)$ to $\mset_1''$, and we delete from $\mset_1$ every pair $(s',t')\in \mset_1$, where either  $P_{s'}$ or $P_{t'}$ contain either of the edges $e$ or $e'$.
If, at the end of this procedure, $|\mset_1'|\geq |\mset_1|/2$, then we can route every pair $(s,t)\in \mset_1'$ by concatenating the paths $P_s$ and $P_t$. This gives a routing of the paths in $\mset_1'$ with congestion at most $2$. From now on we assume that $|\mset_1'|<|\mset_1|/2$.
Observe that for each pair that we add to $\mset_1''$, we delete at most three pairs from $\mset_1$. Therefore,  $|\mset_1''|\geq |\mset_1|/6$. We delete demand pairs from $\mset_1''$, until $|\mset_1''|=\lfloor k_4/6\rfloor$ holds, and we delete from  $\tset_1',\tset_1''$ all terminals that do not participate in the pairs in $\mset_1''$, and we delete their corresponding paths from $\pset_1'$ and $\pset_1''$. For convenience, we will refer to $\mset_1''$ as $\mset_1$ from now on. Observe that now every edge in $\out(S)\cup \out(S')$ participates in at most one path in $\pset_1'\cup \pset_1''$.

To summarize this step, we now have a tree $T^*$ with maximum degree $3$, and a subset $\rset''\sse \rset$ of at least $2\gKRV$ vertex subsets, such that $V(T^*)=\set{v_j\mid S_j\in \rset''}$. Tree $T^*$ contains two special vertices, vertex $v$ corresponding to the set $S\in \rset''$ - the root of the tree, whose degree is at most $2$, and vertex $v'$ corresponding to the set $S'\in \rset''$, whose degree is $1$. It is possible that $S=S'$ and so $v=v'$. For each edge $e=(v_i,v_j)\in E(T^*)$, we are given a set $\pset'_e$ of $\lfloor k_2/2\rfloor$ paths in graph $G$, connecting the vertices of $S_i$ to the vertices of $S_j$. We are also given a subset $\mset_1$ of $\lfloor k_4/6\rfloor <\lfloor k_2/2\rfloor $ demand pairs, and two subset $\tset_1',\tset_1''$ of terminals, such that for each pair $(s,t)\in \mset_1$, $s\in \tset_1'$, $t\in \tset_1''$ or vice versa. Additionally, we have two sets $\pset_1':\tset_1'\connect \out(S)$, $\pset_1'':\tset_1''\connect \out(S')$ of paths. The paths in $\pset_1'\cup \pset_1''\cup \left(\bigcup_{e\in E({T^*})}\pset_e'\right )$ do not contain the vertices of $\bigcup_{S''\in \rset''}S''$ as inner vertices, and the total congestion they cause in graph $G$ is at most $2$. Each edge in $\out(S)\cup \out(S')$ participates in at most one path in $\pset_1'\cup \pset_1''$.

\subsection{Step 3: Building the Good Crossbar}
We assume w.l.o.g. that $\rset''=\set{S_1,\ldots,S_{r''}}$, where $r''\geq 2\gkrv$. We say that a path $P$ in graph $G$ is good iff it does not contain the vertices of $\tset\cup\left (\bigcup_{j=1}^{r''}S_j\right )$ as inner vertices. Consider some set $S_j\in \rset''$, and let $\pset$ be any collection of good paths in $G$. We denote by $\Gamma_j(\pset)\sse \out(S_j)$ the multi-set of edges, that appear as the first or the last edge on any path in $\pset$. That is,
 
 \[\Gamma_j(\pset)=\set{e\in \out(S_j)\mid e\mbox{ is the first or the last edge on some path $P\in\pset$}}\]
 
 Our first step is to select, for each edge $e\in E(T^*)$, a large subset $\pset''_e\sse \pset'_e$ of paths, and two large subsets $\pset'_2\sse \pset'_1,\pset_2''\sse \pset''_1$ of paths, such that, on the one hand, for each $1\leq j\leq r''$, set $S_j$ is $1$-well-linked for the subset $\Gamma_j(\pset')$ of paths, where $\pset'=\pset_2'\cup \pset_2''\cup\left (\bigcup_{e\in E(T^*)}\pset''_e\right )$, while on the other hand the endpoints of paths in $\pset'_2$ and $\pset_2''$ form source-sink pairs in set $\mset_1$.
This is done in the following theorem.

\begin{theorem}\label{thm: step 3: ensuring well-linkedness}
There is an efficient algorithm, that either computes a subset $\mset'\sse \mset_1$ of $k/\poly\log k$ demand pairs and their routing with congestion at most $2$ in $G$, or a $(k_1,\alpha)$-violating partition of some set $S_j\in \rset''$, or it computes, for each edge $e\in E(T^*)$, a subset $\pset''_e\sse \pset'_e$ of $k_5=\Omega(\alphaWL^2k_2)$ paths, and two subsets $\pset_2'\sse \pset_1',\pset_2''\sse \pset_1''$, such that:

\begin{itemize}
\item Each set $S_j\in \rset''$ is $1$-well-linked for the set $\Gamma_j(\pset')$ of edges, where $\pset'=\pset_2'\cup \pset_2''\cup\left (\bigcup_{e\in E(T^*)}\pset''_e\right )$.

\item There is a subset $\mset_2\sse \mset_1$ of $\Omega(\alphaWL^2)\cdot |\mset_1|=\Omega(\alphaWL^2\cdot k_4)$ demand pairs, and two subsets $\tset_2'\sse \tset_1',\tset_2''\sse \tset_1''$ of terminals, such that for each pair $(s,t)\in \mset_2$, $s\in \tset_2',t\in \tset_2''$ or vice versa. Moreover, $\pset_2'$ connects every terminal of $\tset_2'$ to some edge in $\out(S)$, and $\pset_2''$ connects every terminal of $\tset_2''$ to some edge in $\out(S')$.
\end{itemize}
\end{theorem} 
 
\begin{proof} 
 We start with $\pset=\pset'_1\cup \pset''_1\cup\left(\bigcup_{e\in E(T^*)}\pset'_e\right )$, and we consider the sets $S_j$ one-by-one. For each such set $S_j$, we compute the instance $SC(G,S_j,\Gamma_j)$ of the sparsest cut problem, where $\Gamma_j=\Gamma_j(\pset)$, and  run the algorithm $\algsc$ on it. Let $(A,B)$ be the output of the algorithm. If the sparsity of the cut $(A,B)$ is less than $\alpha$, then $(A,B)$ defines a $(k_1,\alpha)$-violating partition of $S_j$. We then stop the algorithm and return this partition. Otherwise, we are guaranteed that set $S_j$ is $\alpha/\alphasc=\alphaWL$-well-linked for $\Gamma_j(\pset)$.

From now on, we assume that for each $1\leq j\leq r''$, set $S_j$ is $\alphaWL$-well-linked for $\Gamma_j(\pset)$. 
Our next step is to boost the well-linkedness of every subset $S_j$ to $1$, by selecting subsets of paths in each set $\pset_e'$ (and in sets $\pset'_1,\pset''_1$). This is done by iteratively applying Theorems~\ref{thm: grouping-many-sets} and \ref{thm: grouping-advanced} to the sets $S_1,\ldots,S_{r''}$.

We start with the case where $v'=v$.
We process the vertices of the tree $T^*$ in the bottom-up fashion. Let $v_j$ be some vertex of $T^*$. We maintain the invariant that once vertex $v_j$ is processed, for every edge $e$ in the sub-tree of $v_j$, we have computed a subset $\pset''_e\sse \pset'_e$ of at least $k_5$ paths, such that the following holds. Let $\qset$ denote the union of all paths in the sets $\pset''_e$, where $e$ belongs to the sub-tree of $v_j$. Then for any vertex $v_{j'}$ in the sub-tree of $v_j$, set $S_{j'}$ is $1$-well-linked for $\Gamma_{j'}(\qset)$. Additionally, if $e$ is the edge connecting $v_j$ to its parent, then we have also computed a subset $\tpset_e\sse \pset'_e$ of at least $\tk=\Omega(\alphaWL k_2)$ paths, such that set $S_j$ is $1$-well-linked for the set $\Gamma_j(\qset\cup \tpset_e)$ of edges.

We now show how to process each vertex, so that the above invariant is preserved. Assume first that $v_j$ is a leaf vertex, and let $e$ be the unique edge of $T^*$ incident to $v_j$. Let $\Gamma=\Gamma_j(\pset_e)$. 
Consider the instance $\SC(G,S_j,\Gamma)$ of the sparsest cut problem. Recall that it is defined on a graph $G_j(\Gamma)$, with a set $\tset(\Gamma)$ of terminals, where each terminal $t_e\in \tset(\Gamma)$ corresponds to an edge $e\in \Gamma$. We apply Theorem~\ref{thm: grouping-many-sets} to graph $G_j(\Gamma)$ with a single set $\tset(\Gamma)$ of terminals. Since $S_j$ is $\alphaWL$-well-linked for $\Gamma$, graph $G_j(\Gamma)$ is $\alphaWL$-well-linked for $\tset(\Gamma)$, and so we will obtain a subset $\tset'(\Gamma)\sse \tset(\Gamma)$ of $\Omega(\alphaWL k_2)\geq \tk$ vertices. Let $\Gamma'_j\sse \Gamma$ be the subset of edges corresponding to the vertices in $\tset'(\Gamma)$. We define $\tpset_e\sse \pset'_e$ to be the subset of paths whose last edge belongs to $\Gamma'_j$. Notice that each edge $e'\in \Gamma'_j$ may serve as the last edge for at most two paths in $\pset'_e$. If two such paths exist, then only one of them is added to $\tpset_e$. This defines the required set $\tpset_e$, with $|\tpset_e|\geq \tk$. Observe that $S_j$ is $1$-well-linked for $\Gamma_j(\tpset_e)$.

Consider now some non-leaf vertex $v_j\neq v$. We assume that $v_j$ has two children. The case where $v_j$ has one child is handled similarly. Assume that $e$ is the edge connecting $v_j$ to its parent, and $e_1,e_2$ are the edges connecting $v_j$ to its two children in $T^*$. We define three subsets of edges in $\out(S_j)$: $\Gamma_0$, containing all edges of $\out(S_j)$ that belong to paths in $\pset'_e$, and $\Gamma_1,\Gamma_2$, containing the edges of $\out(S_j)$ that belong to the paths in $\tpset_{e_1}$, $\tpset_{e_2}$, respectively. Let $\Gamma=\Gamma_0\cup \Gamma_1\cup \Gamma_2$. Recall that $S_j$ is $\alphaWL$-well-linked for $\Gamma$. As before, we build an instance $\SC(G,S_j,\Gamma)$ of the sparsest cut problem, with the graph $G_j(\Gamma)$ and the set $\tset(\Gamma)$ of terminals. Let $\tset_0(\Gamma),\tset_1(\Gamma),\tset_2(\Gamma)\sse \tset(\Gamma)$ be the subsets of these new terminals corresponding to the edges in the sets $\Gamma_0,\Gamma_1$, and $\Gamma_2$, respectively. Then $G_j(\Gamma)$ is $\alphaWL$-well-linked for $\tset(\Gamma)$, and moreover $\tset_0(\Gamma)\geq \lfloor k_2/2\rfloor $ and $\tset_1(\Gamma),\tset_2(\Gamma)\geq \tk/2$ (since each edge in $\out(S_j)$ may participate in up to two paths in $\pset$). We apply Theorem~\ref{thm: grouping-many-sets} to graph $G_j(\Gamma)$ and three subsets $\tset_0(\Gamma),\tset_1(\Gamma),\tset_2(\Gamma)$ of terminals. Let $\tset_0'(\Gamma)\sse \tset_0(\Gamma)$, $\tset_1'(\Gamma)\sse \tset_1(\Gamma)$, and $\tset_2'(\Gamma)\sse \tset_2(\Gamma)$ be the output of the theorem. Recall that $|\tset_0'(\Gamma)|\geq \Omega(\alphaWL k_2)\geq \tk$, while $|\tset_1'(\Gamma)|,|\tset_2'(\Gamma)|\geq \Omega(\alphaWL \tk)=\Omega(\alphaWL^2 k_2)\geq k_5$. Moreover, graph $G_j(\Gamma)$ is $1$-well-linked for $\tset_0'(\Gamma)\cup \tset_1'(\Gamma)\cup\tset_2'(\Gamma)$, and these vertex sets are mutually disjoint. As before, we let $\Gamma'_0\sse \Gamma_0$ be the subset of edges corresponding to the vertices of $\tset'_0(\Gamma)$, and we let $\Gamma_1',\Gamma_2'$ be the subsets of edges corresponding to the vertices of $\tset_1'(\Gamma)$ and $\tset_2'(\Gamma)$, respectively. Finally, we set $\tpset_e\sse \pset'_e$ be the subset of paths that contain an edge in $\Gamma_0'$ (if some edge in $\Gamma_0'$ is contained in two such paths, only one such path is added to $\tpset_e$). 
We define $\pset_{e_1}''\sse \tpset_{e_1}$ and $\pset_{e_2}''\sse \tpset_{e_2}$ similarly, using the sets $\Gamma_1,\Gamma_2$ of edges. Observe that $|\tpset_e|\geq \tk$, and $|\pset_{e_1}|,|\pset_{e_2}|\geq k_5$ as required. It is easy to verify that the invariant is preserved.

In our final step we process the root vertex $v=v_j$. Recall that since we have assumed that $v=v'$, vertex $v$ has degree $1$. Let $e$ be the unique edge incident on $v$, and let $S_j\in \rset''$ be the set corresponding to $v$. We again define three subsets $\Gamma_0,\Gamma_1,\Gamma_2\sse \out(S_j)$ of edges, where $\Gamma_0$ contains all edges participating in the paths in $\tpset_e$, and $\Gamma_1$ and $\Gamma_2$ contain all edges participating in the paths in $\pset_1'$ and $\pset_1''$, respectively. Notice that $\mset_1$ defines a matching $\tmset$ over the set $\Gamma_1\cup \Gamma_2$ of edges, where for each pair $(s,t)\in \mset_1$, if $P,P'\in \pset_1'\cup \pset_1''$ are the two paths originating at $s$ and $t$ respectively, and $e',e''$ are the last edges on these two paths, then $(e,e'')\in \tmset$. As before, we set up an instance $\SC(G,S_j,\Gamma)$ of the sparsest cut problem, with graph $G_j(\Gamma)$ and the set $\tset(\Gamma)$ of terminals, where $\Gamma=\Gamma_0\cup \Gamma_1\cup \Gamma_2$.
Let $\tset_0(\Gamma),\tset_1(\Gamma)$ and $\tset_2(\Gamma)$ be the subsets of terminals in $\tset(\Gamma)$ corresponding to the edges in sets $\Gamma_0$, $\Gamma_1$, and $\Gamma_2$, respectively. 
 We apply Theorem~\ref{thm: grouping-advanced} to graph $G_j(\Gamma)$, where the first set of terminals is $\tset_1(\Gamma)\cup \tset_2(\Gamma)$, together with the matching $\tmset$, while the second set of terminals is $\tset_0(\Gamma)$. 
 Let $\tset_0'(\Gamma)\sse \tset_0(\Gamma)$, $\tset_1'(\Gamma)\sse \tset_1(\Gamma)$, and $\tset_2'(\Gamma)\sse \tset_2(\Gamma)$ be the output of the theorem, and let $\tmset'\sse \tmset$ be the corresponding matching over the vertices of $\tset_1'(\Gamma)\cup \tset_2'(\Gamma)$. As before, we define subsets $\Gamma_0',\Gamma_1',\Gamma_2'$ of edges corresponding to the terminal sets $\tset_0'(\Gamma),\tset_1'(\Gamma)$ and $\tset_2'(\Gamma)$, respectively. The subset $\pset_e''\sse \tpset_e$ of paths is defined exactly as before, from the subset $\Gamma_0'$ of edges. 

Finally, let $\pset_2'\sse \pset_1',\pset_2''\sse \pset_1''$ be the subsets of paths containing edges in $\Gamma_1'$ and $\Gamma_2'$, respectively (recall that each edge in $\Gamma_1'\cup \Gamma_2'$ may belong to at most one such path). Let $\tset_2'\sse \tset_1'$ be the subset of terminals where the paths of $\pset_2'$ originate, and similarly let $\tset_2''\sse \tset_1''$ be the subset of terminals where the paths of $\pset_2''$ originate. Let $\mset_2=\set{(s,t)\in \mset_1\mid s\in \tset_2',t\in \tset_2''}$. From the above discussion, $|\mset_2|=|\tmset'|\geq \Omega(\alphaWL)|\mset_1|$. This finishes the proof of the theorem for the case where $v=v'$. 

Assume now that $v\neq v'$. The algorithm is exactly as before, except that the vertices $v$ and $v'$ are processed differently. In order to process vertex $v'$, assume that $v'=v_j$, and let $e$ be the unique edge incident on $v'$ in tree $T^*$. We consider the two sets $\pset_1''$, $\pset'(e)$ of paths, and the two corresponding subsets $\Gamma_1,\Gamma_2\sse \out(S_j)$ of edges, that lie on the paths in $\pset_1''$ and $\pset'(e)$, respectively. We then process vertex $v'$ as any other degree-$2$ inner vertex in tree $T^*$, using the sets $\pset_1''$, $\pset'(e)$ of paths, and the corresponding subsets $\Gamma_1,\Gamma_2$ of edges. The output of this iteration is a subset $\tpset(e)\sse \pset'(e)$ of at least $\tk$ paths, and a subset $\tpset_1''\sse \pset_1''$ of at least $\Omega(\alphaWL)|\mset_1|$ paths, such that set $S_j$ is $1$-well-linked for $\Gamma_j(\tpset_1''\cup \tpset(e))$.

Let $\ttset_1''\sse \tset_1''$ be the subset of terminals $t$, where $P_t\in \tpset_1''$.
Next, we discard from $\mset_1$ all pairs $(s,t)$ that do not contain a terminal in $\tpset_1''$, obtaining a new set $\tmset\sse \mset_1$ of demand pairs. We also construct a set $\ttset_1'\sse \tset_1'$ of terminals, that participate in the pairs in $\tmset$, and a set $\tpset_1'\sse \pset_1'$ of paths originating from these terminals.

Vertex $v$ is processed as follows. Let $e,e'$ be two edges adjacent to $v$ in $T^*$ (the case where the degree of $v$ is $1$ is handled similarly). We process vertex $v$ exactly like all other inner vertices of tree $T^*$, with the corresponding three subsets of paths: $\tpset(e),\tpset(e')$, and $\tpset_1'$. As a result, we obtain three new subsets of paths: $\pset''(e)\sse \tpset(e)$ and $\pset''(e')\sse \tpset(e')$ containing at least $k_5$ paths each, and a subset $\pset_2'\sse\tpset_1'$ of $\Omega(\alphaWL^2)|\mset_1|$ paths. Let $\tset_2'\sse \ttset_1'$ be the subset of terminals $t$ where $P_t\in \pset_2'$, and let $\mset_2\sse \mset_1$ be the set of all pairs $(s,t)$ that contain a terminal in $\tset_2'$. We set $\tset_2''\sse \ttset_1''$ be the set of all terminals participating in pairs in $\mset_2$, and we let $\pset_2''\sse \tpset_1''$ be the set of paths originating at the terminals of $\pset_2''$.
\end{proof}

Let $k_6=|\mset_2|=\Omega(\alphaWL^2)k_4$. Notice that $k_6<k_5$. For each edge $e\in E(T^*)$, while $|\pset''(e)|>k_6$, we discard arbitrary paths from $\pset''(e)$, until $|\pset''(e)|=k_6$ holds. We now assume that $|\pset''(e)|=k_6$ for all $e\in E(T^*)$, and recall that $|\pset_2'|,|\pset_2''|=k_6$.

We are now ready to define the good crossbar in graph $G$. Let $\sset^*$ contain all sets $S_j$, where the degree of vertex $v_j$ in tree $T^*$ is either $1$ or $2$ (excluding the set $S$ corresponding to the root vertex $v$). Notice that at least half the vertices of $T^*$ must have this property, and therefore, $|S^*|\geq \gKRV$. If $|S^*|>\gkrv$, then we discard vertex subsets from $S^*$ arbitrarily, until $|S^*|=\gkrv$ holds. 

Consider some set $S_j\in S^*$, and let $e$ be any one of the edges adjacent to the vertex $v_j$ in tree $T^*$. We then define $\Gamma^*_j\sse \out(S_j)$ to be the subset of edges that participate in the paths in $\pset''(e)$. Recall that from Theorem~\ref{thm: step 3: ensuring well-linkedness}, set $S_j$ is $1$-well-linked for $\Gamma^*_j$.


Consider any set $S_j\in \rset''$. If the degree of $v_j$ in tree $T^*$ is $2$, then let $e,e'$ be the two edges incident to $v_j$ in $T^*$, and let $\Gamma,\Gamma'\sse \out(S_j)$ be the sets of edges lying on the paths $\pset''(e)$ and $\pset''(e')$, respectively. Since $S_j$ is $1$-well-linked for $\Gamma\cup \Gamma'$, we can find a set $\qset_j: \Gamma\sconnect_1 \Gamma'$ of paths contained in $S_j$.
 
 Assume now that the degree of $v_j$ in tree $T^*$ is $3$ (in this case, $S_j\not \in \sset^*$), and let $e,e'$ and $e''$ be the three edges incident to $v_j$ in $T^*$, where $e$ connects $v_j$ to its parent. Let $\Gamma,\Gamma',\Gamma''\sse \out(S_j)$ be the sets of edges lying on the paths in $\pset''(e),\pset''(e')$, and $\pset''(e'')$, respectively. Since $S_j$ is $1$-well-linked for $\Gamma\cup \Gamma'\cup \Gamma''$, we can find two sets $\qset^1_j: \Gamma\sconnect_1\Gamma'$ and $\qset^2_j: \Gamma\sconnect_1\Gamma''$ of paths inside $S_j$. Let $\qset_j=\qset^1_j\cup \qset^2_j$.

 Assume now that $v=v'$, and
 let $S$ be the set corresponding to the vertex $v$. Let $e$ be the unique edge incident on $v$ in tree $T^*$, and let $\Gamma\sse \out(S)$ be the subset of edges lying on the paths in $\pset''(e)$. Similarly, let $\Gamma'\sse \out(S)$ be the subset of edges lying on the paths in $\pset_2'\cup \pset_2''$. We can then find a set $\qset(v): \Gamma\sconnect_1\Gamma'$ of paths contained in set $S$. 
 
 We are now ready to complete the construction of the good crossbar for the case where $v=v'$. We let $\mset^*=\mset_2$. For each terminal $t\in \tset(\mset^*)$, we construct a tree $T_t$, which is then added to $\tau^*$, as follows. We start with the path $P_t\in \pset_2'\cup \pset_2''$, and let $\tilde e$ be the last edge on this path. We then add to $T_t$ the unique path in $\qset(S)$ that originates from $\tilde e$. Let $\tilde e'$ be the last edge on this path, let $e'$ be the edge of $T^*$ connecting $v$ to its child, and let $P$ be the unique path in $\pset_{e'}''$ originating at $\tilde e'$. We then add $P$ to the tree $T_t$.
 
 In general, we process the vertices of the tree $T^*$ in the top-bottom manner. We maintain the invariant that when a vertex $v_j$ is processed, for each edge $e$ adjacent to $v_j$ in tree $T^*$, one path of $\pset''_e$ has been added to tree $T_t$. Vertex $v_j$ is processed as follows. Let $e$ be the edge connecting $v_j$ to its parent, and let $P\in \pset''_e$ be the path that belongs to tree $T_t$. Let $\tilde e\sse \out(S_j)$ be the last edge on path $P$. If the degree of $v_j$ is $2$ in $T^*$, then we add to $T_t$ the unique path in $\qset_j$ that originates at $\tilde e$. Let $\tilde e'$ be the last edge on this path, and let $e'$ be the edge in $T^*$ connecting $v_j$ to its child. We add the unique path in $\pset_{e'}''$ originating at $\tilde e'$ to the tree $T_t$.
 
  Otherwise, if the degree of $v_j$ is $3$, then we add to $T_t$ two paths:  $P_1\in \qset^1_j$, and $P_2\in \qset_j^2$, that originate at $\tilde e$. Let $e',e''$ denote the two edges of $T^*$ connecting the vertex $v_j$ to its two children. As before, we add a unique path $P_3\in \pset''_{e'}$, originating at the last edge of $P_1$, and $P_4\in \pset''_{e''}$, originating at the last edge of $P_2$ to tree $T_t$.
  
  If the degree of $v_j$ is $1$, then we do nothing. Once all vertices of $T^*$ are processed, we obtain a tree $T_t$, containing the terminal $t$, and for each set $S_j\in \sset^*$, a unique edge in $\Gamma^*$. We let $\tau^*$ denote the collection of all trees $T_t$, for all $t\in \tset(\mset^*)$. It is immediate to verify that the trees in $\tau^*$ cause congestion at most $2$ in graph $G$. Moreover, since the sets $S_j\in \sset^*$ correspond to vertices of degree $1$ or $2$ in $T^*$, each edge of $G[S_j]$ for such sets $S_j$ belongs to at most one tree.
 
 Assume now that $v\neq v'$. Let $e,e'$ be the two edges incident to the vertex $v$ in tree $T^*$ (if the degree of $v$ is $1$, the argument is similar). As before, we denote by $\Gamma',\Gamma''\sse \out(S)$ the subsets of edges that lie on the paths in $\pset''_e,\pset''_{e'}$, respectively, and we let $\Gamma\sse \out(S)$ denote the subset of edges lying on the paths in $\pset_2'$. We then compute two sets $\qset_1(S):\Gamma\sconnect_1\Gamma'$, and $\qset_2(S):\Gamma\sconnect_1\Gamma''$ of paths contained in set $S$.
 
 Let $e$ be the unique edge incident on the vertex $v'$ in tree $T^*$, and let $\Gamma\sse \out(S')$ be the subset of edges participating in the paths in $\pset''_e$. Let $\Gamma'\sse \out(S')$ be the subset of edges participating in the paths in $\pset''_2$. We compute a set $\qset(S'):\Gamma\sconnect_1\Gamma'$ of paths contained in $S'$.
 
 For each terminal $t\in \tset(\mset_2)$, we now build a tree $T_t$. This tree is constructed exactly as before, and it contains the path $P_t\in \pset_2'$, and some additional path $P\in \pset_2''$. As before, for each edge $e\in T^*$, the tree contains exactly one path in $\pset''_e$. Notice that now each tree $T\in \tau^*$ contains two terminals: one in $\tset_2'$, and one in $\tset_2''$. We partition $\mset_2$ into two subsets, $\mset_2'$ and $\mset_2''$, where $\mset_2'$ contains all pairs $(s,t)$ where both $s$ and $t$ belong to the same tree in $\tau^*$, and $\mset_2''$ contains all remaining pairs. If $|\mset_2'|\geq |\mset_2|/2$, then we can route every pair $(s,t)\in \mset_2'$ along the tree $T\in \tau^*$ containing $s$ and $t$. This gives a routing of the pairs in $\mset_2'$ with congestion at most $2$. We now assume that $|\mset_2''|\geq |\mset_2|/2$.
 
 We construct the set $\mset^*\sse \mset_2''$ of demand pairs as follows. Start with $\mset^*=\emptyset$. While $\mset_2''\neq \emptyset$, select any pair $(s,t)\in \mset_2''$ and move it to $\mset^*$. Let $T_s,T_t\in \tau^*$ be the trees in which $s$ and $t$ participate, respectively. Remove from $\mset^*$ all pairs $(s',t')$, where either $s'$ or $t'$ participate in either $T_s$ or $T_t$. Observe that for each pair that we add to $\mset^*$, at most three pairs are removed from $\mset_2''$. Therefore, $|\mset^*|\geq |\mset_2''|/3\geq |\mset_2|/6\geq k_6/6$. Finally, we remove from $\tau^*$ all trees that do not contain terminals in $\tset(\mset^*)$. Observe that now each tree in $\tau^*$ contains exactly one terminal in $\tset(\mset^*)$. For each set $S_j\in \sset^*$, we also discard from $\Gamma^*_j$ all edges that do not participate in trees in $\tau^*$. Observe that for each $S_j\in \sset^*$, we now have $|\Gamma^*_j|=|\tau^*|\geq k_6/6$. Setting $k^*=|\mset^*|=|\tau^*|$ finishes the construction.

\bibliography{congestion-2}
\bibliographystyle{alpha}

\newpage

\label{--------------------------------Start Appendix--------------------------------------------}
\appendix

\section{Proof of Theorems~\ref{thm: grouping-many-sets} and \ref{thm: grouping-advanced}}\label{sec: proofs for groupings}
We start with several definitions and observations that are common to the proofs of both theorems. Both proofs closely follow the arguments of~\cite{ANF}. Assume that we are given a connected graph $G$ with maximum vertex degree at most $4$, and a set $\tset\sse V(G)$ of terminals, such that $G$ is $\alpha$-well-linked for $\tset$, for some $\alpha<1$. 
Let $q$ be the smallest even integer with $q\geq 4/\alpha$, so $4/\alpha\leq q\leq 8/\alpha$. We assume that $|\tset|\geq q$.

We use the following simple theorem, whose proof is very similar to the proof of Theorem~\ref{thm: grouping}, to find an initial clustering of the vertices of $G$.

\begin{theorem}\label{thm: grouping simple degree 4}
There is an efficient algorithm that finds a partition $\cset$ of the vertices of $G$, where each set $C\in \cset$ contains at least $q$ and at most $4q$ terminals. Moreover, each cluster $C\in \cset$ induces a connected component in $G$.
\end{theorem}
\begin{proof}
Let $T$ be any spanning tree of $G$. Our algorithm is iterative. We start with $\cset=\emptyset$. During the algorithm, we remove some vertices from $T$, and we maintain the invariant that $T$ contains at least $q$ terminals.

If $T$ contains at most $4q$ terminals, then we add all the vertices of $T$ into $\cset$ as a single cluster, and terminate the algorithm. Otherwise, we find the lowest vertex $v$ in $T$, such that the sub-tree rooted at $v$ contains at least $q$ terminals. Notice that since we have assumed that the maximum vertex degree in $G$ is $4$, $v$ has at most three children, and the subtree of each such child contains fewer than $q$ terminals. Therefore, the total number of terminals contained in the subtree of $v$ is at most $3q$. We add all the vertices in the sub-tree of $v$ as a new cluster to $\cset$, and remove them from $T$. Notice that since we assumed that $T$ contained at least $4q$ terminals at the beginning of the current iteration, it contains at least $q$ terminals at the end of the current iteration. We then continue to the next iteration.
\end{proof}

Let $\hset$ be the set of all sub-graphs induced by the vertices in the subsets in $\cset$, that is, $\hset=\set{G[C]\mid C\in \cset}$. We will refer to the sub-graphs $H\in \hset$ as \emph{clusters}. Notice that all clusters in $\hset$ are vertex disjoint, and every vertex belongs to at least one cluster. For each cluster $H\in \hset$, we let $\tset(H)\sse \tset$ be the subset of terminals contained in $H$, and for each terminal $t\in \tset$, we let $H(t)\in \hset$ denote the unique cluster to which $t$ belongs.

Let $E'\sse E(G)$ be the set of all cut edges of graph $G$ (recall that $e\in E(G)$ is a cut edge iff its removal disconnects the graph $G$). We need the following definition.

\begin{definition}
Let $H$ be any sub-graph of $G$, $\tset(H)$ any subset of terminals contained in $H$, and let $v\in \tset(H)$. We say that $v$ is a \emph{center} of $H$, iff $v$ can send one flow unit to the vertices of $\tset(H)$ in graph $H$, such that every vertex of $\tset(H)$ receives at most $1/q$ flow units, and
the flow on every edge $e\in E(H)$ is at most $1/2$.

We say that $v$ is a \emph{pseudo-center} of $H$, iff $v$ can send one flow unit to the vertices of $\tset(H)$, with the same restrictions as above, except that the edges $e\in E'$ are allowed to carry up to one flow unit.
\end{definition}

We need the following claim.
\begin{claim}\label{claim: 1-edge-wl}
Assume that we are given a collection $\gset$ of {\bf edge-disjoint} sub-graphs of $G$, a subset $\tset(H)\sse \tset$ of terminals for each sub-graph $H\in \gset$, and a subset $\tset'\subseteq \tset$ of terminals, such that:

\begin{itemize}
\item For $H\neq H'$, $\tset(H)\cap \tset(H')=\emptyset$.

\item For each $t\in \tset'$, there is a cluster $H(t)\in \gset$, such that $t\in \tset(H)$. Moreover, $t$ is a pseudo-center of $H$.

\item For each $H\in \gset$, $|\tset(H)\cap \tset'|\leq 1$.
\end{itemize}

Then $G$ is $1$-well-linked for $\tset'$.
\end{claim}

\begin{proof}
Let $(A,B)$ be any partition of $V(G)$. Denote $\tset_A=A\cap \tset', \tset_B=B\cap \tset'$, and assume w.l.o.g. that $|\tset_A|\leq |\tset_B|$. In order to show that $|E(A,B)|\geq |\tset_A|$, it is enough to prove that there is a flow $F:\tset_A\connect_1\tset_B$ in $G$. Let $\tset_B'\sse \tset_B$ be any subset of terminals of size $|\tset_A|$.

We compute the flow $F$ in two steps. In the first step, we find a flow $F': \tset_A\sconnect\tset_B'$, where the flow on every edge $e\in E(G)\setminus E'$ is at most $1$, and the flow on every edge $e\in E'$ is at most $1.5$. In the second step, we transform $F'$ into the desired flow $F:\tset_A\connect_1\tset_B$.

We start by showing the existence of the flow $F'$. Recall that each terminal $t\in \tset'$ is associated with a cluster $H(t)$, and there is a flow $F_t$ inside $H(t)$, where $t$ sends one flow unit to the vertices of $\tset(H(t))$, and each such vertex receives at most $1/q$ flow units. The congestion due to flow $F_t$ on every edge $e\in E(G)\setminus E'$ is at most $\half$, and the congestion on every edge $e\in E'$ is at most $1$. Since we have selected $q$ to be an even integer, from the integrality of flow, we can assume that the flow $F_t$ is $1/q$-integral (to do so, set the capacity of every edge $e\not\in E'$ to be $q/2$, the capacity of every edge $e\in E'$ to $q$, add a source $s$, and connect it to every vertex $v\in \tset(H(t))$ with an edge of capacity $1$. The integral $s$-$t$ flow of value $q$ in this network, scaled down by factor $q$, gives the desired flow $F_t$).
Every vertex $v\in \tset(H(t))$ receives either $0$ or $1/q$ flow units in $F_t$. Let $S'_t\sse \tset(H(t))$ be the set of exactly $q$ vertices that receive $1/q$ flow units in $F_t$.

Let $A'=\bigcup_{t\in \tset_A}S'_t$, and let $B'=\bigcup_{t\in \tset_B'}S'_t$. Notice that $|A'|=q|\tset_A|=q|\tset_B'|=|B'|$, since the terminal sets $\tset(H),\tset(H')$ are disjoint for $H\neq H'$. Since both sets $A'$ and $B'$ only contain terminals, and the graph $G$ is $\alpha$-well-linked for the terminals, there is a flow $F'':A'\sconnect_{2/\alpha}B'$ in graph $G$. Scaling this flow down by factor $q\geq 4/\alpha$, we obtain a flow $F^*$, where every vertex $v\in A'$ sends $1/q$ flow units, every vertex $v\in B'$ receives $1/q$ flow units, and the total edge congestion is at most $1/2$. Let flow $F'$ be the concatenation of $\bigcup_{t\in\tset_A}F_t$, the flow $F^*$, and $\bigcup_{t\in \tset_B'}F_t$. Then every vertex $v\in \tset_A$ sends one flow unit in $F'$, every vertex $v\in \tset_B'$ receives one flow unit, the flow on every edge $e\not\in E'$ is at most $1$, and the flow on every edge $e\in E'$ is at most $1.5$. We assume w.l.o.g. that flow $F'$ is half-integral. That is, we can decompose $F'$ into flow-paths, each of which carries exactly $1/2$ flow units. 	Notice that every vertex is an endpoint of either $0$ or $2$ such flow-paths.

We now transform the flow $F'$ into a flow $F:\tset_A\connect_1\tset_B$. This part is similar to the proof of Lemma 3.10 in~\cite{ANF}. The transformation replaces every flow-path in $F'$ by a corresponding simple flow-path, by removing all cycles on the path, if such exist. Let $F$ be the resulting flow. It is now enough to show that every edge $e\in E'$ carries at most one flow unit. Let $e\in E'$ be any cut edge, and assume for contradiction that more than one flow unit is sent via $e$. Let $X,Y$ be the two connected components of $G\setminus\set{e}$. 
Since the amount of flow on $e$ is bounded by $1.5$ and we assume that all flow-paths are simple and carry exactly $1/2$ flow units each, there are three flow-paths that use edge $e$, and each such path connects a vertex of $X$ to a vertex of $Y$. 

For every vertex $v\in X$, let $n_1(v)$ be the total number of flow-paths originating or terminating at $v$, and let $n_1=\sum_{v\in X}n_1(v)$. Since $n_1(v)\in \set{0,2}$ for each $v$, $n_1$ is an even integer. Let $n_2(v)$ be the total number of flow-paths originating of terminating at $v$, whose other endpoint also belongs to $X$, and let $n_2=\sum_{v\in X}n_2(v)$. Since every flow-path contained in $X$ contributes $2$ to $n_2$, $n_2$ is also an even integer. However, since there are exactly three flow-paths that start at $X$ and terminate at $Y$, $n_1-n_2=3$, which is impossible.
\end{proof}

\subsection*{Proof of Theorem~\ref{thm: grouping-many-sets}}
The proof roughly follows the arguments of~\cite{ANF}, and consists of three steps. In the first step, we define an initial clustering of the vertices of $G$ and select one representative terminal from each cluster. If we could claim at this point that each selected representative is a pseudo-center of its cluster, then we would be done by Claim~\ref{claim: 1-edge-wl}. However, this is not necessarily the case. To overcome this difficulty, we use a procedure suggested by~\cite{ANF} called \emph{tagging}, where each cluster $C$ ``tags'' some other cluster $C'$, such that the representative of $C$ is a pseudo-center of the merged cluster $C\cup C'$. The tagging procedure is performed in Step 2. Finally, in Step 3, we merge some pairs of clusters and select one pseudo-center for each such merged cluster.  

\paragraph{Step 1: Initial Clustering and Representative Selection}
We use Theorem~\ref{thm: grouping simple degree 4} to find an initial clustering $\hset$ of the vertices of $G$. Next, we select a representative for each cluster $H\in \hset$, using the following lemma.

\begin{lemma}\label{lem: rep selection}
There is an efficient randomized algorithm that w.h.p. finds, for each $1\leq j\leq r$, a subset $\ttset_j\sse \tset_j$ of terminals, such that $|\ttset_j|\geq \Omega(\alpha |\tset_j|/r)$, the sets $\set{\ttset_j}_{j=1}^r$ are mutually disjoint, and for each cluster $H\in \hset$, $\tset(H)\cap \left(\bigcup_{j=1}^r\ttset_j\right )=1$.
\end{lemma}

\begin{proof}
We start by selecting, for each $1\leq j\leq r$, at most one representative $t\in\tset(H)\cap \tset_j$ for every cluster $H$, and we will later ensure that at most one representative is chosen for each cluster overall. 

Fix some $1\leq j\leq r$.
We construct an initial set $\ttset_j'\subseteq \tset_j$ of representatives, as follows. Start with $\ttset_j'=\emptyset$. While $\tset_j\neq \emptyset$, select any terminal $t\in \tset_j$ and add it to $\ttset_j'$. Remove from $\tset_j$ all terminals that belong to $\tset(H(t))$, and continue to the next iteration. Let $\ttset_j'$ be the final subset of terminals. Then $|\ttset_j'|\geq \frac{|\tset_j|}{4q}$, since for each terminal added to $\ttset_j'$, at most $4q$ terminals are removed from $\tset_j$. Moreover, for each cluster $H\in \hset$, at most one terminal from $\tset(H)$ belongs to $\ttset_j'$.

Consider now some cluster $H\in \hset$, and let $\rset(H)$ be a multi-set, containing, for each $1\leq j\leq r$, the unique terminal in $\tset(H)\cap \ttset_j'$, if it exists (if some terminal $t\in \tset(H)$ belongs to several sets $\ttset_j'$, then it appears several times in $\rset(H)$). Since $\rset(H)$ contains at most one representative from each set $\ttset_j'$, $|\rset(H)|\leq r$.  Cluster $H$ then selects one terminal $t\in \rset(H)$ uniformly at random, and terminal $t$ becomes the \emph{representative of the cluster $H$}. 

For each $1\leq j\leq r$, let $\ttset_j\sse \ttset_j'$ be the subset of terminals that serve as representatives of clusters in $\hset$. Then the sets $\set{\ttset_j}_{j=1}^r$ are mutually disjoint, and each cluster $H\in \hset$ has exactly one terminal in $\tset(H)\cap \left(\bigcup_{j=1}^r\ttset_j\right )$. It now only remains to show that for each $1\leq j\leq r$, $|\ttset_j|\geq \Omega(\alpha |\tset_j|/r)$.

Fix some $1\leq j\leq r$, and let $\mu_j$ denote the expectation of $|\ttset_j|$. Then $\mu_j\geq |\ttset_j'|/r\geq |\tset_j|/(4qr)=\Omega(\alpha |\tset_j|/r)=\Omega(r\log r)$. From the Chernoff bound, $\prob{|\ttset_j|<\mu_j/2}\leq e^{-\mu_j/12}\leq 1/r^r$. Using the union bound, the probability that $|\ttset_j|\geq \Omega(\alpha |\tset_j|/r)$ for all $1\leq j\leq r$ is at least $1/r^{r-1}$. 
By repeating the procedure a polynomial number of times, we can ensure that this algorithm succeeds w.h.p.
\end{proof}

For each cluster $H\in \hset$, the unique terminal in $\tset(H)\cap \left(\bigcup_{j=1}^r\ttset_j\right )$ is called the \emph{representative of $H$}.
So far we have constructed the sets $\ttset_j$, for $1\leq j\leq r$ of terminals that have all desired properties, except that we are not guaranteed that $G$ is $1$-well linked for $\bigcup_{j=1}^r\ttset_j$. If we could guarantee that each terminal $t\in \bigcup_{j=1}^r\ttset_j$ is a pseudo-center of its cluster $H(t)$, then this property would follow from Claim~\ref{claim: 1-edge-wl}. However, we cannot guarantee this property at this stage. We overcome this difficulty in the next step, using the tagging technique of~\cite{ANF}. For each $1\leq j\leq r$, we denote $\tk_j=|\ttset_j|=\Omega(\alpha |\tset_j|/r)$.

\paragraph{Step 2: Tagging}
The idea of the tagging technique is that we merge some pairs of clusters in $\hset$, and select one representative for each such merged cluster (from among the two representatives of the merged clusters), in a way that ensures that these representatives are pseudo-centers of their new clusters. We use the following result of~\cite{ANF}.

\begin{theorem} (Lemma 3.13 in~\cite{ANF})\label{thm: tagging}
Let $H\in \hset$ be any cluster, and let $s\in V(H)$ be any vertex that is not a pseudo-center of $H$. Then there is an edge $e=(x,y)$ with $x\in V(H)$, $y\not\in V(H)$, such that for any path $P$, connecting $H$ to another cluster $H'\in \hset$, where the first edge of $P$ is $e$, and $P$ is edge-disjoint from $H$, vertex $s$ is a pseudo-center of the cluster $H\cup P\cup H'$. Moreover, the edge $e=(x,y)$ can be found efficiently.
\end{theorem}

For every terminal $t\in \bigcup_{j=1}^r\ttset_j$, if $t$ is not a pseudo-center of $H(t)$, let $e_t=(x_t,y_t)$ be the corresponding edge guaranteed by Theorem~\ref{thm: tagging}, and let $H'(t)$ be the cluster of $\hset$ containing $y_t$ (such a cluster exists since every vertex of $G$ belongs to one of the clusters in $\hset$, from Theorem~\ref{thm: grouping simple degree 4}). We then say that the terminal $t$ and the cluster $H(t)$ \emph{tag} the cluster $H'(t)$.

Intuitively, we can now create a new cluster, by merging $H(t)$ and $H'(t)$. From Theorem~\ref{thm: tagging}, we are guaranteed that $t$ is a pseudo-center of the new cluster. However, since we require that the clusters are disjoint and every cluster only has one representative, the representative of the cluster $H'(t)$ will need to be discarded. While in general we only expect to discard a constant fraction of the terminals, it is possible that many of these terminals belong to one of the sets $\ttset_j$, and we will end up discarding too many terminals from one such set. Another difficulty is when many terminals tag the same cluster $H'$. In this case, following the approach of~\cite{ANF}, we can partition all such terminals into pairs, and connect their corresponding clusters by edge-disjoint paths that are contained in $H'$. For each such new merged cluster, we can select one of the two representative terminals that is guaranteed to be a pseudo-center by Theorem~\ref{thm: tagging}. Again, we will need to discard some clusters (such as the cluster $H'$ in this case), and we need to ensure that we do not discard too many representatives from any set  $\ttset_j$. 

We start by building a collection $\hset^*\sse \hset$ of  clusters, such that for any pair $(H,H')$ of clusters, where $H$ tags $H'$, at most one of these clusters may belong to $\hset^*$. We also ensure that among the representatives of the clusters in $\hset^*$, enough terminals belong to each set $\ttset_j$.

For convenience, for each $1\leq j\leq r$, we say that all terminals in $\ttset_j$ are of color $j$.
For each terminal $t\in \ttset_j$, the color of the corresponding cluster  $H(t)$ is also $j$. Therefore, we have at least $\tk_j$ clusters of each color $j$. Following the arguments of~\cite{ANF}, we build a directed graph $D$, whose vertex set corresponds to the clusters in $\hset$, that is, $V(D)=\set{v_H\mid H\in \hset}$, and there is an edge $(v_H,v_{H'})$ iff $H$ tags $H'$. We say that a vertex $v_H$ has color $j$ iff the color of $H$ is $j$.
 As in~\cite{ANF}, we need to find a large independent set $\iset$ in graph $D$. However, we also need to ensure that $\iset$ contains $\Omega(\tk_j/r)$ vertices of each color $j$.

\begin{claim}\label{claim: independent set}
There is an efficient algorithm that finds an independent set $\iset$ in graph $D$, containing at least $\Omega(\tk_j/r)$ vertices of each color $1\leq j\leq r$.
\end{claim}

\begin{proof}
For each $1\leq j\leq r$, let $k'_j$ denote the number of vertices of color $j$ in graph $D$. We will assume w.l.o.g. that $k'_2=k'_3=\cdots=k'_r=\Omega(\alpha k_2/r)$, and $k'_1=\Omega(\alpha k_1/r)$,  by discarding, if necessary, some vertices from the graph $D$. We consider two cases. 

Assume first that $k_1'\leq rk_2'$. Let $\beta=rk'_2/k'_1$.
Notice that the out-degree of every vertex in $D$ is at most $1$. We say that a vertex $v_H$ of color $1$ is good, iff the total number of edges adjacent to $v_H$, whose other endpoint's color is different from $1$, is at most $8\beta$. Since the total number of vertices in the graph $D$ is at most $r k_2'+k_1'\leq 2rk_2'=2\beta k_1'$, the total number of edges is also at most $2\beta k_1'$. Therefore, at least half the vertices of color $1$ are good. Let $D'$ be the sub-graph of $D$ induced by the good color-$1$ vertices. Since the out-degree of every vertex in $D'$ is at most $1$, and $D'$ contains at least $k_1'/2$ vertices, we can efficiently find an independent set $\iset_1$ of size exactly $\lfloor \frac{k_1'}{16r}\rfloor$ in graph $D'$, by iteratively greedily selecting a vertex with minimum total degree. The selected vertex $v$ is then added to $\iset_1$, and it is deleted from $D'$, together with all its neighbors. Let $\iset_1$ be the resulting independent set of size $\lfloor \frac{k_1'}{16r}\rfloor$, and
let $V'$ be the subset of the vertices of colors $2,\ldots,r$, whose neighbors belong to $\iset_1$. Then $|V'|\leq \lfloor \frac{k_1'}{16r}\rfloor \cdot  8\beta \leq \frac{k_2'}{2}$. 

Let $D''$ be the sub-graph of $D$ induced by the vertices of colors $2,\ldots,r$, that do not belong to $V'$. Then for each $2\leq j\leq r$, $D''$ contains at least $k_2'/2$ vertices of color $j$.
Next we construct, for each $2\leq j\leq r$, a set $\iset_j$ of $\lfloor \frac{k_2'}{6r}\rfloor$ vertices of color $j$, such that $\bigcup_{j=1}^r\iset_j$ is an independent set in graph $D$. We start with $\iset_j=\emptyset$ for all $2\leq j\leq r$, and the perform iterations. In each iteration, we select a minimum-degree vertex $v$ from $D''$, and add it to the set $\iset_j$, where $j$ is the color of $v$. We then delete $v$ and all its neighbors from $D''$.
If, at any point of the algorithm execution, the size of the set $\iset_j$, for any $2\leq j\leq r$,  reaches $\lfloor \frac{k_2'}{6r}\rfloor$, then we delete all  the remaining vertices of color $j$ from $D''$, and continue. Notice that since the out-degree of every vertex is at most $1$  in any induced sub-graph of $D''$, we can always find a vertex whose degree is at most $2$ in the current graph. Therefore, in each iteration, at most $3$ vertices are deleted from $D''$.

It is easy to see that at the end of the algorithm, for each $2\leq j\leq r$, $|\iset_j|=\lfloor \frac{k_2'}{6r}\rfloor$. Indeed, assume otherwise, and assume that for some $j^*$, set $\iset_{j^*}$ contains fewer than $\lfloor \frac{k_2'}{6r}\rfloor$ vertices upon the termination of the algorithm. But the original graph $D''$ contained at least $k_2'/2$ vertices of color $j^*$, the number of iterations is bounded by $(r-1)\cdot \lfloor \frac{k_2'}{6r}\rfloor<\frac{k_2'}{6}$, and in each iteration at most $3$ vertices of color $j^*$ are deleted from $D''$. Therefore, upon the termination of the algorithm, at least one vertex of color $j^*$ must be present in the graph, a contradiction. The output of the algorithm is $\bigcup_{j=1}^r\iset_j$.

Consider now the second case, where $k'_1>rk_2'$. Let $\beta=\frac{k_1'}{rk_2'}$. We say that a vertex of colors $2,\ldots, r$ is good if it has at most $8\beta r$ neighbors of color $1$. Since the total number of edges in the graph is bounded by $k_1'+rk_2'\leq 2k_1'= 2\beta rk_2'$, for every color $2\leq j\leq r$, at least half the vertices of color $j$ are good. We construct a sub-graph $D'$ of graph $D$, induced by all good vertices of colors $2,\ldots,r$. Since graph $D'$ contains at least $k_2'/2$ vertices of each color, using the same algorithm as before, we can find, for each $2\leq j\leq r$, a collection $\iset_j$ of $\floor{\frac{k_2'}{16r}}$ good color-$j$ vertices, such that the set $\iset'=\bigcup_{j=1}^r\iset_j$ is an independent set in $D'$. Let $V'$ be the subset of color-$1$ vertices that have neighbors in $\iset'$. Then, since all vertices in $\iset'$ are good, $|V'|\leq r\cdot \frac{k_2'}{16r}\cdot 8\beta r=\frac{\beta k_2'r}{2}=\frac{k_1'}{2}$. Let $V''$ be the subset of color-$1$ vertices that do not belong to $V'$, and let $D''$ be the subgraph of $D$ induced by $V''$. Then $D''$ contains at least $k'_1/2$ vertices, and using the same algorithm as before, we can find an independent set $\iset_1$ of size at least $k'_1/8$ in graph $D''$. We then output $\bigcup_{j=1}^r\iset_j$, which is guaranteed to be an independent set.
\end{proof}

Let $\hset^*\sse \hset$ contain the set of clusters $H$ whose corresponding vertex $v_H\in \iset$. For each $1\leq j\leq r$, let $\ttset_j^*\sse \ttset_j$ denote the sets of terminals $t$ where $H(t)\in \hset^*$. Then for each $1\leq j\leq r$, $|\ttset_j^*|=\Omega(\alpha |\tset_j|/r^2)$ w.h.p.

\paragraph{Step 3: Merging the Clusters}
In this step we merge some pairs of clusters, and select a representative for each new merged cluster, such that on the one hand, this representative is a pseudo-center of the merged cluster, and on the other hand, we obtain enough representatives from each set $\ttset_j^*$ of terminals.

We construct a new collection $\gset$ of merged clusters, and the final collection $\tset'=\bigcup_{j=1}^r\tset_j'$ of representative terminals, where for each $1\leq j\leq r$, $\tset_j'\sse \ttset_j^*$.
We start with $\gset=\emptyset$ and $\tset'=\emptyset$.

We add to $\gset$ all clusters $H\in\hset^*$ that do not tag any cluster. Recall that in this case, the unique terminal $t\in \tset(H)\cap \left(\bigcup_{j=1}^r\ttset_j^*\right)$ is a pseudo-center of $H$. We add $t$ to $\tset'$.

Consider now some cluster $H\in \hset$, and let $\gset(H)\sse \hset^*$ be the set of all clusters in $\hset^*$ that tag $H$. If $|\gset(H)|=1$, then let $H'\in \hset^*$ be the unique cluster that tags $H$. Let $t\in \bigcup_{j=1}^r\ttset_j^*$ be the representative of $H$, that is, $t\in \tset(H)$. We then create a new merged cluster, $H^*=H\cup H'\cup e_{t(H)}$, which is added to $\gset$. We set $\tset(H^*)=\tset(H)\cup \tset(H')$. By Theorem~\ref{thm: tagging}, we are guaranteed that $t$ is a pseudo-center of $H^*$. We add $t$ to $\tset'$. 

Assume now that $|\gset(H)|>1$. We use the following lemma from~\cite{ANF}.
\begin{lemma} (Lemma 3.14 in~\cite{ANF}) \label{lemma: pairing}
Let $T$ be a tree and $A$ some even multi-set of vertices in $V(T)$. Then we can efficiently find $|A|/2$ edge-disjoint paths in $T$, such that each vertex $v\in A$ is the endpoint of exactly $n_v$ of these paths, where $v$ occurs $n_v$ times in $A$.
\end{lemma}

Let $\gset(H)=\set{H_1,\ldots,H_p}$, and for each $1\leq i\leq p$, let $t_i\in \tset(H_i)$ be the unique terminal that belongs to $\bigcup_{j=1}^r\ttset_j^*$.
If $p$ is odd, then we choose one of the clusters $H_1,\ldots,H_p$ uniformly at random, and discard it from $\hset^*$. Notice that $p\geq 3$ must hold in this case, so the probability that any cluster is discarded is at most $1/3$. We assume from now on that $p$ is even.
 Let $A=\set{y_{t_1},y_{t_2},\ldots,y_{t_p}}$. We use Lemma~\ref{lemma: pairing} to find a collection $\pset$ of edge-disjoint paths contained in $H$ (by first  building a spanning tree $T$ of $H$), where every path connects a distinct pair $y_{t_i},y_{t_j}$ of vertices. For each such path $P\in \pset$, if $P$ connects $y_{t_i}$ to $y_{t_j}$, then we construct a new cluster $H^*=H_i\cup P\cup H_j$, and add it to $\gset$.  We set $\tset(H^*)=\tset(H_i)\cup \tset(H_j)$. Notice that from Theorem~\ref{thm: tagging}, both $t_i$ and $t_j$ are pseudo-centers of $H^*$. We select one of these two terminals uniformly at random, and add it to $\tset'$.


Consider the final set $\gset$ of clusters. From the construction, we are guaranteed that all these clusters are edge-disjoint, since the original clusters $\hset$ were vertex-disjoint. Every terminal $t\in \tset'$ is guaranteed to be a pseudo-center of its cluster in $\gset$.
It is now enough to show that for each $1\leq j\leq r$, at least $|\ttset_j^*|/12=\Omega(\alpha|\tset_j|/r^2)$ terminals from $\ttset_j^*$ belong to $\tset'$ w.h.p.

Fix some $1\leq j\leq r$, let $t\in \ttset_j^*$ be any terminal, and let $H(t)\in \hset^*$ be the cluster that $t$ represents. Recall that $H(t)$ may be discarded from $\hset^*$ with probability at most $1/3$ (if $H(t)$ tags some other cluster $H'$, for which the parameter $p$ is an odd number greater than $1$), and additionally, if $H(t)$ is merged with another cluster $H'$, then $t$ is added to $\tset'$ with probability at least $\half$. Using the standard Chernoff bound, the probability that more than a $2/3$-fraction of clusters of color $j$ are discarded from $\hset^*$ is at most $1/\poly(r)$, and similarly, the probability that fewer than $\frac{|\ttset^*_j|}{12}$ terminals from $\ttset^*_j$ are added to $\tset'$ is bounded by $1/\poly(r)$, since $|\ttset^*_j|\geq \Omega(\alpha|\tset_j|/r^2)=\Omega(\log r)$. Using the union bound over all $1\leq j\leq r$, with probability at least $\half$, set $\tset'$ contains at least $|\ttset_j^*|/12$ terminals from $\ttset_j$, for all $1\leq j\leq r$. By repeating this procedure a polynomial number of times, we can ensure that it succeeds with high probability.
Applying Claim~\ref{claim: 1-edge-wl} to the set $\gset$ of clusters and the set $\tset'$ of terminals, we conclude that $G$ is $1$-well-linked for $\tset'$.

\subsection*{Proof of Theorem~\ref{thm: grouping-advanced}}

The proof closely follows the proof of Theorem~\ref{thm: grouping-many-sets}, except that now we are only given two subsets $\tset_1,\tset_2$ of terminals, and we need to ensure that the terminals in the set $\tset'_1$ are chosen so that for every pair $(s,t)\in \mset$, either both $s,t$ belong to $\tset'_1$, or none of them.

\paragraph{Step 1: Initial Clustering and Representative Selection}
As before, we use Theorem~\ref{thm: grouping simple degree 4} to find an initial clustering $\hset$. We then select 
representatives for the clusters $H\in \hset$, using the following analogue of Lemma~\ref{lem: rep selection}.

\begin{lemma}
There is an efficient algorithm, that either routes a subset $\mset'\sse \mset$ of $\Omega(\alpha k_1)$ pairs on edge-disjoint paths in $G$, or finds two subsets $\ttset_1\sse \tset_1$ and $\ttset_2\sse \tset_2$,  with $|\ttset_1|=\Omega(\alpha k_1)$, $|\ttset_2|=\Omega(\alpha k_2)$, such that:

\begin{itemize}
\item $\ttset_1\cap \ttset_2=\emptyset$, 

\item for each cluster $H\in \hset$, $|\tset(H)\cap (\ttset_1\cup \ttset_2)|\leq 1$, and 

\item for each pair $(s,t)\in \mset$, either both $s,t$ belong to $\ttset_1$, or neither of them does.
\end{itemize}
\end{lemma}

\begin{proof}
 We start by selecting a subset $\ttset_2'\sse \tset_2$ of representatives from the set $\tset_2$, exactly as in the proof of Lemma~\ref{lem: rep selection}. We construct a collection $\ttset_2'\sse \tset_2$ of at least $\frac{k_2}{4q}$ terminals, such that for each cluster $H\in \hset$, at most one terminal of $\tset(H)$ belongs to $\ttset_2'$.

The selection of representatives for $\tset_1$ is performed slightly differently.
We construct a new matching $\mset_1\sse \mset$, that consists of the selected representatives. We start with $\mset_1=\emptyset$. While $\mset\neq \emptyset$, let $(s,t)\in \mset$ be any pair of terminals in $\mset$. Add $(s,t)$ to $\mset_1$, and remove from $\tset_1$ all terminals contained in $\tset(H(s))\cup \tset(H(t))$. Remove from $\mset$ all pairs that are not contained in the current set $\tset_1$ of terminals. Let $\mset_1$ be the final matching that the algorithm has constructed. It is easy to verify that $|\mset_1|\geq \frac{k_1}{16q}$, since for every pair that we add to $\mset_1$, we remove at most  $8q$ pairs from $\mset$. Let $\ttset_1'\sse\tset_1$ be the subset of terminals that participate in pairs in $\mset_1$. Then for every cluster $H\in \hset$, either $|\tset(H)\cap \ttset_1|\in\set{0,1}$, or $|\tset(H)\cap \ttset_1|=2$. In the latter case, the two terminals $s,t\in \tset(H)\cap \ttset_1$ belong to $\mset$ as a pair.

Partition $\mset_1$ into two subsets $\mset_1'$ and $\mset_1''$, where $\mset_1'$ contains all pairs $(s,t)$ where both $s,t\in \tset(H)$ for some cluster $\tset(H)$, and $\mset_1''$ contains all remaining pairs. If $|\mset_1'|\geq \half |\mset_1|$, then we have found a collection $\mset_1'\sse \mset$ of $\Omega(\alpha k_1)$ pairs, where each pair $(s,t)\in \mset_1'$ is associated with a distinct cluster $H\in \hset$, where $s,t\in \tset(H)$. Since each cluster $H\in \hset$ is connected, and the clusters are mutually edge-disjoint, we can route the pairs in $\mset_1'$ on edge-disjoint paths.
Therefore, we assume from now on, that $|\mset_1''|\geq |\mset_1|/2\geq \frac{k_1}{32q}$. To simplify the notation, we denote $\mset_1''$ by $\mset_1$ from now on. We denote by $\ttset_1'$ the set of terminals participating in the pairs in $\mset_1$, so $|\ttset_1'|=\Omega(\alpha k_1)$.

Notice that now for every cluster $H\in \hset$, at most one terminal $t\in \tset(H)$ belongs to $\ttset_1'$, and at most one terminal $t'\in \tset(H)$ belongs to $\ttset_2'$, while it is possible that $t=t'$. If $\tset(H)$ contains both a terminal of $\ttset_1$ and a terminal of $\ttset_2$, then we say that $H$ is a mixed cluster. Let $\hset'$ be the set of all mixed clusters.

We partition $\hset'$ into two subsets $\hset'_1,\hset'_2$, such that $|\hset'_1|,|\hset'_2|\geq |\hset'|/2-1$. 
Our partition ensures that if, for some pair $(s,t)\in \mset_1$, there are two clusters $H,H'\in \hset'$ with $s\in \tset(H)$ and $t\in \tset(H')$, then either both $H,H'\in \hset_1'$, or both $H,H'\in \hset_2'$.
The partition is performed by using the following greedy procedure. Start with $\hset_1',\hset_2'=\emptyset$. While $\hset'\neq \emptyset$, let $H\in \hset'$ be any cluster, and let $s$ be the unique terminal in $\tset(H)\cap \ttset_1$. Let $t$ be the terminal such that $(s,t)\in \mset_1$, and let $H'\in \hset$ be the cluster for which $t\in \tset(H')$. If $H'\not\in \hset'$, then we add $H$ to either $\hset_1'$ or $\hset_2'$ - whichever currently contains fewer clusters. Otherwise, we add both $H$ and $H'$ to one of the two sets $\hset_1',\hset_2'$, that contains fewer clusters.
The ties are broken arbitrarily. Let $(\hset'_1,\hset'_2)$ be the resulting partition of $\hset'$.

For each cluster $H\in \hset'_1$, we remove the unique terminal $t\in \tset(H)\cap \ttset_2'$ from $\ttset_2'$. For each cluster $H\in \hset'_2$, let $t$ be the unique terminal in $\tset(H)\cap\ttset_1'$, and let $s\in \ttset_1'$ such that $(s,t)\in \mset_1$. We remove $(s,t)$ from $\mset_1$, and we remove both $s$ and $t$ from $\ttset_1'$.
Notice that since $|\hset_1'|\geq |\hset_2'|-2$, we only remove a constant fraction of the terminals from $\ttset_1'$. Therefore, at the end of this procedure, we obtain a subset $\tmset\sse \mset$ of pairs of terminals, with $|\tmset|=\Omega(\alpha k_1)$, a set $\ttset_1$ of terminals participating in the pairs in $\tmset$, and another set $\ttset_2\sse \tset_2$ of terminals, with $|\ttset_2|\geq \Omega(\alpha k_2)$, such that $\ttset_1\cap \ttset_2=\emptyset$, and for every cluster $H\in \hset$, $\tset(H)$ contains at most one terminal $t\in \ttset_1\cup \ttset_2$. 
\end{proof}

From now on, we denote by $\tmset\sse \mset$ the collection of terminal pairs contained in $\ttset_1$. For each cluster $H\in \hset$, we say that the unique terminal $t\in \tset(H)\cap (\ttset_1\cup \ttset_2)$ is the \emph{representative} of $H$. Notice that some clusters $H\in \hset$ may have no representative.

\paragraph{Step 2: Tagging}
The tagging is performed exactly as before. Consider some cluster $H\in \hset$, that has a representative $t\in \tset(H)\cap (\ttset_1\cup\ttset_2)$. If $t$ is not a pseudo-center of $H$, then we can find an edge $e_t=(x_t,y_t)$ with $x_t\in V(H),y_t\not\in V(H)$, guaranteed by Theorem~\ref{thm: tagging}. Let $H'\in \hset$ be a cluster containing the vertex $y_t$ (again, such a cluster must exist since every vertex belongs to some cluster by Theorem~\ref{thm: tagging}). We say that cluster $H$ and terminal $t$ \emph{tag} the cluster $H'$.

We say that a pair $(s,t)\in \tmset$ is good if $H(s)$ tags $H(t)$ or vice versa. Let $\tmset'\sse \tmset$ be the set of all good pairs, and let $\tmset''=\tmset\setminus\tmset'$. If $|\tmset'|\geq |\tmset|/2$, then we can route all pairs in $\tmset'$ via edge-disjoint paths, where each pair $(s,t)\in \tmset'$ is routed inside the connected component $H(s)\cup H(t)\cup e_s$, if $H(s)$ tagged $H(t)$, or inside the component $H(s)\cup H(t)\cup e_t$ if $t$ tagged $s$. It is easy to see that all resulting paths are edge-disjoint. Therefore, we assume from now on that $|\tmset''|\geq |\tmset|/2$. To simplify notation, we will refer to $\tmset''$ as $\tmset$, and we discard from $\ttset_1$ terminals that do not belong to the new set $\tmset$.

We say that the terminals of $\ttset_1$ are red and the terminals of $\ttset_2$ are blue. If $H\in \hset$ has a representative terminal in $\ttset_1$, then we say that $H$ is a red cluster; if it has a blue representative, then it is a blue cluster. Notice that it is possible that $H$ has no representative, in which case it has no color.

As before, we would like to find a subset $\hset^*\sse \hset$ of clusters, such that whenever $H$ tags $H'$, only one of these two clusters belongs to $\hset^*$. However, we have an additional restriction: for every pair $(s,t)\in \tmset$,  we would like to ensure that either both $H(s)$ and $H(t)$ belong to $\hset^*$, or none of them.

In order to achieve this, we construct a directed graph $D$. The set of vertices of $D$ consists of two subsets: the set of vertices representing the terminals in $\ttset_2$, $V_2=\set{v_t\mid t\in \ttset_2}$, and the set of vertices representing the pairs of terminals in $\tmset$, $V_1=\set{v_{s,t}\mid (s,t)\in \tmset}$.

We say that a vertex $v_t\in V_1$ represents the cluster $H(t)$, and a vertex $v_{s,t}\in V_2$ represents the clusters $H(s)$ and $H(t)$. We let $V(D)=V_1\cup V_2$, and we add an edge $(u,u')$ to $D$ iff at least one of the clusters represented by $u$ tags a cluster represented by $u'$. We denote $|V_1|=\tk_1$ and $|V_2|=\tk_2$. Notice that $\tk_1=\Omega(\alpha k_1)$ and $\tk_2=\Omega(\alpha k_2)$. We use the following variation of Claim~\ref{claim: independent set}.

\begin{claim}
There is an independent set $\iset$ in $D$, containing at least $\lfloor \tk_1/24\rfloor$ vertices of $V_1$ and at least $\lfloor \tk_2/24\rfloor$ vertices of $V_2$.
\end{claim}

\begin{proof}
For simplicity, we treat $D$ as a general directed graph, where the out-degree of every vertex is at most $2$. Therefore, we can assume w.l.o.g. that $\tk_1\leq \tk_2$, and we denote $\beta=\tk_2/\tk_1$
We say that a vertex $v\in V_1$ is good iff it is adjacent to at most $8\beta$
vertices of $V_2$. Since the total number of vertices in the graph is bounded by $2\beta k_1$, the number of edges is bounded by $4\beta k_1$, and so at least half the vertices of $V_1$ are good.

Let $D'$ be the sub-graph of $D$ induced by the good vertices in $V_1$. Since the out-degree of every vertex in $D'$ is at most $2$, and $|V(D)|\geq \tk_1/2$, we can find an independent set $\iset_1$ in $D'$ of size exactly $\lfloor \tk_1/24\rfloor$, using the standard greedy algorithm as before.

Let $V'\sse V_2$ be the subset of vertices adjacent to the vertices of $\iset_1$. Since all vertices in $\iset_1$ are good, $|V'|\leq 8\beta \lfloor\frac{\tk_1}{24}\rfloor\leq \frac{\tk_2}{3}$. Let $D''$ be the sub-graph of $D$ induced by the vertices of $V_2\setminus V'$. Since the out-degree of every vertex in $D''$ is at most $2$, and $|V(D'')|\geq |V_2|/2$, we can find an independent set $\iset_2$ of size $\lfloor \tk_2/24\rfloor$ in $D''$. We output $\iset_1\cup \iset_2$ as our final solution.
\end{proof}

Let $\mset^*\sse \tmset$ denote the set of pairs $(s,t)$ with $v_{s,t}\in \iset$, let $\tset_1^*\sse \ttset_1$
be the set of all terminals participating in pairs in $\mset^*$, and let $\tset_2^*$ be the set of terminals $t$ with $v_t\in \iset$. Finally, let $\hset^*=\set{H(t)\mid t\in \tset^*_1\cup \tset^*_2}$. Then for every cluster $H\in\hset^*$, $\tset(H)$ has exactly one representative in $\tset_1^*\cup\tset_2^*$, and for any pair $H,H'$ of clusters, where $H$ tags $H'$, only one of the two clusters may belong to $\hset^*$. Notice that we are also guaranteed that $|\tset_1^*|=\Omega(\tk_1)=\Omega(\alpha k_1)$, and $|\tset_2^*|=\Omega(\tk_2)=\Omega(\alpha k_2)$.

\paragraph{Step 3: Merging the Clusters}
In this step, we create the final set $\gset$ of clusters, by merging some pairs of clusters in $\hset^*$. For each cluster $H\in \gset$, we first select up to two terminals $t,t'\in \tset(H)\cap (\tset_1^*\cup \tset_2^*)$ that we call \emph{potential pseudo-centers} for $H$. In the end, at most one such terminal will be selected for every cluster.

If $H\in \hset^*$ is a cluster that does not tag any other cluster, then we add $H$ to $\gset$. Let $t\in \tset(H)$ be the unique representative of $H$ in $\tset_1^*\cup \tset_2^*$. Then we say that $t$ is the potential pseudo-center of $H$. Notice that from Theorem~\ref{thm: tagging}, $t$ is indeed a pseudo-center of $H$.

Consider now some cluster $H\in \hset$, and let $\gset(H)\sse \hset^*$ be the set of clusters $H'$ that tag $H$. Assume that $\gset(H)=\set{H_1,\ldots,H_p}$. For each $1\leq i\leq p$, let $t_i\in \tset(H_i)$ be the representative of $H_i$ in $\tset_1^*\cup \tset_2^*$.  If $p=1$, then we create one new cluster $H^*=H\cup H_1\cup\set{e_{t_1}}$ and add it to $\gset$. We set $\tset(H^*)=\tset(H_1)\cup \tset(H)$, and we let $t_1$ be its potential pseudo-center. Notice that $t_1$ is indeed a pseudo-center of the new cluster, from Theorem~\ref{thm: tagging}.

Assume now that $p>1$, and it is even. As before, we use Lemma~\ref{lemma: pairing} to find a collection $\pset$ of edge-disjoint paths contained in $H$, where every path connects a distinct pair $y_{t_i},y_{t_j}$ of vertices. For each such path $P\in \pset$, if $P$ connects $y_{t_i}$ to $y_{t_j}$, then we construct a new cluster $H^*=H_i\cup P\cup H_j$, and add it to $\gset$. We set $\tset(H^*)=\tset(H_i)\cup \tset(H_j)$. Notice that from Theorem~\ref{thm: tagging}, both $t_i$ and $t_j$ are pseudo-centers of the new cluster. We say that $t_i$ and $t_j$ are \emph{potential pseudo-centers} for the new cluster, and we will later select at most one of them to represent this cluster.

Finally, assume that $p$ is odd. If the number of blue clusters in $\gset(H)$ is at least two, then we discard one arbitrary blue cluster from $\gset(H)$, and continue as in the case where $p$ is even. Otherwise, $\gset(H)$ must contain at least two red clusters. In this case we say that $H$ is an \emph{odd cluster}.

For each such odd cluster $H$, we would like to select one red cluster $H'\in \gset(H)$ to discard. However, once such a cluster is discarded, if $s\in \tset(H')$ is the representative of $H'$ in $\tset_1^*\cup \tset_2^*$, and $t$ is the terminal with $(s,t)\in \mset^*$, we will need to discard the terminal $t$ as well. That is, $t$ will not be able to serve as a pseudo-center of its cluster $H(t)$. In particular, if $H(t)\in \gset(H')$ of some other odd cluster $H'$, we need to ensure that $H(t)$ is the cluster that we discard from $\gset(H')$ in this case.

We proceed in two stages. In the first stage, while there is a pair $(s,t)\in \mset^*$, such that $H(s)\in \gset(H)$, $H(t)\in\gset(H')$, with $H\neq H'$, and both $H$ and $H'$ are odd clusters, we discard $H(s)$ from $\gset(H)$, $H(t)$ from $\gset(H')$, and $(s,t)$ from $\mset^*$. The clusters $H$ and $H'$ now stop being odd, as $|\gset(H)|$, $|\gset(H')|$ are now even. We process these two clusters as in the case where $p$ is even. When no such pair $(s,t)\in \mset^*$ remains, the first stage finishes and the second stage begins. We process the remaining odd clusters one-by-one. For each such cluster $H$, there must be at least one red cluster $H'\in \gset(H)$, such that, if $t\in \tset(H')$ is the representative of $H'$ in $\tset^*_1$, and $s$ is the terminal with $(s,t)\in \mset^*$, then $H(s)$ does not belong to the set $\gset(H'')$ of any odd cluster $H''$. We then discard $(s,t)$ from $\mset^*$, and discard $H'$ from $\gset(H)$. If $s$ is a potential pseudo-center for any cluster in $\gset$, then it stops being a potential pseudo-center for that cluster.

Let $\mset^{**}\sse \mset^*$ be the set of all surviving pairs. We need the following simple claim.

\begin{claim}
$|\mset^{**}|\geq |\mset^*|/4$.
\end{claim}

\begin{proof}
Consider the first stage of the algorithm, where we have removed a collection of pairs $(s,t)$, where $H(s)\in\gset(H)$, $H(t)\in\gset(H')$, for odd clusters $H$ and $H'$. For each such odd cluster $H^*$, we have removed at most one such pair. However, since $\gset(H^*)$ contained at least two red clusters, there is another pair $(s',t')\in \mset^*$, such that $H(t')\in \gset(H^*)$, and $(s',t')$ was not removed from $\mset^*$ during the first step. It is easy to verify that at most half the pairs in $\mset^*$ were removed during the first stage. Let $\mset'$ be the collection of pairs that survive after the first stage.

In the second stage, for each remaining odd cluster $H$, we select one red cluster $H'\in \gset(H)$, and delete the pair $(s,t)$ with $H(s)=H'$ from $\mset'$. However, since $H$ still remains an odd cluster, we are guaranteed that there is another pair $(s',t')\in \mset'$, such that $H(s')\in \gset(H)$, and $H(t')$ does not belong to any current odd cluster. (Otherwise we could have continued the first phase for another step). We can therefore charge the pair $(s,t)$ to the pair $(s',t')$, that will not be removed during the second step. So overall, at most half the pairs are removed from $\mset'$ during the second stage.
\end{proof}

We conclude that $|\mset^{**}|=\Omega(\alpha k_1)$. Let $\tset_1^{**}$ be the set of all red terminals participating in the pairs in $\mset^{**}$, and let $\tset_2^{**}$ be the set of all blue terminals that serve as potential pseudo-centers for clusters in $\gset$. It is easy to see that $|\tset_2^{**}|\geq |\tset_2^*|/2=\Omega(\alpha k_2)$, since every blue cluster discarded from $\hset^*$ can be charged to another blue cluster that remains in $\hset^*$. 

To summarize, we have defined a collection $\hset^*$ of edge-disjoint clusters, a set $\mset^{**}\sse \mset$ of $\Omega(\alpha k_1)$ pairs of terminals, a subset $\tset_1^{**}\sse \tset_1$ of terminals that participate in the pairs in $\mset^{**}$, and a subset $\tset_2^{**}\sse \tset_2$ of $\Omega(\alpha k_2)$ terminals, with $\tset_1^{**}\cap \tset_2^{**}=\emptyset$. Moreover, every terminal $t\in \tset_1^{**}\cup \tset_2^{**}$ is a potential pseudo-center for one of the clusters in $\gset$, and every
 cluster $H\in \gset$ now has either $0$, $1$, or $2$ potential pseudo-centers. The sets $\set{\tset(H)}_{H\in \gset}$ of terminals are mutually disjoint.

Notice that it is possible that for some clusters $H\in \gset$, we have two potential pseudo-centers $s,t$ for $H$, which form a pair in $\mset^{**}$. Let $\mset^{**}_1\sse \mset^{**}$ be the subset of all such pairs of terminals. In other words, $(s,t)\in \mset^{**}_1$ iff both $s$ and $t$ are potential pseudo-centers of the same cluster  $H\in \gset$. If $|\mset^{**}_1|\geq |\mset^{**}|/2$, then every pair in $\mset^{**}_1$ can be routed inside its own cluster, via edge-disjoint paths. In this case we terminate the algorithm and return this collection of paths. From now on we assume that this is not the case. To simplify notation, we denote by $\mset^{**}$ the set of pairs in $\mset^{**}\setminus \mset^{**}_1$, and by $\tset^{**}_1$ the set of terminals participating in pairs in $\mset^{**}$.

\paragraph{Step 4: Selecting the Pseudo-Centers}
In this step we  select a subset $\mset'\sse \mset^{**}$ of terminal pairs, a set $\tset'_1\sse \tset_1^{**}$ of terminals participating in pairs in $\mset'$, and a subset $\tset'_2\sse \tset_2^{**}$, such that for each cluster $H\in \gset$, at most one potential pseudo-center belongs to $\tset_1'\cup \tset_2'$. 

Our first step is to ensure that every cluster in $\gset$ contains at most one red potential pseudo-center. In order to do so, we start with $\mset'=\emptyset$. While $\mset^{**}\neq \emptyset$, we select an arbitrary pair $(s,t)\in \mset^{**}$, remove it from $\mset^{**}$ and add it to $\mset'$. Assume that $s$ is a potential pseudo-center of $H$, and $t$ is a potential pseudo-center of $H'$. We also remove from $\mset^{**}$ any other pair $(s',t')$, where either $s'$ or $t'$ are potential pseudo-centers of $H$ or $H'$. Notice that for each pair added to $\mset'$, at most three pairs are deleted from $\mset^{**}$. Therefore, $|\mset'|\geq |\mset^{**}|/3=\Omega(\alpha k_1)$. Let $\tset_1'$ be the set of all terminals participating in pairs in $\mset^{**}$.

Next, we ensure that every cluster in $\gset$ contains at most one blue potential pseudo-center. In order to do so, we start with $\tset_2'=\emptyset$. For every cluster $H\in \gset$ that contains two blue potential pseudo-center, we select one of these two blue terminals arbitrarily and add it to $\tset_2'$. Notice that $|\tset_2'|\geq |\tset_2^{**}|/2=\Omega(\alpha k_2)$.

Finally, we need to take care of clusters $H\in \gset$ that contain one red and one blue potential pseudo-centers. We call such clusters \emph{mixed clusters}, and we take care of them like in Step 1. Let $\gset'\sse \gset$ be the set of all mixed clusters.
We partition $\gset'$ into two subsets $\gset'_1,\gset'_2$, such that $|\gset'_1|,|\gset'_2|\geq |\gset'|/2-1$. 
As before, we ensure that if, for some pair $(s,t)\in \mset'$, there are two clusters $H,H'\in \gset'$, where $s$ is a potential pseudo-center of $H$, and $t$ is a potential pseudo-center of $H'$, then either both $H,H'\in \gset_1'$, or both $H,H'\in \gset_2'$. The partition is performed exactly as in Step 1.

For each cluster $H\in \gset'_1$, we remove its blue potential pseudo-center from $\tset_2'$. For each cluster $H\in \gset'_2$, let $t$ be its red potential pseudo-center, and let $s\in \tset_1'$ such that $(s,t)\in \mset'$. We remove $(s,t)$ from $\mset'$, and we remove both $s$ and $t$ from $\tset_1'$.
Since $|\gset_1'|\geq |\gset_2'|-2$, we only remove a constant fraction of the terminals from $\tset_1'$. Therefore, at most a constant fraction of the pairs is removed from $\mset'$, and at most a constant fraction of  terminals is removed from $\tset_2'$.

Let $\mset',\tset_1',\tset_2'$ be the resulting subsets of terminal pairs and terminals. Then $|\tset_1'|=\Omega(\alpha k_1)$, $|\tset_2'|=\Omega(\alpha k_2)$. For every cluster $H\in \gset$, at most one potential pseudo-center for $H$ belongs to $\tset_1'\cup \tset_2'$. Using Claim~\ref{claim: 1-edge-wl}, graph $G$ is $1$-well-linked for the set $\tset'=\tset_1'\cup \tset_2'$ of terminals.

\section{Proof of Theorem~\ref{thm: main}}\label{sec: complete proof of main thm} 
In this section we complete the proof of Theorem~\ref{thm: main}, using Theorem~\ref{thm: main: find good crossbar or find a routing}.

As in the previous work on the \EDPwC problem~\cite{RaoZhou,Andrews,EDP-old}, we follow the outline of~\cite{CKS}, by first partitioning the graph $G$ into a collection of sub-instances, where each sub-instance is well-linked for the corresponding set of the terminals. We then solve each such sub-instance separately, by embedding an expander into it. 
The specific type of embedding that we use is similar to the one proposed in~\cite{EDP-old}, and is summarized in the following definition.

\begin{definition}
Let $\mset'\sse \mset$ be a subset of the demand pairs, and let $\tset'=\tset(\mset')$ be the set of terminals participating in the pairs in $\mset'$. Let $H$ be any expander whose vertex set is $\tset'$. An \emph{embedding} of $H$ into $G$ maps every vertex $t\in V(H)$ into a connected component $C_t\sse G$ with $t\in C_t$, and every edge $e=(t,t')\in E(H)$ into a path $P_e$ in graph $G$, connecting some vertex $v\in C_t$ to some vertex $v'\in C_{t'}$. Given an edge $e'\in E(G)$, the \emph{load} on edge $e'$ is the total number of components $\set{C_t}_{t\in \tset'}$ and paths $\set{P_e}_{e\in E(H)}$ to which edge $e'$ belongs. The \emph{congestion} of the embedding is the maximum load on any edge $e'\in E(G)$.
\end{definition}

Our algorithm uses the cut-matching game of \cite{KRV}, summarized in Theorem~\ref{thm: CMG}, to embed an expander into the graph $G$. To simplify notation, we say that an $N$-vertex (multi)-graph $X$ is a \emph{good expander}, iff $N$ is even, the expansion of $X$ is at least $\alphaCMG(N)$, and 
the degree of every vertex of $X$ is exactly $\gcmg(N)$. Once we embed a good expander into the graph $G$, we will need to route demand pairs across $X$ via vertex-disjoint paths.

There are many algorithms for routing on expanders, e.g. \cite{LR,BFU,BFSU,KleinbergR,Frieze}, that give different types of guarantees. We use the following theorem, due to Rao and Zhou~\cite{RaoZhou} (see also a proof in~\cite{EDP-old}).

\begin{theorem}[Theorem 7.1 in~\cite{RaoZhou}]\label{thm: vertex-disjoint routing on expanders}
Let $G=(V,E)$ be any $n$-vertex $d$-regular $\alpha$-expander.
Assume further that $n$ is even, and that the vertices of $G$ are partitioned into $n/2$ disjoint demand pairs $\mset=\set{(s_1,t_1),\ldots,(s_{n/2},t_{n/2})}$. Then there is an efficient algorithm that routes $\Omega\left (\frac {\alpha n} {\log n\cdot d^2}\right )$ of the demand pairs on vertex-disjoint paths in $G$.
\end{theorem}

\begin{corollary}
Let $X$ be a good $N$-vertex expander, and assume that the vertices of $X$ are partitioned into $N/2$ disjoint demand pairs $\mset=\set{(s_1,t_1),\ldots,(s_{N/2},t_{N/2})}$. Then there is an efficient algorithm that routes $
\Omega\left (\frac {N\cdot \alphaCMG(N)} {\log N\cdot \gkrv^2(N)}\right )=\Omega\left (\frac{N}{\log^4N}\right )$ of the demand pairs on vertex-disjoint paths in $X$.
\end{corollary}

The following theorem allows us to construct a good expander $H$ and embed it into $G$ with congestion at most $2$, given a good crossbar $(\sset^*,\mset^*,\tau^*)$ in graph $G$.

\begin{theorem}\label{thm: main: embed an expander or find a routing}
Assume that we are given an undirected graph $G=(V,E)$ with vertex degrees at most $4$, and a set $\mset$ of $k$ demand pairs, defined over a set $\tset$ of terminals. Assume further that the degree of every terminal is $1$, every terminal participates in exactly one pair in $\mset$, and $G$ is $1$-well-linked for $\tset$. Then there is an efficient randomized algorithm, that with high probability outputs one of the following:

\begin{itemize}
\item Either a subset $\mset'\sse \mset$ of $k/\poly\log k$ demand pairs and the routing of the pairs in $\mset'$ with congestion at most $2$ in $G$;

\item Or a good expander $H$ together with the embedding of $H$ into $G$ with congestion at most $2$.
\end{itemize}
\end{theorem}
\begin{proof}
We apply Theorem~\ref{thm: main: find good crossbar or find a routing} to graph $G$. If the outcome is a subset $\mset'\sse \mset$ of $k/\poly\log k$ demand pairs and their routing in $\mset'$ with congestion at most $2$, then we return this routing, and terminate the algorithm. We assume from now on that the algorithm in Theorem~\ref{thm: main: find good crossbar or find a routing} produces a good crossbar $(\sset^*,\mset^*,\tau^*)$.

The good expander $H$ is constructed as follows. The set of vertices of $H$, $V(H)=\tset(\mset^*)$. The embedding of each vertex $t\in V(H)$ is the unique tree $T\in \tau^*$, to which terminal $t$ belongs.
In order to compute the set of edges of $H$ and their embedding into $G$, we perform the cut-matching game, using Theorem~\ref{thm: CMG}. Recall that this game consists of $\gkrv$ iterations, where in iteration $j$, for $1\leq j\leq \gkrv$, the cut player computes a partition $(A_j,B_j)$ of $V(H)$, and the matching player responds with a perfect matching between the vertices of $A_j$ and the vertices of $B_j$. Given the partition $(A_j,B_j)$ of $V(H)$, computed by the cut player, we find a corresponding partition $(A'_j,B_j')$ of $\Gamma^*_j$, as follows. For each $t_i\in V(H)$, if $t_i\in A_j$, then we add $e_{i,j}$ to $A'_j$, and otherwise we add it to $B_j'$. Since the set $S^*_j$ is $1$-well-linked for $\Gamma^*_j$, we can find a collection $\qset^*_j: A'_j\sconnect_1 B'_j$ of edge-disjoint paths contained in $S^*_j$. These paths define a matching $M_j$ between the edges of $A'_j$ and the edges of $B'_j$, which in turn defines a matching between the terminals of $A_j$ and the terminals of $B_j$. We then add the edges of the matching to the graph $H$. For each such edge $e\in M_j$, its embedding into $G$ is the corresponding path $Q\in \qset^*_j$. From Theorem~\ref{thm: CMG}, after $\gkrv$ iterations, we obtain a good expander $H$, together with its embedding into $G$. It is easy to see that the embedding causes congestion at most $2$ in graph $G$. Indeed, consider any edge $e\in E$. If edge $e$ belongs to any sub-graph $G[S_j]$, for $S_j\in \sset^*$, then it belongs to at most one tree in $\tau^*$, and to at most one path in $\set{P_{e'}\mid e'\in E(H)}$. If $e$ does not belong to any such sub-graph, then it is contained in at most two trees in $\tau^*$, and it is not contained in any path in $\set{P_{e'}\mid e'\in E(H)}$. Therefore, the total congestion of this embedding is at most $2$.
\end{proof}

We are now ready to complete the proof of Theorem~\ref{thm: main}.
Our starting point is similar to that used in previous work on \EDP~\cite{ANF,CKS,RaoZhou,Andrews,EDP-old}. We use the standard multicommodity flow LP-relaxation for the \EDPwC problem to partition our graph into several disjoint sub-graphs, that are well-linked for their respective sets of terminals. In the standard LP-relaxation for \EDPwC, we have an indicator variable $x_i$ for each $1\leq i\leq k$, for whether the pair $(s_i,t_i)$ is routed. Let $\pset_i$ be the set of all paths connecting $s_i$ to $t_i$ in $G$. The LP relaxation is defined as follows.

\begin{eqnarray} 
\mbox(LP1)&\max \quad\sum_{i=1}^kx_i& \nonumber\\
\mbox{s.t.}
&\sum_{P\in \pset_i}f(P)\geq x_i&\forall 1\leq i\leq k\\
&\sum_{P: e\in P}f(P)\leq 2&\forall e\in E \label{eq: capacity constraint}\\
&0\leq x_i\leq 1&\forall 1\leq i\leq k\\
&f(P)\geq 0&\forall 1\leq i\leq k,\forall P\in \pset_i
\end{eqnarray}

While this LP has exponentially many variables, it can be efficiently solved using standard techniques, e.g. by using an equivalent polynomial-size LP formulation. We denote by $\opt$ the value of the optimal solution to the LP. Clearly, the value of the optimal solution to the \EDPwC problem instance with congestion $2$ is at most $\opt$.

Notice that if we replace Constraint~(\ref{eq: capacity constraint}) with the following constraint:

\[\sum_{P: e\in P}f(P)\leq 1\quad\quad \forall e\in E, \]

then we obtain the standard multicommodity LP relaxation for the \EDP problem itself, where no congestion is allowed. Let $\opt'$ denote the optimal solution value of this new LP. Then clearly $\opt'\leq \opt\leq 2\opt'$.

The next theorem follows from the work of Chekuri, Khanna and Shepherd~\cite{ANF,CKS}, and we provide a short proof sketch for completeness.

\begin{theorem}\label{thm: starting point}
Suppose we are given a graph $G=(V,E)$ and a set $\mset$ of $k$ source-sink pairs in $G$. Then we can efficiently partition $G$ into a collection $G_1,\ldots,G_{\ell}$ of vertex-disjoint induced sub-graphs, and compute, for each $1\leq i\leq \ell$, a collection $\mset_i\sse \mset$ of source-sink pairs contained in $G_i$, such that $\sum_{i=1}^{\ell}|\mset_i|=\Omega(\opt/\log^2k)$, and moreover, if for each $1\leq i\leq \ell$, $\tset_i$ denotes the set of terminals participating in pairs in $\mset_i$, then $\tset_i\sse V(G_i)$, and $G_i$ is $1$-well-linked for $\tset_i$.
\end{theorem}

\begin{proof}
We need the following definition.

\begin{definition}
Given a graph $G=(V,E)$, and a subset $\tset\sse V$ of vertices called terminals, we say that $\tset$ is \emph{flow-well-linked} in $G$, iff any matching $\mset$ on $\tset$ can be fractionally routed with congestion at most $2$ in $G$.
\end{definition}

The next theorem follows from the work of Chekuri, Khanna and Shepherd~\cite{ANF,CKS}, and its proof also appears in~\cite{EDP-old}.


\begin{theorem}\label{thm: starting point2}
Suppose we are given a graph $G=(V,E)$ and a set $\mset$ of $k$ source-sink pairs in $G$. Then we can efficiently partition $G$ into a collection $G_1,\ldots,G_{\ell}$ of vertex-disjoint induced sub-graphs, and compute, for each $1\leq i\leq \ell$, a collection $\mset_i\sse \mset$ of source-sink pairs contained in $G_i$, such that $\sum_{i=1}^{\ell}|\mset_i|=\Omega(\opt/\log^2k)$, and moreover, if for each $1\leq i\leq \ell$, $\tset_i$ denotes the set of terminals participating in pairs in $\mset_i$, then $\tset_i\sse V(G_i)$, and $G_i$ is flow-well-linked for $\tset_i$.
\end{theorem}

Fix some $1\leq i\leq \ell$. From the min-cut max-flow theorem, since $G_i$ is flow-well-linked for $\tset_i$, 
then it must be $1/2$-well-linked for $\tset_i$. We apply the grouping technique from Theorem~\ref{thm: grouping-advanced} to graph $G_i$, where the first set of terminals is $\tset_i$, together with the matching $\mset_i$ defined over $\tset_i$, and the second set of terminals is $\emptyset$. Notice that $\alpha=1/2$ in this case. If the outcome is a subset $\mset'_i\sse \mset_i$ of $\Omega(|\mset_i|)$ pairs, together with their routing on edge-disjoint paths, then graph $G_i$ is $1$-well-linked for $\tset(\mset'_i)$. Otherwise, we obtain a subset $\mset'_i\sse \mset_i$ of $\Omega(|\mset_i|)$ pairs, where $G_i$ is $1$-well-linked for $\tset(\mset'_i)$. In any case, for each graph $G_i$, we now obtain a subset $\mset'_i\sse \mset_i$ of the demand pairs, such that $G_i$ is $1$-well-linked for $\tset(\mset'_i)$, and $\sum_{i=1}^{\ell}|\mset'_i|=\Omega(\opt/\log^2k)$.
\end{proof}

We now proceed to solve the problem on each one of the graphs $G_i$ separately. For each $i$, let $k_i=|\mset'_i|$.
Notice that in order to complete the proof of Theorem~\ref{thm: main}, it is now enough to show that there is an efficient randomized algorithm that for each $1\leq i\leq \ell$, routes $\Omega(k_i/\poly\log k_i)$ pairs in $\mset'_i$ with congestion at most $2$ in graph $G_i$ w.h.p.

Fix some $1\leq i\leq \ell$.
We apply Theorem~\ref{thm: main: embed an expander or find a routing} to graph $G_i$ and the set $\mset'_i$ of the demand pairs. If the output of Theorem~\ref{thm: main: embed an expander or find a routing} is a collection $\mset_i''\sse \mset'_i$ of $k_i/\poly\log k_i$ demand pairs, together with their routing with congestion at most $2$ in $G_i$, then we terminate the algorithm and return this routing for graph $G_i$.
Otherwise, we obtain a collection $\mset_i''\sse \mset'_i$ of $k_i/\poly\log k_i$ demand pairs, a good expander $H$ whose vertex set is $V(H)=\tset(\mset_i'')$, and the embedding of $H$ into the graph $G_i$ with congestion $2$. Using Theorem~\ref{thm: vertex-disjoint routing on expanders}, we can compute a collection $\tmset_i\sse \mset_i''$ of $\Omega\left (\frac{|\mset''_i|}{(\log |\mset''_i|)^4}\right )=\Omega\left (k_i/\poly\log k_i\right )$ demand pairs, and a routing $\pset$ of the pairs in $\tmset_i$ in the expander $H$ via vertex-disjoint paths. We now construct a routing of the pairs in $\tmset_i$ in the graph $G_i$, with congestion at most $2$.

In order to construct the routing, consider any path $P\in \pset$, and assume that it connects a pair $(s,t)\in \tmset_i$. We transform the path $P$ into a path $Q$, connecting $s$ to $t$ in graph $G_i$, as follows. Assume that $P=(t_0=s,t_1,\ldots,t_r=t)$, and let $e_j$ be the edge connecting $t_j$ to $t_{j+1}$, for $0\leq j<r$. For each such edge $e_j$, let $P_{e_j}$ be the path via which $e_j$ is embedded into $G$. Let $v_j$ be the first vertex on  $P_{e_j}$, and let $u_{j+1}$ be the last vertex on $P_{e_j}$. Then $v_j$ belongs to the connected component $C_{t_j}$, and $u_{j+1}$ belongs to the connected component $C_{t_{j+1}}$. 

For each $1\leq j<r$, let $Q_j$ be any path  connecting $u_j$ to $v_j$ inside the connected component $C(t_j)$. Let $Q_0$ be any path connecting $t_0$ to $v_0$ inside $C_{t_0}$, and let $Q_r$ be any path connecting $u_r$ to $t_r$ inside $C_{t_r}$. Path $Q$ is then obtained by concatenating $Q_0,P_{e_0},Q_1,\ldots,P_{r-1},Q_r$. Notice that path $Q$ connects the original pair $(s,t)$ of terminals to each other.

Let $\qset$ be the final set of paths, obtained after processing all paths $P\in\pset$. From the above discussion, every pair $(s,t)\in\tmset_i$ has a path connecting $s$ to $t$ in $\qset$. In order to bound the congestion caused by $\qset$, recall that the paths in $\pset$ are vertex-disjoint. Therefore, the paths in $\qset$ traverse every component $C_t$ for $t\in V(H)$ at most once, and use every path $P_e$ for $e\in E(H)$ at most once. Since the congestion of the embedding of $H$ is bounded by $2$, the paths in $\qset$ cause congestion at most $2$ overall. 

To conclude, for each $1\leq i\leq \ell$, we have computed a collection $\mset^*_i\sse \mset_i$ of $k_i/\poly\log k_i$ demand pairs, and a routing $\qset_i$ of the pairs in $\mset^*_i$ with congestion at most $2$ in $G_i$. Since $\sum_{i=1}^{\ell}k_i=\Omega(\opt/\log^2k)$, overall $\bigcup_{i=1}^{\ell}\qset_i$ is a $(\poly\log k)$-approximate solution for the \EDPwC instance $(G,\mset)$ with congestion $2$.


\section{Proof of Claim~\ref{claim: random partition into gamma sets}}\label{sec: proof of partitioning claim}

Let $H=G'\setminus\tset$.
Fix some $1\leq j\leq \gamma$. 
Let $\event_1(j)$ be the bad event that $\sum_{v\in X_j}d_{H}(v)> \frac{2m}{\gamma}\cdot \left (1+\frac 1 {\gamma}\right )$. In order to bound the probability of $\event_1(j)$, we define, for each vertex $v\in V(H)$, a random variable $x_v$, whose value is $\frac{d_{H}(v)}{k_1}$ if $v\in X_j$ and $0$ otherwise. Notice that $x_v\in [0,1]$, and the random variables $\set{x_v}_{v\in V(H)}$ are pairwise independent. Let $B=\sum_{v\in V(H)} x_v$. Then the expectation of $B$, $\mu_1=\sum_{v\in V(H)} \frac{d_{H}(v)}{\gamma k_1}=\frac{2m}{\gamma k_1}$. Using the standard Chernoff bound (see e.g. Theorem 1.1 in~\cite{measure-concentration}),

\[\prob{\event_1(j)}=\prob{B> \left (1+1/\gamma\right )\mu_1}\leq e^{-\mu_1/(3\gamma^2)}=e^{-\frac{2m}{3\gamma^3 k_1}}<\frac 1 {6\gamma}\]

since $m\geq k/6$ and $k_1=\frac{k}{192\gamma^3\log \gamma}$.

For each terminal $t\in \tset$, let $e_t$ be the unique edge adjacent to $t$ in graph $G'$, and let $u_t$ be its other endpoint. Let $U=\set{u_t\mid t\in \tset}$. For each vertex $u\in U$, let $w(u)$ be the number of terminals $t$, such that $u=u_t$. Notice that $w(u)\leq k_1$ must hold. We say that a bad event $\event_2(j)$ happens iff $\sum_{u\in U\cap X_j}w(u)\geq \frac k{\gamma}\cdot \left (1+\frac 1 {\gamma}\right )$. 
In order to bound the probability of the event $\event_2(j)$, we define, for each $u\in U$, a random variable $y_u$, whose value is $w(u)/k_1$ iff $u\in X_j$, and it is $0$ otherwise. Notice that $y_u\in [0,1]$, and the variables $y_u$ are independent for all $u\in U$. Let $Y=\sum_{u\in U} y_u$. The expectation of $Y$ is $\mu_2=\frac{k}{k_1\gamma}$, and event $\event_2(j)$ holds iff $Y\geq \frac{k}{k_1\gamma}\cdot  \left (1+\frac 1 {\gamma}\right )\geq \mu_2\cdot \left (1+\frac 1 {\gamma}\right )$. Using the standard Chernoff bound again, we get that:

\[\prob{\event_2(j)}\leq e^{-\mu_2/(3\gamma^2)}\leq e^{-k/(3k_1\gamma^3)}\leq \frac 1 {6\gamma}\]

since $k_1=\frac{k}{192\gamma^3\log \gamma}$. Notice that if events $\event_1(j),\event_2(j)$ do not hold, then:

\[|\out_{G'}(X_j)|\leq \sum_{v\in X_j}d_{H}(v)+\sum_{u\in U\cap X_j}w(u)\leq \left (1+\frac 1 {\gamma}\right )\left (\frac{2m}{\gamma}+\frac{k}{\gamma}\right )< \frac{10m}{\gamma}\]

since $m\geq k/6$.

Let $\event_3(j)$ be the bad event that $|E_{G'}(X_j)|< \frac{m}{2\gamma^2}$. We next prove that $\prob{\event_3(j)}\leq \frac 1 {6\gamma}$. We say that two edges $e,e'\in E(G'\setminus \tset)$ are \emph{independent} iff they do not share any endpoints. Our first step is to compute a partition $U_1,\ldots,U_r$ of the set $E(G'\setminus \tset)$ of edges, where $r\leq 2k_1$, such that for each $1\leq i\leq r$, $|U_i|\geq \frac m{4k_1}$, and all edges in set $U_i$ are mutually independent. In order to compute such a partition, we construct an auxiliary graph $Z$, whose vertex set is $\set{v_e\mid e\in E(H)}$, and there is an edge $(v_e,v_{e'})$ iff $e$ and $e'$ are not independent. Since the maximum vertex degree in $G'$ is at most $k_1$, the maximum vertex degree in $Z$ is bounded by $2k_1-2$. Using the Hajnal-Szemer\'edi Theorem~\cite{Hajnal-Szemeredi}, we can find a partition $V_1,\ldots,V_r$ of the vertices of $Z$ into $r\leq 2k_1$ subsets, where each subset $V_i$ is an independent set, and $|V_i|\geq \frac{|V(Z)|}{r}-1\geq \frac{m}{4k_1}$. The partition $V_1,\ldots,V_r$ of the vertices of $Z$ gives the desired partition $U_1,\ldots,U_r$ of the edges of $G'\setminus \tset$. For each $1\leq i\leq r$, we say that the bad event $\event_3^i(j)$ happens iff $|U_i\cap E(X_j)|< \frac{|U_i|}{2\gamma^2}$. Notice that if $\event_3(j)$ happens, then event $\event_3^i(j)$ must happen for some $1\leq i\leq r$. Fix some $1\leq i\leq r$. The expectation of $|U_i\cap E(X_j)|$ is $\mu_3=\frac{|U_i|}{\gamma^2}$. Since all edges in $U_i$ are independent, we can use the standard Chernoff bound to bound the probability of $\event_3^i(j)$, as follows:

\[\prob{\event_3^i(j)}=\prob{|U_i\cap E(X_j)|<\mu_3/2}\leq e^{-\mu_3/8}=e^{-\frac{|U_i|}{8\gamma^2}}\].

Since $|U_i|\geq \frac{m}{4k_1}$, $m\geq k/6$, $k_1=\frac{k}{192\gamma^3\log \gamma}$, and $\gamma=\Theta(\log^4k)$, this is bounded by $\frac{1}{12k_1\gamma}$. We conclude that $\prob{\event_3^i(j)}\leq \frac{1}{12k_1\gamma}$, and by using the union bound over all $1\leq i\leq r$, $\prob{\event_3(j)}\leq \frac{1}{6\gamma}$.

 Using the union bound over all $1\leq j\leq \gamma$, with probability at least $\half$, none of the events $\event_1(j),\event_2(j),\event_3(j)$ for $1\leq j\leq \gamma$ happen, and so for each $1\leq j\leq \gamma$, $|\out_{G'}(X_j)|< \frac{10m}{\gamma}$, and $|E_{G'}(X_j)|\geq\frac{m}{2\gamma^2}$ must hold.

\section{Integrality Gap of the Multi-Commodity Flow Relaxation for EDP}\label{sec: gap}

Recall that in the multcommodity flow relaxation of the \EDP problem, the goal is to maximize the total amount of flow sent between the demand pairs, with no congestion, with at most one flow unit being sent between any pair. The gap construction for this relaxation, due to Garg, Vazirani and Yannakakis~\cite{trees2} starts from a sub-graph of the $(k+1)\times (k+1)$ grid depicted in Figure~\ref{fig: grid example}. Every degree-$4$ vertex in this graph is then replaced by the gadget depicted in the figure. This gadget ensures that whenever two paths traverse a degree-$4$
vertex of the grid, one vertically and one horizontally, then they must share an edge.

\begin{figure}[h]
\begin{center}
\scalebox{0.4}{\includegraphics{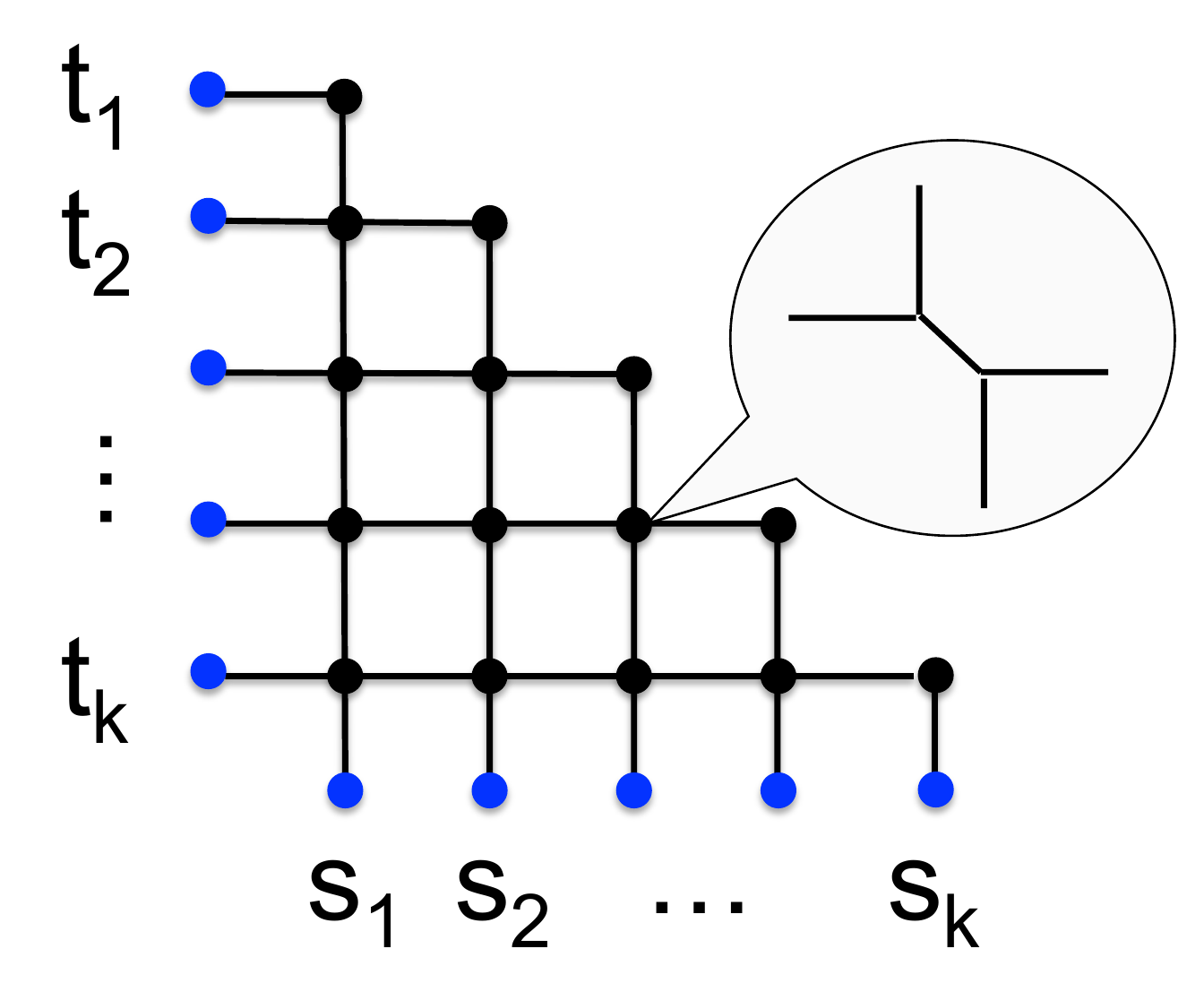}} \caption{\label{fig: grid example} The gap example for the multicommodity flow relaxation}
\end{center}
\end{figure}

Notice that every pair $(s_i,t_i)$ can send $\half$ flow unit from $s_i$ to $t_i$ along the path that contains the vertical segment of the grid incident to $s_i$ and the horizontal segment incident to $t_i$. Therefore, the value of the multicommidity flow relaxation is $k/2=\Omega(\sqrt n)$. However, it is easy to see that the value of the optimal integral solution is $1$. Indeed, for any path $P$ connecting any pair $(s_i,t_i)$, if we remove all edges of $P$ from the graph, then all other demand pairs become disconnected. Therefore, the integrality gap of this example is $\Omega(\sqrt n)$.

Notice also that in this example we can route $k=\Omega(\sqrt n)$ demand pairs integrally with congestion $2$, while at most one demand pair can be routed with congestion $1$. This shows a polynomial gap between routing with congestion $1$ and routing with congestion $2$ and higher.

Figure~\ref{fig: brick-wall graph} depicts a brick-wall graph. Obtaining sub-polynomial approximation algorithms for \EDP on brick-wall graphs remains an interesting open problem.

\begin{figure}[h]
\begin{center}
\scalebox{0.4}{\includegraphics{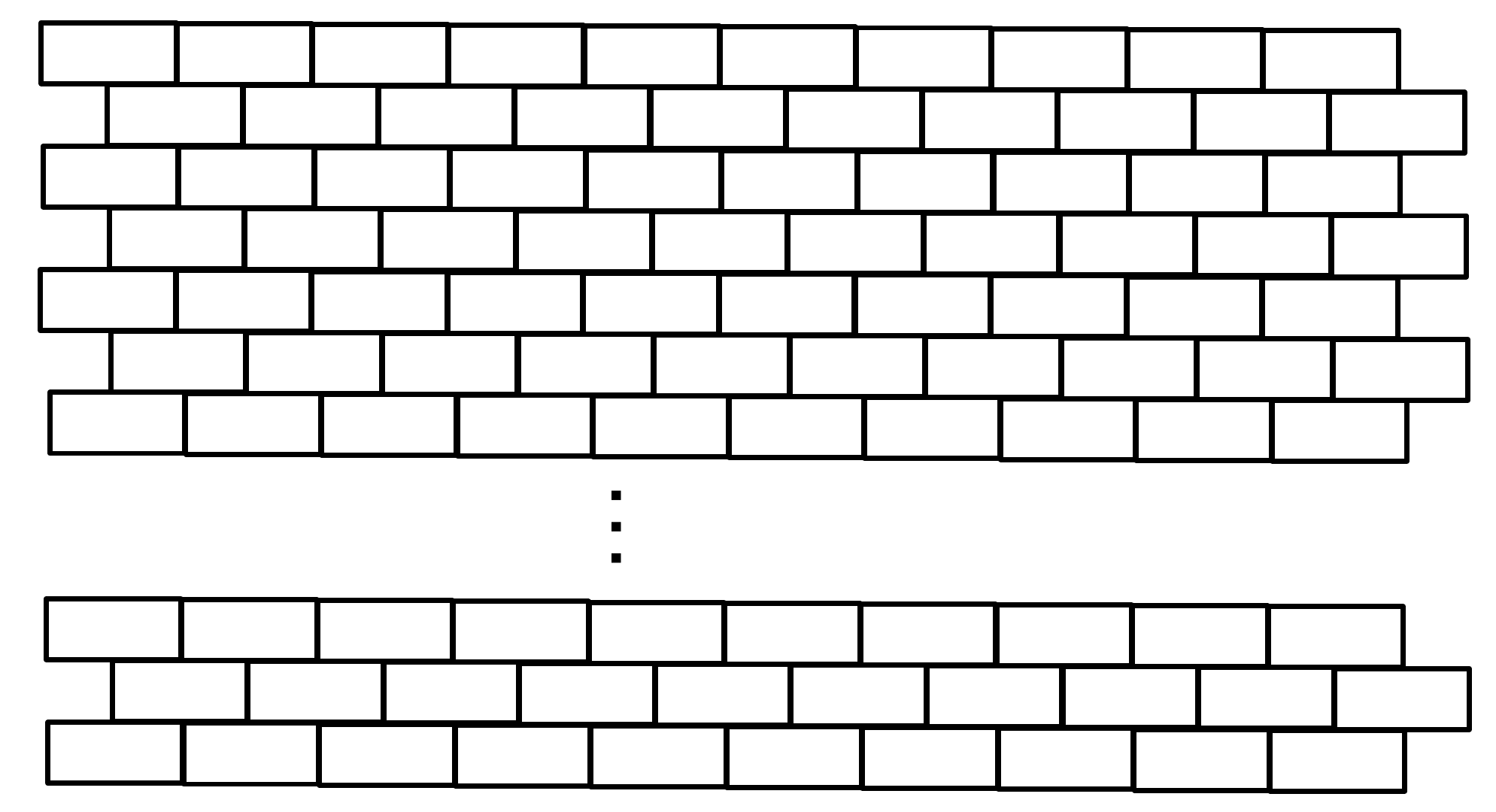}} \caption{\label{fig: brick-wall graph} A brick-wall graph}
\end{center}
\end{figure}

\section{Table of Parameters}\label{sec: appendix-params}
\renewcommand{\arraystretch}{1.4}
\begin{tabular}{|l|l|p{10cm}|} \hline
$\beta(k)$&$\Theta(\log k)$&Flow-cut gap in undirected graphs.\\ \hline
$\gkrv$&$\Theta(\log^2k)$&Parameter for the cut-matching game, from Theorem~\ref{thm: CMG}. \\ \hline
$\gamma$&$2^{24}\gkrv^4=O(\log^8k)$&  Number of sets in the initial partition, $\rset=\set{S_1,\ldots,S_{\gamma}}$ \\ \hline
$\gamma'$&$\gamma^{1/4}=2^6\gkrv$& Threshold for the number of leaves for Cases 1 and 2 \\ \hline
$\alphasc(k)$&$O(\sqrt{\log k})$& Approximation factor of the algorithm of~\cite{ARV} for Sparsest Cut.\\ \hline
$\alpha$&$\frac{1}{2^{11}\cdot \gamma\cdot\log k}=\Omega\left (\frac 1 {\log^9k}\right )$&Parameter for well-linked decomposition.\\ \hline
$\alphaWL$&$\alpha/\alphasc(k)=\Omega\left (\frac 1 {\log^{9.5}k}\right )$&Well-linkedness parameter.\\ \hline
$r$&$8\gkrv$& Number of sets in $\rset'$ \\ \hline
$k_1$&$\frac{k}{192\gamma^3\log \gamma}=\Omega \left (\frac k {\log^{24}k\log\log k}\right )$&Threshold for large clusters.\\ \hline
$k'$&$\floor{\frac{k_1}{6\gamma^2}}=\Omega \left (\frac k {\poly\log k}\right )$&Parameter in the construction of graph $Z$ in Step 1.\\ \hline
$k_2$&$\Omega\left (\frac{k_1\alpha\cdot\alphaWL}{\gamma^{3.5}}\right )=\Omega \left (\frac k {\poly\log k}\right )$&Number of paths for every edge in $\tT$ in Step 1.\\ \hline
$k_4$&$\Omega(k_2/r^2)=\Omega \left (\frac k {\poly\log k}\right )$&Number of demand pairs routed to $S,S'$ in Step 2.\\ \hline
$k_6$&$\Omega(\alphaWL^2k_4)=\Omega \left (\frac k {\poly\log k}\right )$&Number of trees in $\tau^*$ is at least $k_6/2$.\\ \hline
\end{tabular}

\end{document}